\documentclass{lmcs}
\pdfoutput=1

\usepackage{lastpage}
\lmcsdoi{15}{1}{16}
\lmcsheading{}{\pageref{LastPage}}{}{}%
{Jan.~14,~2018}{Feb.~22,~2019}{}

\usepackage{graphicx}
\usepackage{amssymb}
\usepackage{qtree}
\usepackage{extarrows}
\usepackage[subrefformat=parens]{subcaption}
\usepackage{stmaryrd,xspace,wrapfig}
\usepackage{xpatch}
\usepackage{mathtools}
\usepackage{verbatim}
\usepackage{mathrsfs}

\theoremstyle{plain}
\newtheorem{thmNN}{Theorem}

\setlist{  
  listparindent=\parindent,
}

\allowdisplaybreaks[1]

\newcommand\mline[1]{\begin{minipage}{0.55\textwidth}\ \\[-.5ex]#1\\[-.5ex]\end{minipage}}
\newcommand\obsolete[1]{}
\newcommand\Ta{\mathtt{a}}
\newcommand\Tb{\mathtt{b}}
\newcommand\Tc{\mathtt{c}}

\newcommand\numofvars{\mathtt{\#vars}}
\newcommand\SC{strongly-connected} 
\newcommand\COL{\mathbin{:}}
\newcommand\Hhole[2]{\hhole_{#1}^{#2}}

\newcommand\Ttwice[1]{\tau^{(#1)}}
\newcommand\swap[3]{[#2\leftrightarrow #3]#1}
\newcommand\affine{affine }
\newcommand\anaffine{an affine }
\newcommand\linear{linear }
\newcommand\ctyping{context type}
\newcommand\mvlinctr{S} 
\newcommand{\T}{\mathtt{o}}
\newcommand{\ar}{\mathtt{iar}}
\newcommand{\ord}{\mathtt{ord}}
\newcommand{\tar}[3]{\mathtt{iar}(#1 \p #2:#3)}
\newcommand{\tord}[3]{\mathtt{ord}(#1 \p #2:#3)}
\newcommand{\fv}{\mathbf{FV}}
 
\newcommand{\ft}{{\rightarrow}} 
\newcommand{\fa}{\rightarrow} 
\newcommand{\tends}{\rightarrow}
\newcommand{\uv}{{\ast}} 
\newcommand{\V}{\mathbf{V}} 
\newcommand{\hole}{[\,]}

\newcommand\esym{\emptyset}
\newcommand{\grm}{{\mathcal G}}
\newcommand{\lgrm}[3]{{\mathcal G}{(#1,#2,#3)}}
\newcommand{\egrm}[3]{{\mathcal G}^{\esym}{(#1,#2,#3)}}

\newcommand{\ralph}{\Sigma}
\newcommand{\lralph}[3]{\Sigma{(#1,#2,#3)}}

\newcommand{\nont}{{\mathcal N}}
\newcommand{\lnont}[3]{{\mathcal N}{(#1,#2,#3)}}
\newcommand{\enont}[3]{{\mathcal N}^{\esym}{(#1,#2,#3)}}

\newcommand{\rules}{{\mathcal R}}
\newcommand{\lrules}[3]{{\mathcal R}{(#1,#2,#3)}}
\newcommand{\erules}[3]{{\mathcal R}^{\esym}{(#1,#2,#3)}}

\newcommand{\rew}{\longrightarrow}

\newcommand{\var}{\overline{x}}
\newcommand{\defeq}{\triangleq}
\newcommand{\order}{\mathrm{O}}
\renewcommand{\(}{$}
\renewcommand{\)}{$}

\newcommand{\terms}[3]{\Lambda(#1,#2,#3)}
\newcommand{\termsg}[5]{\Lambda(\typing{#1}{#2},#3,#4,#5)}
\newcommand{\termsn}[3]{\Lambda_n(#1,#2,#3)}
\newcommand{\termsN}[4]{\Lambda_{#1}(#2,#3,#4)}

\newcommand{\termsng}[5]{\Lambda_n(\typing{#1}{#2},#3,#4,#5)}

\newcommand{\ctxN}[2]{{\mathcal S}_{#1}(#2)}
\newcommand{\ctxn}[1]{\ctxN{n}{#1}}
\newcommand{\ctx}[1]{{\mathcal S}(#1)}
\newcommand{\ctxinf}[1]{{\mathcal S}^{\mathrm{inf}}(#1)}

\newcommand{\octx}[1]{{\mathcal U}(#1)}
\newcommand{\octxnn}[2]{{\mathcal U}_{#1}(#2)}
\newcommand{\octxn}[1]{\octxnn{n}{#1}}
\newcommand{\types}[2]{\mathrm{Types}(#1,#2)}
\newcommand{\tr}{T}
\newcommand{\TR}[1]{{\mathcal T}(#1)}
\newcommand{\TRn}[2]{{\mathcal T}_{#1}(#2)}

\newcommand{\typing}[2]{\langle #1; #2 \rangle}

\newcommand{\NT}[2]{N_{\typing{#1}{#2}}}
\newcommand{\N}[1]{N_{#1}}

\newcommand{\inhav}[1]{{\mathcal L}\!\left(#1\right)}
\newcommand{\inhavn}[1]{{\mathcal L}_n\!\left(#1\right)}
\newcommand{\inhavnn}[2]{{\mathcal L}_{#1}\!\left(#2\right)}

\newcommand{\pinhav}[1]{\inhav{#1}}
\newcommand{\pinhavn}[1]{\inhavn{#1}}
\newcommand{\pinhavnn}[2]{\inhavnn{#1}{#2}}

\newcommand{\Exp}[2]{{\mathbf{exp}_{#1}(#2)}}

\newcommand{\nat}{\mathbb{N}}
\newcommand{\real}{\mathbb{R}}

\newcommand{\ie}{i.e.,\;}

\newcommand{\cf}{{\it cf.\;}}
\newcommand{\Sec}[1]{Section~\ref{sec:#1}}

\newcommand{\D}{\delta} 
\newcommand{\I}{\iota}  
\newcommand{\K}{\xi}    
\newcommand{\cards}{\#}
\newcommand{\card}[1]{\cards\!\left(#1\right)}
\newcommand{\floor}[1]{\lfloor #1 \rfloor}
\newcommand{\ceil}[1]{\lceil #1 \rceil}
\newcommand{\expl}[2]{\mathit{Expl}_{#1}^{#2}}
\newcommand{\ctwo}{\overline{2}}

\newcommand{\ctwok}[1]{\ctwo_{#1}}
\newcommand{\subc}{\preceq}

\newcommand{\subw}{\sqsubseteq}
\newcommand{\fold}[3]{#1\mathop{{\uparrow}^{#2}}#3}

\newcommand{\tec}{\pi}
\newcommand{\ee}{\emptyset}

\newcommand{\len}[1]{\card{#1}}

\newcommand{\reason}[1]{\text{($\because$ #1)}}

\newcommand{\image}[1]{\mathrm{Im}\left(#1\right)}
\newcommand{\dom}{\mathop{\mathrm{Dom}}}

\newcommand{\pct}{C}
\newcommand{\pU}{U}

\newcommand\closeddecomp{\Phi_m^{\bullet}}

\newcommand\closeddecompp[1]{\Phi_m^{\bullet,#1}} 
\newcommand\closeddecompm[1]{\Phi_{#1}^{\bullet}}

\newcommand{\closedgdecomp}{\widetilde{\Phi}_{m}^{\bullet,\grm}}
\newcommand{\closedgdecompm}[1]{\widetilde{\Phi}_{#1}^{\bullet,\grm}}

\newcommand{\gctx}{\widetilde{E}}

\newcommand{\nontinfp}[1]{{#1}^{\mathrm{inf}}}
\newcommand{\nontinf}{\nontinfp{\nont}}
\newcommand{\css}{\mathcal{U}}
\newcommand{\cszon}[2][n]{\css_{\le #1}({#2})}

\newcommand{\uframe}{\mathcal{E}_n^m}
\newcommand{\tframe}{\widetilde{\mathcal{E}}_n^m}
\newcommand{\tframeNM}[2]{\widetilde{\mathcal{E}}_{#1}^{#2}}
\newcommand{\useq}{\mathcal{P}_E^m}

\newcommand{\tseq}{\widetilde{\mathcal{P}}_{\gctx}^m}

\newcommand{\ucomp}[1][\hhole_k^n]{\mathcal{U}^m_{#1}}
\newcommand{\tcomp}[1][\hhole_{\kappa}^n]{\widetilde{\mathcal{U}}^m_{#1}(\grm)}
\newcommand{\tcompCM}[2]{\widetilde{\mathcal{U}}^{#2}_{#1}(\grm)}
\newcommand{\ucomph}[1]{\mathcal{U}^m_{#1}}

\newcommand\set[1]{\{#1\}}
\newcommand\p{\vdash} 
\newcommand\td{\Delta} 

\newcommand\aeq{\sim_{\alpha}}
\newcommand\hn[1]{\mathtt{hn}(#1)}

\newcommand\dt{\mathit{Dup}}
\newcommand\idt{\mathit{Id}}

\newcommand\rn{\rho}
\newcommand\rnp[5]{\rn^{(#1,#2,#3)}_{\typing{#4}{#5}}}

\newcommand{\emb}{e}
\newcommand\embp[3]{\emb^{(#1,#2,#3)}}

\newcommand\tecp[5]{\tec^{(#1,#2,#3)}_{\typing{#4}{#5}}}

\newcommand\ssss[1]{\mathchoice{#1}{#1}{#1}{{\scriptstyle #1}}}

\newcommand\fr[2]{\frac{\ssss{#1}}{\ssss{#2}}}

\newcommand\sh[2]{#1 {\centerdot} #2}

\newcommand\shn{\mathtt{shn}}
\newcommand\decomp{\Phi_m} 
\newcommand\decompm[1]{\Phi_{#1}}

\newcommand\NIL{\epsilon}

\newcommand\hhole{\llbracket\,\rrbracket}
\newcommand\Hfill[1]{\llbracket #1 \rrbracket}
\newcommand\Ar{{\Rightarrow}}
\newcommand\cbf[2]{#1 \COL #2}

\newcommand\z{}
\newcommand\f{\kappa}

\newcommand\ctr{C}

\newcommand\prj[2]{#1 {\centerdot} #2}

\newcommand\mr{r}

\newcommand\cn{p}

\newcommand\cs[1]{\mathcal{C}({#1})}

\newcommand\nset{\mathcal{I}}

\newcommand\ds[2]{\mathrm{def}^{#1}_{#2}}
\DeclareMathOperator*\limd{lim def}

\newcommand\zps[1]{Z_{#1}} 
 
\newif\iffull
\fulltrue

\newif\ifdraft
\draftfalse

\ifdraft
\newcommand\asd[1]{{\footnotesize \color[RGB]{105,10,130}[#1 -asd]}}
\newcommand\nk[1]{{\footnotesize \color{blue}{[#1 -nk]}}}
\newcommand\tk[1]{{\footnotesize \color[RGB]{10,120,170}{[#1 -takeshi]}}}
\newcommand\ryoma[1]{{\footnotesize \color{DarkGreen}{[#1 -ryoma]}}}
\newcommand\garbage[1]{{\footnotesize \color[RGB]{90,2,2} [Garbage: #1]}} 
\newcommand{\new}[1]{{\color{red}#1}}
\else
\newcommand\asd[1]{}
\newcommand\nk[1]{}
\newcommand\tk[1]{}
\newcommand\ryoma[1]{}
\newcommand\garbage[1]{}
\newcommand{\new}[1]{#1}
\fi

\begin{document}
\title[Almost Every Simply Typed $\lambda$-Term Has a Long
$\beta$-Reduction Sequence]{Almost Every Simply Typed $\lambda$-Term\\ Has a Long
  $\beta$-Reduction Sequence
}

\author[K.~Asada]{Kazuyuki Asada\rsuper{a,*}}
\thanks{%
  \lsuper{*}The current affiliation is Tohoku University, Japan, and the mail address is \texttt{asada@riec.tohoku.ac.jp}}%
\author[N.~Kobayashi]{Naoki Kobayashi\rsuper{a}}
\address{\lsuper{b}The University of Tokyo, Japan}
\email{asada@kb.is.s.u-tokyo.ac.jp}
\email{koba@is.s.u-tokyo.ac.jp}
\email{tsukada@kb.is.s.u-tokyo.ac.jp}

\author[R.~Sin'ya]{Ryoma Sin'ya\rsuper{b}}
\address{\lsuper{b}Akita University, Japan}
\email{ryoma@math.akita-u.ac.jp}

\author[T.~Tsukada]{Takeshi Tsukada\rsuper{a}}

\begin{abstract}
It is well known that the length of a $\beta$-reduction sequence of
a simply typed $\lambda$-term of order $k$ can be huge;
it is as large as $k$-fold exponential in the size of the
$\lambda$-term
in the \emph{worst} case.
We consider the following relevant question about \emph{quantitative}
properties, instead of the worst case: \emph{how many} simply typed
$\lambda$-terms
have very long reduction sequences?
We provide a partial answer to this question,
by showing that asymptotically almost every simply typed
$\lambda$-term of order $k$ has
a reduction sequence as long as $(k-1)$-fold exponential in the term
size,
under the assumption that the arity of functions and
the number of variables that may occur in every subterm are
bounded above by a constant.
To prove it, we have extended the infinite monkey theorem for words
to a parameterized one for regular tree languages, which may be of independent interest.
The work has been motivated by quantitative analysis of the complexity
of higher-order model checking.

 \end{abstract}
\maketitle

\section{Introduction}\label{sec:introduction}
It is well known that a $\beta$-reduction sequence of
a simply typed $\lambda$-term can be extremely long.
Beckmann~\cite{Beckmann01} showed that, for any $k \geq 0$,
\[
 \max\{ \beta(t) \mid t \text{ is a simply typed } \lambda\text{-term of order } k \text{ and size }
 n \} = \Exp{k}{\Theta(n)}
\]
where $\beta(t)$ is the maximum length of the $\beta$-reduction sequences
of the term $t$, and $\Exp{k}{x}$ is defined by: $\Exp{0}{x} \defeq x$ and $\Exp{k+1}{x}
\defeq 2^{\Exp{k}{x}}$.
Indeed, the following order-$k$ term~\cite{Beckmann01}:
\newcommand\twice{\mathit{Twice}}
\newcommand\Twice[1]{\ctwok{#1}}
\[
(\Twice{k})^n \Twice{k-1} \cdots \Twice{2} (\lambda x^{\T}.a\,x\,x)
{((\lambda x^{\T}. x)c)},
\]
where $\Twice{j}$ is the twice function $\lambda f^{\Ttwice{j-1}}.\lambda x^{\Ttwice{j-2}}.f(f\,x)$ 
(with $\Ttwice{j}$ being the order-$j$ type defined by: $\Ttwice0=\T$ and $\Ttwice{j}=\Ttwice{j-1}\to\Ttwice{j-1}$),
has a $\beta$-reduction sequence of length $\Exp{k}{\Omega(n)}$.

Although the worst-case length of the longest $\beta$-reduction sequence is well known as above,
much is not known about the \emph{average-case} length
 of the longest $\beta$-reduction sequence:
\emph{how often} does one encounter a term having a very long $\beta$-reduction sequence?
In other words, suppose we pick a simply-typed $\lambda$-term $t$ of order $k$ and 
size $n$ \emph{randomly}; then
what is the probability that $t$ has a $\beta$-reduction sequence longer than a certain bound, like
$\Exp{k}{cn}$ (where $c$ is some constant)? One may expect that, although there exists
a term (such as the one above) whose reduction sequence is as long as $\Exp{k}{\Omega(n)}$,
such a term is rarely encountered.

In the present paper, we provide a partial answer to the above question, by showing that
almost every simply typed $\lambda$-term of order $k$ has a \(\beta\)-reduction
sequence as long as $(k-1)$-fold exponential in the term size, under
a certain assumption. More precisely, we shall
show:
\[
 \lim_{n \rightarrow \infty} \frac{\card{\{ [t]_\alpha \in
 \termsn{k}{\I}{\K}
 \mid \beta(t) \geq \Exp{k-1}{n^p}
 \}}}{\card{\termsn{k}{\I}{\K}}} = 1
\]
for some constant \(p>0\), where
$\termsn{k}{\I}{\K}$ is the set of ($\alpha$-equivalence
classes $[-]_\alpha$ of) simply-typed $\lambda$-terms
such that the term size is $n$, the order is up to $k$, the (internal)
arity is up to $\I \geq k$ and the number of variable names is up to
$\K$ (see the next section for the precise definition).

Related problems have been studied in the context of 
the quantitative analysis of 
\emph{untyped} $\lambda$-terms~\cite{SN,grygiel2013counting,Bendkowski}.
For example, David et al.~\cite{SN} have shown that
almost all untyped $\lambda$-terms are strongly normalizing, whereas the result
is opposite in the corresponding combinatory logic.
A more sophisticated analysis is, however, required in our case,
for considering only well-typed terms, and also for reasoning about
the \emph{length} of a reduction sequence instead of a qualitative property like strong normalization.

To prove our main result above, we have extended the infinite monkey theorem 
(a.k.a. ``Borges's theorem''~\cite[p.61, Note I.35]{Flajolet}) to a parameterized version
for regular tree languages. The infinite monkey theorem states that for any word \(w\),
a sufficiently long word almost surely contains \(w\) as a subword (see Section~\ref{sec:mainresults} for
a more precise statement). Our extended theorem, roughly speaking, states that, 
for any regular tree grammar \(\mathcal{G}\) 
that satisfies a certain condition and any \emph{family} \((U_m)_{m \in \nat}\) of \emph{trees}
(or tree contexts) generated by \(\mathcal{G}\) \new{such that $|U_m| = \order(m)$},
a sufficiently large tree \(T\) generated by \(\mathcal{G}\) almost
surely contains \(U_{\ceil{p\log{|T|}}}\) as a subtree
(where \(p\) is some positive constant).
Our main result is then obtained by preparing a regular tree grammar for simply-typed \(\lambda\)-terms,
and using as \(U_m\) a term having a very long \(\beta\)-reduction sequence, like 
\((\Twice{k})^m \Twice{k-1} \cdots \Twice{2} (\lambda x^{\T}.a\,x\,x)
\new{((\lambda x^{\T}. x) c)}\) given above.
The extended infinite monkey theorem mentioned above may be of independent interest and applicable to
other problems.

Our work is a part of our long-term project on the quantitative
analysis of the complexity of higher-order model checking~\cite{Knapik02FOSSACS,Ong06LICS}. The
higher-order model checking asks
whether the (possibly infinite) tree generated by a ground-type term
of the $\lambda$Y-calculus (or, a higher-order recursion scheme)
satisfies a given regular property, and it is known that the problem is
$k$-EXPTIME complete for order-$k$ terms~\cite{Ong06LICS}. Despite the huge worst-case
complexity, practical model checkers~\cite{Kobayashi13JACM,Kobayashi13horsat,Ramsay14POPL} 
have been built, which run fast for many typical inputs, and have successfully been
applied to automated verification of functional programs~\cite{Kobayashi09POPL,KSU11PLDI,Ong11POPL,SUK13PEPM}. 
The project aims to provide a theoretical justification for it, by studying \emph{how many}
inputs actually suffer from the worst-case complexity.
Since the problem appears to be hard due to recursion, 
as an intermediate step towards the goal, we aimed to analyze the variant of the problem
considered by Terui~\cite{Terui12RTA}: given a term of the simply-typed $\lambda$-calculus (without recursion)
of type Bool, decide whether it evaluates to true or false (where Booleans are Church-encoded; see \cite{Terui12RTA}
for the precise definition). Terui has shown that even for the problem, the complexity is $k$-EXPTIME complete
for order-$(2k+2)$ terms. If, contrary to the result of the present paper, 
the upper-bound of the lengths of $\beta$-reduction sequences \emph{were} small for almost every term,
then we could have concluded that the decision problem above is easily solvable for most of the inputs.
The result in the present paper does not necessarily provide a negative answer to the question above,
because one need not necessarily apply $\beta$-reductions to solve Terui's decision problem.

The present work may also shed some light on other problems on typed $\lambda$-calculi
with exponential or higher worst-case complexity. For example, 
despite DEXPTIME-completeness of ML typability~\cite{Mairson90POPL,Kfoury90},
it is often said that the exponential behavior is rarely seen \emph{in practice}. That is,
however, based on only empirical studies.
A variation of our technique may be used to provide a theoretical justification (or possibly \emph{un}justification). 

A preliminary version of this article appeared in Proceedings of FoSSaCS 2017~\cite{SAKT17FOSSACS}. 
Compared with the conference version, we have strengthened the main result (from ``almost every \(\lambda\)-term 
of size \(n\) has \(\beta\)-reduction sequence as long as \(\exp_{k-2}(n)\)'' to 
``... as long as \(\exp_{k-1}(n^p)\)''), and added proofs. We have also generalized the argument in
the conference version to formalize the parameterized infinite monkey theorem for regular tree languages.

The rest of this paper is organized as follows.
\Sec{mainresults} states our main result formally.
\Sec{proofimf} proves the extended infinite monkey theorem, and
\Sec{proofbeta} 
proves the main result.
\Sec{relatedwork} discusses related work, and 
\Sec{conclusion} concludes this article.

 \section{Main Results}\label{sec:mainresults}

This section states the main results of this article: the result on the quantitative analysis of
the length of \(\beta\)-reduction sequences of simply-typed \(\lambda\)-terms (Section~\ref{sec:mainresult})
and the parameterized infinite monkey theorem for
regular tree grammars (Section~\ref{sec:RegTreeGram}).
The latter is used in the proof of the former in Section~\ref{sec:proofbeta}.
In the course of giving the main results, we also introduce various notations used
in later sections; the notations are summarized in Appendix~\ref{sec:symbols}.
We assume some familiarity with the simply-typed \(\lambda\)-calculus and
regular tree grammars and omit to define some standard concepts (such as \(\beta\)-reduction); 
readers who are not familiar with them
may wish to consult \cite{TypedLambda} about the simply-typed \(\lambda\)-calculus
and \cite{Saoudi92,tata2007} about regular tree grammars.

For a set $A$, we denote by $\card{A}$ the cardinality of $A$; by abuse of notation, we write \( \card{A} = \infty \) to mean that \( A \) is infinite.
For a sequence \(s\), we also denote by $\len{s}$ the length of $s$.
For a sequence \(s=a_1 a_2 \cdots a_{\len{s}}\) and \(i\le\len{s}\),
we write \(\prj{s}{i}\) for the \(i\)-th element \(a_i\). We write \(s_1\cdot s_2\) or just \(s_1s_2\)
for the concatenation of two sequences \(s_1\) and \(s_2\).
We write \(\NIL\) for the empty sequence.
We use \(\biguplus\) to denote the union of disjoint sets.
For a set $I$ and a family of sets $(A_i)_{i \in I}$, we define 
\(\coprod_{i \in I} A_i \defeq \biguplus_{i \in I} (\{i\} \times A_i) = \set{(i,a) \mid i\in I, a\in A_i}\).
For a map \(f\), we denote the domain and image of \(f\) by \(\dom(f)\) and \(\image{f}\), respectively.
We denote by \(\log{x}\) the binary logarithm \(\log_{2}{x}\).
\nk{Omitted ``...and by \(\ln x\) the natural logarithm \(\log_{e}{x}\).'' since
it was used only in the proof of convergence, which has now been omitted.}

\subsection{The Main Result on Quantitative Analysis of the Length of \(\beta\)-Reductions}
\label{sec:mainresult}

The set of (\emph{simple}) \emph{types}, ranged over by \(\tau,\tau_1,\tau_2,\ldots\), 
is defined by the grammar:
\[\tau ::= \T \mid \tau_1 \rightarrow \tau_2.\]
Note that the set of types can also be generated by the following grammar:
\[\tau ::= \tau_1 \to\cdots\to \tau_k\to \T\]
where \(k \geq 0\); we sometimes use the latter grammar for inductive definitions of functions on types.
We also use \(\sigma\) as a metavariable for types.
\begin{rem}
  We have only a single base type \(\T\) above.  The main result
  (Theorem~\ref{mainthm})  would not change even if there are
  multiple base types.
\end{rem}

Let $V$ be a countably infinite set, which is ranged over by \(x, y, z
\, (\text{and } x', x_1, \text{etc.})\).
The set of \emph{$\lambda$-terms} (or \emph{terms}), ranged over by $t,t_1,t_2,\ldots$, is defined by:
\begin{eqnarray*}
 t ::= x \mid \lambda \var^\tau\!. t \mid t_1 \, t_2\qquad\qquad
 \var ::= x \mid \uv.
\end{eqnarray*}
We call elements of \(V \cup \set{\uv}\) \emph{variables}, and use meta-variables
\(\var, \overline{y}, \overline{z} \, (\text{and } \var', \var_1, \text{etc.})\) for them.
We sometimes omit type annotations and just write $\lambda \var.t$ for $\lambda
\var^\tau\!. t$.
We call 
the special variable $\uv$ an \emph{unused variable}, which may be
bound by \(\lambda\), but must not occur in the body.
In our quantitative analysis below, we will count
the number of variable names occurring in a term, except \(\uv\).
For example, the term \(\lambda x.\lambda y.\lambda z.x\) in the standard syntax can be represented by \(\lambda x.\lambda \uv.\lambda \uv.x\),
and the number of used variables in the latter is counted as \(1\).

Terms of our syntax can be translated to usual $\lambda$-terms by regarding
elements in \(V\cup\set{\uv}\) as usual variables.
Through this identification we define the notions of free variables, closed terms, and
\(\alpha\)-equivalence \(\aeq\).
The $\alpha$-equivalence class of a term $t$ is written as $[t]_\alpha$.
In this article, we distinguish between a term and its $\alpha$-equivalence class, and 
we always use $[-]_\alpha$ explicitly.
For a term \(t\), we write \(\fv(t)\) for the set of all the free variables of \(t\).

For a term $t$,
we define the set $\V(t)$ of variables (except $\uv$) in $t$ by:
\[
\V(x) \defeq \{x\} \quad\V(\lambda x^\tau\!.t) \defeq \{x\} \cup  \V(t)
\quad
  \V(\lambda \uv^\tau\!.t) \defeq \V(t) \quad \V(t_1 t_2) \defeq \V(t_1) \cup \V(t_2).
\]
Note that neither \(\V(t)\) nor even \(\card{\V(t)}\) is preserved by \(\alpha\)-equivalence.
For example, 
\(t= \lambda x. (\lambda y. y)(\lambda z. x)\)
and
\(t'=\lambda x. (\lambda x. x)(\lambda \uv. x)\)
are \(\alpha\)-equivalent, but \(\card{\V(t)}=3\) and \(\card{\V(t')}=1\).
We write \(\numofvars{(t)}\) for
\(\min_{t' \in [t]_\alpha}\!\! \card{\V(t')}\), i.e.,
the minimum number of variables required
to represent an \(\alpha\)-equivalent term.
For example, for \(t=\lambda x. (\lambda y. y)(\lambda z. x)\) above,
\(\numofvars{(t)}=1\),
  because \(t\sim_\alpha t'\) and
\(\card{\V(t')}=1\).

A \emph{type environment} $\Gamma$ is a finite set of type bindings of the form
 $x: \tau$ such that if $(x:\tau) ,(x:\tau') \in \Gamma$ then \(\tau = \tau'\);
sometimes we regard an environment also as a function.
When we write \(\Gamma_1\cup\Gamma_2\), we implicitly require that
$(x:\tau) \in\Gamma_1$ and $(x:\tau') \in \Gamma_2$ imply \(\tau = \tau'\),
so that \(\Gamma_1\cup\Gamma_2\) is well formed.
Note that \((\uv:\tau)\) cannot belong to a type environment;
we do not need any type assumption for \(\uv\) since it does not
 occur in terms.

The type judgment relation \(\Gamma \p t : \tau\) is inductively defined by the following typing rules.
\\[1pt]
\[
\begin{gathered}
\frac{}{x:\tau \p x : \tau}
\qquad
\frac{
\Gamma_1 \p t_1 : \sigma \ft \tau
\qquad
\Gamma_2 \p t_2 : \sigma
}{\Gamma_1 \cup \Gamma_2 \p t_1 t_2 : \tau}
\\[10pt]
\frac{
\Gamma' \p t:\tau
\qquad \Gamma'=\Gamma\text{ or }\Gamma' = \Gamma\cup\set{\var:\sigma}
\qquad \var \notin \dom(\Gamma)
}{\Gamma \p \lambda \var^{\sigma}\!.t:\sigma \ft \tau}
\end{gathered}
\]
\\[5pt]
The type judgment relation is equivalent to the usual one for the simply-typed \(\lambda\)-calculus, except that a
type environment for a term may contain only variables that occur free in the term 
(i.e., if \(\Gamma \p t :\tau\) then \(\dom(\Gamma)=\fv(t)\)).
Note that if \(\Gamma \p t :\tau\) is derivable, then the derivation is unique.
\nk{I have omitted the terminology ``a \emph{\(\typing{\Gamma}{\tau}\)-term}''.}
Below we consider only well-typed $\lambda$-terms, i.e., those that are inhabitants of some typings.

\begin{defi}[Size, Order and Internal Arity of a Term]
The \emph{size} of a term $t$, written $|t|$, is defined by:
\begin{eqnarray*}
  |x| \defeq 1 \;\;\;\;\;\;  |\lambda \var^\tau \!. t| \defeq |t| + 1 \;\;\;\;\;\;
   |t_1 t_2| \defeq |t_1| + |t_2| + 1.
\end{eqnarray*}
The \emph{order} and \emph{internal arity} of a type $\tau$, written
 $\ord(\tau)$ and $\ar(\tau)$ respectively, are defined by:
\[
 \begin{aligned}
&\ord(\tau_1 \rightarrow \cdots \rightarrow \tau_n \rightarrow \T) 
\defeq \max(\{0\} \cup \{ \ord(\tau_i) + 1 \mid 1 \le i \le n \}) 
\\
&\ar(\tau_1 \rightarrow \cdots \rightarrow \tau_n \rightarrow \T) 
\defeq \max(\{n\} \cup \{ \ar(\tau_i) \mid 1 \le i \le n \}) 
 \end{aligned}
\]
where \(n \ge 0\). We denote by $\types{\D}{\I}$ the set of types $\{ \tau \mid \ord(\tau) \leq \D,
 \ar(\tau) \leq \I \}$. 
For a judgment 
\(\Gamma \p t: \tau\),
we define
 the \emph{order} and \emph{internal arity} of \(\Gamma \p t: \tau\), written
 $\tord{\Gamma}{t}{\tau}$ and $\tar{\Gamma}{t}{\tau}$ respectively, by:
\[
\begin{aligned} &
\tord{\Gamma}{t}{\tau} \defeq \max\set{ \ord(\tau') \mid (\Gamma' \p t' : \tau') \text{ occurs in } \td }
\\&
\tar{\Gamma}{t}{\tau} \defeq \max\set{ \ar(\tau') \mid (\Gamma' \p t' : \tau') \text{ occurs in } \td }
\end{aligned}
\]
where $\td$ is the (unique) derivation tree for $\Gamma \p t: \tau$.
\end{defi}
Note that the notions of size, order, internal arity, and \(\beta(t)\)
(the maximum length of \(\beta\)-reduction sequences of \(t\), as defined in Section~\ref{sec:introduction})
are well-defined with respect to
\(\alpha\)-equivalence.

\begin{exa}
Recall the term \(\Twice{j} \defeq \lambda f^{\Ttwice{j-1}}.\lambda x^{\Ttwice{j-2}}.f(f\,x) \) (with $\Ttwice{j}$ being the order-$j$ type defined by: $\Ttwice0=\T$ and $\Ttwice{j}=\Ttwice{j-1}\to\Ttwice{j-1}$) in Section~\ref{sec:introduction}.
For \(t=\Twice{3}\,\Twice{2}\,(\lambda x^{\T}.x)\,y\),
we have \(\tord{y\COL\T}{t}{\T} = \tar{y\COL\T}{t}{\T}=3\).
Note that the derivation for \(y\COL\T\p t:\T\) contains the type judgment \(\emptyset\p \Twice{3}:
\Ttwice{3}\), where \(\Ttwice{3} = ((\T\to\T)\to(\T\to\T))\to(\T\to\T)\to\T\to\T\), and 
\(\ord(\Ttwice{3})=\ar(\Ttwice{3})=3\).
\end{exa}

We now define the sets of ($\alpha$-equivalence classes of) terms
with bounds on the order, the internal arity, and the number of variables.
\begin{defi}[Terms with Bounds on Types and Variables]
\label{df:set-of-terms}
Let $\D,\I,\K \geq 0$ and $n \geq 1$ be integers.
For each $\Gamma$ and $\tau$, 
\(\termsg{\Gamma}{\tau}{\D}{\I}{\K}\) and 
\(\termsng{\Gamma}{\tau}{\D}{\I}{\K}\) are defined by:
\[
\begin{aligned}
 \termsg{\Gamma}{\tau}{\D}{\I}{\K} &\defeq \{ [t]_\alpha \mid
  \Gamma \p t: \tau,
  \tord{\Gamma}{t}{\tau} \leq \D, \tar{\Gamma}{t}{\tau} \leq \I,
  \numofvars{(t)}\le \K
  \}
\\
 \termsng{\Gamma}{\tau}{\D}{\I}{\K} &\defeq
 \{ [t]_\alpha \in \termsg{\Gamma}{\tau}{\D}{\I}{\K} \mid |t| = n \}.
\end{aligned}
\]
We also define:
\begin{equation*}
 \terms{\D}{\I}{\K} \defeq \, \biguplus_{\mathclap{\tau \in
 \types{\D}{\I}}}\, \termsg{\ee}{\tau}{\D}{\I}{\K} \qquad
 \termsn{\D}{\I}{\K}  \defeq \{ [t]_\alpha \in \terms{\D}{\I}{\K} \mid
 |t| = n \}.
\end{equation*}
\end{defi}
\noindent
Intuitively,  \(\termsg{\Gamma}{\tau}{\D}{\I}{\K}\)
is the set of (the equivalence classes of)
terms of type \(\tau\) under \(\Gamma\), within
given bounds \(\D\), \(\I\), and \(\K\) on the order, (internal) arity, and the number of variables respectively;
and  \(\termsng{\Gamma}{\tau}{\D}{\I}{\K}\) is the subset of 
\(\termsg{\Gamma}{\tau}{\D}{\I}{\K}\) consisting of terms of size \(n\).
The set \( \terms{\D}{\I}{\K}\)  consists of (the equivalence classes of)
all the closed (well-typed) terms, and
\(\termsn{\D}{\I}{\K}\) consists of those of size \(n\).
For example,
\(t=[\lambda x^\T.\lambda y^\T.\lambda z^\T.x]_\alpha\) belongs to 
\( \termsN{4}{1}{3}{1}\); note that
\(|t|=4\) and 
\(\numofvars{\new{(}t\new{)}}
= \numofvars{(\lambda x^\T.\lambda \uv.\lambda \uv.x)}=1\).

We are now ready to state our main result, which is proved in Section~\ref{sec:proofbeta}.
\begin{thm}\label{mainthm}
For $\D, \I, \K \geq 2$ and $k = \min\{\D,
 \I\}$, there exists a real number $p > 0$ such that
\[
  \lim_{n \rightarrow \infty} \frac{\card{\{ [t]_\alpha \in
 \termsn{\D}{\I}{\K}
 \mid \beta(t) \geq \Exp{k-1}{n^p}
 \}}}{\card{\termsn{\D}{\I}{\K}}} = 1.
\]
\end{thm}
\noindent
As we will see later (in Section~\ref{sec:proofbeta}), the
denominator \(\card{\termsn{\D}{\I}{\K}}\) is nonzero if \(n\) is sufficiently
large.
The theorem above says that if the order, internal arity, and the number of used variables are
bounded independently of term size, most of the simply-typed \(\lambda\)-terms of size \(n\) have
a very long \(\beta\)-reduction sequence, which is as long as 
$(k-1)$-fold exponential in \(n\).
\begin{rem}
Recall that, when \(k=\D\), the worst-case length of \(\beta\)-reduction sequence is \(k\)-fold
 exponential~\cite{Beckmann01}.
We do not know whether \(k-1\) can be replaced with \(k\) in the theorem above.

Note that in the above theorem, the order \(\delta\), the internal arity \(\I\) and the number \(\K\) of
variables are bounded above by a constant, \emph{independently of the term size} \(n\). Our proof of the theorem
(given in Section~\ref{sec:proofbeta}) makes use of this assumption to model the set of simply-typed \(\lambda\)-terms 
as a regular tree language.
It is debatable whether our assumption is reasonable. A slight change of the assumption
may change the result, as is the case for strong normalization of untyped \(\lambda\)-terms~\cite{SN,Bendkowski}.
When \(\lambda\)-terms are viewed as models of functional programs, 
our rationale behind the assumption is as follows. 
The assumption that the size of types (hence also the order and the internal arity) is fixed
is sometimes assumed in the context of type-based program analysis~\cite{Heintze97}. The assumption on the number
of variables comes from the observation that a large program usually consists of a large number of 
\emph{small} functions, and that the number of variables is bounded by the size of each function. 
\end{rem}
 
\subsection{Regular Tree Grammars and Parameterized Infinite Monkey Theorem}
\label{sec:RegTreeGram}

To prove Theorem~\ref{mainthm} above, 
we extend the well-known \emph{infinite monkey theorem} (a.k.a. ``Borges's
theorem''~\cite[p.61, Note I.35]{Flajolet}) to a parameterized version for regular tree grammars,
and apply it to the regular tree grammar that generates
the set of ($\alpha$-equivalence classes of) simply-typed
\(\lambda\)-terms. Since the extended infinite monkey theorem may be of independent interest,
we state it (as Theorem~\ref{thm:imt}) in this section as one of the main results. The theorem is proved in Section~\ref{sec:proofimf}.
We first recall some basic definitions for regular tree grammars
in Sections~\ref{sec:context} and \ref{sec:gram}, and then state
the parameterized infinite monkey theorem in Section~\ref{sec:pimf}.
\subsubsection{Trees and Tree Contexts}
  \label{sec:context}
A \emph{ranked alphabet} $\ralph$ is a map from a finite set of \emph{terminal} symbols
to the set of natural numbers. We use the metavariable \(a\) for a terminal,
and often write \(\Ta,\Tb,\Tc,\ldots\) for concrete terminal symbols.
For a terminal $a \in \dom(\ralph)$, we call $\ralph(a)$ the \emph{rank}
of $a$.
A \emph{$\ralph$-tree} is a tree constructed from terminals in $\ralph$
according to their ranks:
$a(\tr_1, \ldots, \tr_{\ralph(a)})$ is a $\ralph$-tree if
$\tr_i$ is a $\ralph$-tree for each $i \in \set{1,\ldots, \ralph(a)}$.
Note that \(\ralph(a)\) may be \(0\): $a(\,)$ is a $\ralph$-tree if \(\ralph(a)=0\).
We often write just \(a\) for \(a(\,)\).
For example, if \(\ralph=\set{\Ta\mapsto 2, \Tb\mapsto 1, \Tc\mapsto 0}\),
then \(\Ta(\Tb(\Tc), \Tc)\) is a tree; see (the lefthand side of) Figure~\ref{fig:tree} for
a graphical illustration of the tree.
We use the meta-variable $\tr$ for trees.
The \emph{size} of $\tr$,
written $|\tr|$, is the number of occurrences of terminals in $\tr$.
We denote the set of all $\ralph$-trees by $\TR{\ralph}$,
and the set of all $\ralph$-trees of size \(n\) by $\TRn{n}{\ralph}$.

\begin{figure}
  \begin{minipage}[t]{.4\textwidth}
  \(
  \Tree[ [ [].{$\Tc$} ].{$\Tb$} [ ].{$\Tc$} ].{$\Ta$}
  \)
  \end{minipage}
  \begin{minipage}[t]{.4\textwidth}
  \(
  \Tree[  [ [].{$\hole$}  [].{$\Tc$} ].{$\Ta$} [ [].{$\hole$} ].{$\Tb$} ].{$\Ta$}
  \)
  \end{minipage}
\caption{A tree \(\Ta(\Tb(\Tc), \Tc)\) (left) and a context 
\(\Ta(\Ta(\hole, \Tc), \Tb(\hole))\) (right).}
\label{fig:tree}
\end{figure}
Before introducing grammars, we define tree contexts,
which
play an important role in our formalization and proof of the (parameterized) infinite
monkey theorem for regular tree languages.
The set of \emph{contexts} over a ranked alphabet $\ralph$, ranged over by
$\pct$, is defined by:
\[
 \pct ::= \hole \mid a(\pct_1, \ldots, \pct_{\ralph(a)}).
\]
In other words, a context is a tree where the alphabet is extended 
with the special nullary symbol \(\hole\). For example, \(\Ta(\Ta(\hole, \Tc), \Tb(\hole))\)
is a context over \(\set{\Ta\mapsto 2, \Tb\mapsto 1, \Tc\mapsto 0}\);
see (the \new{righthand} side of) Figure~\ref{fig:tree} for
a graphical illustration.
We write \(\cs{\ralph}\) for the set of contexts over $\ralph$.
For a context \(\pct\),
we write \(\hn{\pct}\) for the number of occurrences of \(\hole\) in $\pct$.
We call \(\pct\) a \emph{\(k\)-context}
if \(\hn{\pct}=k\),
and an \emph{\affine{context}} if \(\hn{\pct} \le 1\).
A 1-context is also called a \emph{\linear{context}}.
We use the metavariable \(S\) for \linear{contexts} and \(U\) for \affine{contexts}.
We write \(\hole_i\) for the
\(i\)-th hole occurrence (in the left-to-right order) of \(\pct\).

For contexts \(C\), \(C_1, \dots, C_{\hn{C}}\),
we write 
\( C[C_1, \dots, C_{\hn{C}}] \) 
for the context obtained by
replacing each \(\hole_i\) in $C$ with \(C_i\). For example, if
\(C=\Ta(\hole, \Ta(\hole,\hole))\),
\(C_1 = \Tb(\hole)\), \(C_2 = \Tc\), and \(C_3 = \hole\), then
\(C[C_1,C_2,C_3] = \Ta(\Tb(\hole), \Ta(\Tc,\hole))\);
see Figure~\ref{fig:cfill} for a graphical illustration.
Also, for contexts \(C\), \(C'\) and \(i \in \set{1,\dots,\hn{C}}\),
we write \(C[C']_i\) for the context obtained by
replacing \(\hole_i\) in $C$ with \(C'\). For example, for \(C\) and \(C_1\) above,
\(C[C_1]_2 = \Ta(\hole, \Ta(\Tb(\hole), \hole))\).

For contexts \(\ctr\) and \(\ctr'\),
we say that \(\ctr\) is a \emph{subcontext of} \(\ctr'\) and write $\ctr \subc \ctr'$ if
there exist \(\ctr_0\), \(\ctr_1,\ldots,\ctr_{\hn{\ctr}}\), and \(1 \le i \le \hn{\ctr_0}\) such that \(\ctr' =
\ctr_0[\ctr[\ctr_1, \dots ,\ctr_{\hn{\ctr}}]]_i\). For example,
\(\ctr=\Ta(\Tb(\hole),\hole)\)
is a subcontext of \(\ctr'=\Ta(\hole, \Ta(\Tb(\hole), \Tc))\), because
\(\ctr'=\ctr_0[\ctr[\hole,\Tc]]_2\) for \(\ctr_0=\Ta(\hole,\hole)\); see also
Figure~\ref{fig:subcontext} for a graphical illustration.
The subcontext relation \(\subc\) may be regarded as a generalization
of the subword relation.
 In fact, a word \(w = a_1\cdots a_k\) can be viewed as a 
\linear{context} \(w^\sharp = a_1(\cdots (a_k \hole)\cdots)\); and
if \(w\) is a subword of \(w'\), i.e., if
\(w'=w_1ww_2\), then \(w'^\sharp = w_1^\sharp[w^\sharp[w_2^\sharp]]\).
\begin{figure}
  \begin{minipage}[t]{.2\textwidth}
  \(C:\)\\
  \(
  \Tree[ [].{$\hole$}  [ [].{$\hole$}  [].{$\hole$} ].{$\Ta$} ].{$\Ta$}
  \)
  \end{minipage}
  \begin{minipage}[t]{.2\textwidth}
  \(C_1:\)\\
  \(
  \Tree[ [].{$\hole$}  ].{$\Tb$}
  \)
  \end{minipage}
  \begin{minipage}[t]{.2\textwidth}
  \(C[C_1,C_2,C_3]:\)\\
  \(
  \Tree[ [ [].{$\hole$} ].{$\Tb$}  [ [].{$\Tc$}  [].{$\hole$} ].{$\Ta$} ].{$\Ta$}
  \)
  \end{minipage}
  \caption{Context filling (where \(C_2=c\) and \(C_3=\hole\)).}
  \label{fig:cfill}
\end{figure}
\begin{figure}
  \includegraphics[scale=0.5]{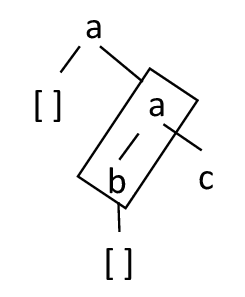}
  \caption{A graphical illustration of a subcontext. The part surrounded by
    the rectangle is the subcontext \(\Ta(\Tb(\hole),\hole)\).}
  \label{fig:subcontext}
\end{figure}
The \emph{size} of a context \(U\), written \(|U|\), is defined by: 
\(|\hole| \defeq 0\) and \(|a(\ctr_1,\dots,\ctr_{\ralph(a)})| \defeq 1 + |\ctr_1|+\dots+|\ctr_{\ralph(a)}|\).
Note that \(\hole\) and 0-ary terminal \(a\) have different sizes: \(|\hole|=0\) but \(|a|=1\).
For a \(0\)-context \(\ctr\), \(|\ctr|\) coincides with the size of \(\ctr\) as a tree.

\subsubsection{Tree Grammars}
\label{sec:gram}
A \emph{regular tree grammar}~\cite{Saoudi92,tata2007} (\emph{grammar} for short) is a triple $\grm = (\ralph,
\nont, \rules)$ where 
(i) $\ralph$ is a ranked alphabet; 
(ii) $\nont$ is a finite set of symbols called \emph{nonterminals};
(iii) $\rules$ is a finite set of \emph{rewriting rules} of the form 
$N \rew C[N_1, \dots, N_{\hn{C}}]$ where 
$C \in \cs{\ralph}$ and $N, N_1, \dots, N_{\hn{C}} \in \nont$.
We use the metavariable \(N\) for nonterminals.
We write \(\ralph\cup\nont\) for
the ranked alphabet \(\ralph\cup\set{N \mapsto 0 \mid N\in \nont}\) and often regard the right-hand-side of a rule as a \( (\ralph\cup\nont) \)-tree.

The rewriting relation \(\rew_\grm\) on \(\TR{\ralph \cup \nont}\) is inductively defined by the following rules:
\\
\begin{align*}
\frac{(N \rew C[N_1, \dots, N_{\hn{C}}]) \in \rules}
{N \rew_\grm C[N_1, \dots, N_{\hn{C}}]}
\qquad
\frac{T_i \rew_\grm T'_i}{a(T_1,\dots,T_k)
\rew_\grm a(T_1,\dots,T_{i-1},T'_i,T_{i+1},\dots,T_k)}
\end{align*}
\\[4pt]
We write \( \rew^*_{\grm} \) for the reflexive and transitive closure of \( \rew_{\grm} \).
For a tree grammar $\grm = (\ralph, \nont, \rules)$ and a nonterminal
$N \in \nont$, \emph{the language $\inhav{\grm,N}$ of} $N$ is defined by
$\inhav{\grm,N} \defeq \{ \tr \in \TR{\ralph} \mid N \rew^*_\grm \tr \}$.
We also define
$\inhavn{\grm,N} \defeq \{ \tr \in \inhav{\grm,N} \mid |\tr| = n\}$.
If \(\grm\) is clear from the context, we often omit \(\grm\) and 
just write \(T \rew^* T'\), \(\inhav{N}\), and \(\inhavn{N}\) for
\(T \rew^*_\grm T'\), \(\inhav{\grm,N}\), and  \(\inhavn{\grm,N}\), respectively.

We are interested in not only trees, but also contexts generated by a grammar.
Let \(\grm = (\ralph, \nont, \rules)\) be a grammar.
The set of \emph{\ctyping}s of \(\grm\), ranged over by $\f$,
is defined by:
\[
\f ::=  N_1\cdots N_k\Ar N
\]
where \(k \in \nat\) and $N, N_1, \ldots, N_k \in \nont$. Intuitively, \(N_1\cdots N_k\Ar N\)
denotes the type of \(k\)-contexts \(\pct\) such that 
\(N\rew^* \pct[N_1,\ldots,N_k]\).
Based on this intuition, we define
the sets \(\pinhav{\grm,\f}\)
and \(\pinhavn{\grm,\f}\) of contexts by:
\begin{align*}
  \pinhav{\grm,N_1 \cdots N_k \Ar N}
    &\defeq \set{ \pct \in \cs{\ralph} \mid \hn{\pct} = k, N \rew^{*}_{\grm}
        \pct[N_1,\ldots,N_k]
  }
\\
\pinhavn{\grm,\f} &\defeq \set{\pct \in \pinhav{\grm,\f} \mid |\pct| = n}.
\end{align*}
We also write \(C\COL\f\) when \(C\in\pinhav{\grm,\f}\) (assuming that \(\grm\) is clear from the context).
Note that \(\pinhav{\grm,\Ar N} = \inhav{\grm,N}\), and
we identify a \ctyping{} \(\Ar N\) with the nonterminal \(N\).
We call an element in \(\pinhav{\grm,\f}\) a \emph{\(\f\)-context},
and also a \emph{\(\f\)-tree} if it is a tree.
It is clear that,
if \(\ctr \in \pinhav{\grm,N_1 \cdots N_k \Ar N}\)
and \(\ctr_i \in \pinhav{\grm,N^i_1 \cdots N^i_{\ell_i} \Ar N_i}\) (\(i=1,\dots,k\)), then 
\( 
\ctr[\ctr_1,\dots,\ctr_{k}]
\in \pinhav{\grm, N^1_1 \cdots N^1_{\ell_1} \cdots N^k_1 \cdots N^k_{\ell_k} \Ar N}\).
We also define:
\begin{align*}
 \inhav{\grm} &\defeq \, \bigcup_{\mathclap{N \in \nont}}\, \pinhav{\grm,N}
&
 \inhavn{\grm} &\defeq \{ T \in \inhav{\grm} \mid |T| = n\}
\\
 \ctx{\grm} &\defeq \, \bigcup_{\mathclap{N,N' \in \nont}}\, \pinhav{\grm,N\Ar N'}
&
 \ctxn{\grm} &\defeq \{ S \in \ctx{\grm} \mid |S| = n\}
\\
\octx{\grm} &\defeq \inhav{\grm} \cup \ctx{\grm} 
&
\octxn{\grm} &\defeq \set{U \in \octx{\grm} \mid |U| = n}.
\end{align*}

\begin{exa}\label{ex:grammar}
Consider the grammar \(\grm_0 = (\set{\Ta\mapsto 2,\Tb\mapsto 1,\Tc\mapsto 0},\set{A_0,B_0},\rules_0)\), where \(\rules_0\)
consists of:
\[
A_0 \rew \Ta(B_0, B_0) \qquad A_0 \rew \Tc \qquad B_0 \rew \Tb(A_0).
\]
Then, \(\Ta(\hole,\hole)\in \pinhavnn{1}{\grm_0,B_0\,B_0 \Ar A_0}\) and
 \(\Ta(\Tb(\hole),\Tb(\Tc))\in \pinhavnn{4}{\grm_0, A_0 \Ar A_0}\).
\end{exa}

This article focuses on grammars that additionally satisfy two properties called 
\emph{strong connectivity} and \emph{unambiguity}.
We first define strong connectivity.

\begin{defi}[Strong Connectivity~\cite{Flajolet}\footnote{\label{foot:StrConApr}
In~\cite{Flajolet}, strong connectivity is defined for \emph{context-free specifications}. Our
definition is a straightforward adaptation of the definition for regular
tree grammars.
}]
Let $\grm = (\ralph, \nont, \rules)$ be a regular tree grammar.
We say that \(N'\) is \emph{reachable} from \(N\) if 
\(\pinhav{\grm, N'\Ar N}\) is non-empty, i.e., if
there exists
a \linear{context} \( \mvlinctr\in \cs{\ralph} \) such that \( N \rew^* \mvlinctr[N'] \).
We say that $\grm$ is \emph{strongly connected} if for any pair $N, N' \in \nont$ of nonterminals,
\(N'\) is reachable from \(N\).
\end{defi}
\begin{rem}
  There is another reasonable, but slightly weaker condition of reachability: 
\( N' \) is \emph{reachable} from \( N \) if there exists 
a \((\ralph \cup \nont)\)-tree \(\tr\) such that \(N \rew^* \tr\) and $N'$ occurs in \(\tr\).
  The two notions coincide if \( \inhav{\grm, N} \neq \emptyset \) for every \( N \in \nont \).
  Furthermore this condition can be easily fulfilled by simply removing all the nonterminals \( N \) 
with \( \inhav{\grm, N} = \emptyset \) and the rules containing those nonterminals.
\end{rem}

\begin{exa}\label{ex:sc}
The grammar $\grm_0$ in Example~\ref{ex:grammar} is strongly connected,
 since $\Tb(\hole) \in \inhav{\grm_0, A_0 \Ar B_0}$ and $\Ta(\hole, \Tb(\Tc)) \in
 \inhav{\grm_0, B_0 \Ar A_0}$.

 Next, consider the grammar $\grm_1 = (\set{ \Ta \mapsto 2, \Tb \mapsto 1, \Tc
 \mapsto 0}, \set{A_1, B_1}, \rules_1)$, where $\rules_1$ consists of:
 \[
 A_1 \rew \Ta(\Tc, A_1) \qquad A_1 \rew \Tb(B_1) \qquad B_1 \rew \Tb(B_1) \qquad B_1 \rew \Tc.
 \]
$\grm_1$ is not strongly connected, since $\inhav{\grm_1, A_1 \Ar B_1}
 = \emptyset$.
\end{exa}

To define the unambiguity of a grammar, we define (the standard notion of) leftmost rewriting.
The \emph{leftmost rewriting relation} \( \rew_{\grm,\ell} \) is the restriction of the rewriting relation that allows only the leftmost occurrence of nonterminals
to be rewritten.
Given \( T \in \TR{\ralph} \) and \( N \in \nont \),
a \emph{leftmost rewriting sequence of \( T \) from \( N \)} is a sequence \( T_1, \dots, T_n \) of \( (\ralph \cup \nont) \)-trees such that
\[
N = T_1 \rew_{\grm,\ell} T_2 \rew_{\grm,\ell} \dots \rew_{\grm,\ell} T_n = T.
\]
It is easy to see that \( T \in \inhav{\grm, N} \) if and only if there exists a leftmost rewriting sequence from \( N \) to \( T \).
\begin{defi}[Unambiguity]\label{def:unambiguity}
  A grammar \( \grm = (\ralph, \nont, \rules) \) is \emph{unambiguous} if, for each \( N \in \nont \) and \( T \in \TR{\ralph} \),
 there exists 
at most one leftmost rewriting sequence from \( N \) to \( T \).
\end{defi}

\subsubsection{Parameterized Infinite Monkey Theorem for Tree Grammars}
\label{sec:pimf}
We call a partial function \(f\) from \(\nat\) to \(\real\) with infinite definitional domain 
(i.e., \(\dom(f)=\set{n\in\nat \mid f(n)\mbox{ is defined}}\) is infinite)
a \emph{partial sequence}, and write \(f_n\) for \(f(n)\).
For a partial sequence \(f = (f_n)_{n\in \nat}\), we define
\(\ds{f}{(-)}\) as the bijective monotonic function from \(\nat\) to \(\dom(f)\), 
and then we define:
\[
\limd_{n \rightarrow \infty} f_n \defeq \lim_{n \rightarrow \infty} f_{\ds{f}{n}}.
\]
All the uses of \(\limd\) in this article are of the form:
\(\limd_{n\rightarrow \infty}\fr{b_n}{a_n}\)
where \((a_n)_{n \in \nat}\) and \((b_n)_{n \in \nat}\) are (non-partial) sequences (hence, \(\fr{b_n}{a_n}\) is 
undefined only when \(a_n=0\)).
The following property on \(\limd\)
is straightforward, which we leave for the reader to check.
\begin{lem}
\label{lem:convSum}
If \(\limd\limits_{n \rightarrow \infty} \fr{b^{(i)}_n}{a^{(i)}_n} = 1\) for \(i=1,\dots,k\), 
and if \(0 \le b^{(i)}_n \le a^{(i)}_n\) for \(i=1,\dots,k\) and for \(n \in \nat\),
then \(\limd\limits_{n \rightarrow \infty} \fr{\sum_i b^{(i)}_n}{\sum_i a^{(i)}_n} = 1\).
\end{lem}

Now we are ready to state the second main result of this article.
\newcommand\thmstatementIMTcap{Parameterized Infinite Monkey Theorem for Regular Tree Languages}
\newcommand\thmstatementIMT{
Let $\grm = (\ralph, \nont, \rules)$ be an unambiguous and
\SC{} 
regular tree grammar
such that \(\card{\inhav{\grm}}=\infty\),
and \((S_n)_{n \in \nat}\) be a family of \linear{contexts} in \(\ctx{\grm}\)
such that $|S_n| = \order(n)$.
Then there exists a real constant \(\cn > 0\) such that
for any \(N \in \nont\) the following equation holds:
\[
 \limd_{n \rightarrow \infty} \frac{\card{\{ \tr \in \inhavn{\grm,N} \mid
 S_{\ceil{\cn \log n}} \subc \tr
 \}}}{\card{\inhavn{\grm,N}}} = 1.
\]
}
\begin{thm}[\thmstatementIMTcap]
\label{thm:imt}
\thmstatementIMT
\end{thm}
Intuitively, the equation above says that if \(n\) is large enough,
almost all the trees of size \(n\) generated from \(N\) contain
\(S_{\ceil{\cn \log n}}\) as a subcontext.
\begin{rem}
  \label{rem:imf}
The standard infinite monkey theorem says that for any finite word
 \(w=a_1\cdots a_k\) over an alphabet $A$, the probability that a randomly chosen word of length \(n\) contains \(w\) as a subword converges to \(1\), i.e.,
\[\lim_{n\to\infty}\frac{\card{\set{w'\in A^n \mid \exists w_0,w_1\in
      A^*.w'=w_0 w w_1 
}}}{\card{A^n}} = 1.\]
Here, $A^*$ and $A^n$ denote the set of all words over $A$ and the set
 of all words of length $n$ over $A$, respectively.
This may be viewed as a special case of our parameterized infinite monkey theorem above,
where (i) the components of the grammar are given by 
\(\ralph = \set{a\mapsto 1 \mid a\in A}\cup \set{e {\, \mapsto 0}}\), \(\nont = \set{N}\), and 
\(\rules = \set{N\rew a(N)\mid a\in A}\cup \set{N\rew e}\), and (ii) 
\(S_n= a_1(\cdots a_k(\hole)\cdots)\) independently of \(n\).
\end{rem}

Note that in the theorem above, we have used \(\limd\) rather than \(\lim\).
To avoid the use of \(\limd\), we need an additional condition on \(\grm\), called
aperiodicity.
\begin{defi}[Aperiodicity~\cite{Flajolet}]\label{def:aperiodicity}
Let \(\grm\) be a grammar.
For a nonterminal \(N\), \(\grm\) is called \(N\)-\emph{aperiodic} if 
there exists \(n_0\) such that 
\(\card{\inhavn{N}} > 0\) for any \(n\ge n_0\).
Further, \(\grm\) is called \emph{aperiodic} if 
\(\grm\) is \(N\)-aperiodic for any nonterminal \(N\).
\end{defi}
\noindent
In Theorem~\ref{thm:imt},
if we further assume that \(\grm\) is \(N\)-aperiodic,
then \(\limd\limits_{n \rightarrow \infty}\) in the statement can be replaced with
 \(\lim\limits_{n \rightarrow \infty}\).

In the rest of this section, we reformulate the theorem above in the form more convenient for
proving Theorem~\ref{mainthm}. 
In Definition~\ref{df:set-of-terms},
\(\termsn{\D}{\I}{\K}\) was defined
as the disjoint union of \((\termsng{\ee}{\tau}{\D}{\I}{\K})_{\tau \in \types{\D}{\I}}\).
To prove Theorem~\ref{mainthm}, we will construct a grammar \(\grm\) so that, for each \(\tau\),
\(\termsng{\ee}{\tau}{\D}{\I}{\K} \cong \inhav{\grm,N}\) holds for some nonterminal \(N\).
Thus, the following style of statement is more convenient, which is obtained as an immediate corollary of
Theorem~\ref{thm:imt}
and Lemma~\ref{lem:convSum}.
\begin{cor}
\label{cor:imt}
Let $\grm$ be an unambiguous and
\SC{} regular tree grammar such that \(\card{\inhav{\grm}}=\infty\),
and \((S_n)_{n \in \nat}\) be a family of \linear{contexts} in \(\ctx{\grm}\)
such that $|S_n| = \order(n)$.
Then there exists a real constant \(\cn > 0\) such that
for any non-empty set \(\nset\) of nonterminals of \(\grm\) the following equation holds:
\[
 \limd_{n \rightarrow \infty} \frac{\card{\{ (N,\tr) \in \coprod_{N \in \nset} \inhavn{\grm,N} \mid
 S_{\ceil{\cn \log n}} \subc \tr
 \}}}{\card{\coprod_{N \in \nset} \inhavn{\grm,N}}} = 1.
\]
\end{cor}

We can also replace 
a family of \linear{contexts} \((S_n)_n\) above
with that of trees \((T_n)_n\). We need only this corollary in the proof of Theorem~\ref{mainthm}.

\begin{cor}
\label{cor:imtTree}
Let $\grm$ be an unambiguous 
and \SC{} regular tree grammar
such that \(\card{\inhav{\grm}}=\infty\) and
there exists \(C \in \cup_{\f} \inhav{\grm,\f}\) with \(\hn{C}\ge2\).
Also, let \((T_n)_{n \in \nat}\) be a family of trees in 
\(\inhav{\grm}\)
such that $|T_n| = \order(n)$.
Then there exists a real constant \(\cn > 0\) such that
for any non-empty set \(\nset\) of nonterminals of \(\grm\) the following equation holds:
\[
 \limd_{n \rightarrow \infty} \frac{\card{\{ (N,\tr) \in \coprod_{N \in \nset} \inhavn{\grm,N} \mid
 T_{\ceil{\cn \log n}} \subc \tr
 \}}}{\card{\coprod_{N \in \nset} \inhavn{\grm,N}}} = 1.
\]
\end{cor}
\begin{proof}
Let \(\grm=(\ralph,\nont,\rules)\).
Let \(C \in \inhav{\grm,\f}\) be the context in the assumption and let \(\f\) be of the form \(N'_1\cdots
N'_k\Ar N'\), i.e., \(N'\rew^* C[N'_1,\ldots,N'_k]\), where \(k = \hn{C} \ge 2\).
By Corollary~\ref{cor:imt}, it suffices to construct a family of
\linear{contexts} \((S_n)_n\) such that (i) \(S_n\in\ctx{\grm}\),
(ii) \(T_n\subc S_n\), and (iii) $|S_n| = \order(n)$.
 For each \(n\), there exists \(N_n\) such that
\(N_n\rew^* T_n\).
Since \(\card{\inhav{\grm}}=\infty\) and by the strong connectivity, 
for each \(i \le k\) there exists \(T'_i\) such that
\(N'_i \rew^* T'_i\).
For each \(n\), by the strong connectivity, there exists
\(S'_n \) such that \(N_1' \rew^* S'_n[N_n]\),
and since \(\nont\) is finite, we can choose \(S'_n\) so that the size is bounded above by a constant (that is independent of \(n\)). Let \(S_n\) be \(C[S'_n[T_n],\hole,T'_3,\dots,T'_k]\).
Then (i) \(S_n\in \inhav{\grm,N'_2\Ar \new{N'}}\subseteq \ctx{\grm}\) because \(N'\rew^*
C[N'_1,\ldots,N'_k] \rew^* C[S'_n[T_n],N'_2,T'_3,\dots,T'_k]\),
(ii) \(T_n\subc S_n\), and (iii) \(|S_n|=\order(n)\), as required.
\end{proof}

\begin{rem}
\label{rem:avoidLimDef}
In each of Corollaries~\ref{cor:imt} and~\ref{cor:imtTree},
if we further assume that \(\grm\) is \(N\)-aperiodic for any \(N \in \nset\),
then \(\limd\limits_{n \rightarrow \infty}\) in the statement can be replaced with
\(\lim\limits_{n \rightarrow \infty}\).
\end{rem}

\begin{rem}
In Theorem~\ref{mainthm} (and similarly in Corollaries~\ref{cor:imt} and~\ref{cor:imtTree}),
one might be interested in the following form of probability:
\[
 \limd_{n \rightarrow \infty} \frac{\card{\uplus_{m \le n}\{ \tr \in \inhavnn{m}{\grm,N} \mid
 S_{\ceil{\cn \log m}} \subc \tr
 \}}}{\card{\uplus_{m \le n}\inhavnn{m}{\grm,N}}} = 1,
\]
which discusses trees of size \emph{at most} \(n\) rather than \emph{exactly} \(n\).
Under the assumption of aperiodicity,
the above equation follows from Theorem~\ref{mainthm},
by the following \emph{Stolz-Ces\`aro theorem}~\cite[Theorem~1.22]{Marian}:
Let $(a_n)_{n \in \nat}$ and $(b_n)_{n \in \nat}$ be a sequence of real numbers and
assume $b_n > 0$ and \(\sum_{n=0}^{\infty} b_n = \infty\). Then
for any \(c \in [-\infty, +\infty]\),
\[
\lim_{n \rightarrow \infty} \frac{a_n}{b_n} = c
\qquad\text{implies}\qquad
\lim_{n \rightarrow \infty} \fr{\sum_{m\le n}a_m}{\sum_{m\le n}b_m} = c.
\]
(We also remark that the inverse implication also holds if there exists \(\gamma \in \real \setminus \{1\}\) such that
\(\lim_{n \rightarrow \infty} \fr{\sum_{m\le n}b_m}{\sum_{m\le n+1}b_m} = \gamma\)~\cite[Theorem~1.23]{Marian}.)
\end{rem}

 \section{Proof of the Parameterized Infinite Monkey Theorem for Regular Tree Languages (Theorem~\ref{thm:imt})}\label{sec:proofimf}

Here we prove Theorem~\ref{thm:imt}.
In Section~\ref{sec:IMTforWords}, we first prove
a parameterized version\footnote{Although the parameterization is a simple extension, we are not aware of
literature that explicitly states this parameterized version.} of infinite monkey
theorem for \emph{words}, and explain how the proof for the word case can be extended to deal with
regular tree languages. The structure of the rest of the section is explained at the end of Section~\ref{sec:IMTforWords}.

\subsection{Proof for Word Case and Outline of this Section}\label{sec:IMTforWords}
Let $A$ be an alphabet, \ie a finite non-empty set of symbols.
For a word $w = a_1\cdots a_n$ over \(A\),
We write $|w|$ for \(n\) and call it the \emph{size} (or length) of $w$.
As usual, we denote by $A^n$ the set of all words of size $n$ over $A$,
and by $A^*$ the set of all finite words over $A$: $A^* = \bigcup_{n \geq 0} A^n$. 
For two words $w, w' \in A^*$, 
we say $w'$ is a \emph{subword of $w$} and write $w' \subw w$ if $w = w_1 w' w_2$ for
some words $w_1, w_2 \in A^*$.
The infinite monkey theorem
states that, for any word $w \in A^*$,
 the probability that a randomly chosen word of size $n$ contains $w$ as
 a subword tends to one if $n$ tends to infinity (recall Remark~\ref{rem:imf}).
 The following theorem is 
a parameterized version, where \(w\) may depend on \(n\).
It may also be viewed as a special case of Theorem~\ref{thm:imt}, where
\(\ralph = \set{a\mapsto 1 \mid a\in A}\cup \set{e}\), \(\nont = \set{N}\), 
\(\rules = \set{N\rew a(N)\mid a\in A}\cup \set{N\rew e}\), and 
\(S_n= a_1(\cdots a_k(\hole)\cdots)\) for each \(w_n = a_1\cdots a_k\).
\begin{prop}[Parameterized Infinite Monkey Theorem for Words]\label{prop:monkey}
Let $A$ be an alphabet and $(w_n)_n$ be a family of words over $A$ such
 that $|w_n|= \order(n)$.
Then, there exists a real constant \(p>0\) such that we have:
\[
 \lim_{n \rightarrow \infty} \fr{\card{\{ w \in A^n \mid w_{\ceil{p \log n}} \subw w\}}}{\card{A^n}} = 1.
\]
\end{prop}
\begin{proof}\upshape
Let $\zps{n}$ be $1 - \card{\{ w \in A^n \mid w_{\ceil{p \log n}} \subw w \}}/\card{A^n}$,
i.e., the probability that a word of size $n$ does \emph{not} contain $w_{\ceil{p \log n}}$.
By the assumption \(|w_n|=O(n)\),
there exists \(b>0\) such that \(|w_n| \le bn\) for sufficiently large \(n\).
Let \(0<q<1\) be an arbitrary real number, and we define
\(p \defeq \frac{q}{2b\log{\card{A}}}\).
We write $s(n)$ for \(|w_{\ceil{p \log n}}|\) 
and $c(n)$ for 
$\floor{n/s(n)}$.
Let \(n \in \nat\).
Given a word $w = a_1\cdots a_n\in A^n$, let us decompose it to subwords of length $s(n)$ as follows.
\[
w = \underbrace{a_1 \cdots a_{s(n)}}_{\text{first subword}}
 \cdots
 \underbrace{a_{(c(n)-1) s(n)+1} \cdots a_{c(n)s(n)}}_{c(n)\text{-th subword}} 
 \underbrace{a_{c(n) s(n)+1} \cdots a_{n}}_{\text{remainder}}.
\]
Then, 
\begin{align*}
0 \le \zps{n} &\leq \mbox{``the probability that none of the $i$-th subword is $w_{\ceil{p \log n}}$''}
\\ &= 
\fr{\card{(A^{s(n)} \setminus \{w_{\ceil{p \log n}}\})^{c(n)}A^{n-c(n)s(n)}}}{\card{A^{n}}}
\\ &= 
\left(\fr{\card{A^{s(n)} \setminus \{w_{\ceil{p \log n}}\}}}{\card{A^{s(n)}}}
		  \right)^{c(n)} =
\left(\fr{\card{A^{s(n)}} - 1}{\card{A^{s(n)}}}
		  \right)^{c(n)} =
 \left(1-\fr{1}{\card{A}^{s(n)}}\right)^{c(n)}
\end{align*}
and we can show \(\lim_{n \rightarrow \infty} \left(1-\fr{1}{\card{A}^{s(n)}}\right)^{c(n)} = 0\)
as follows. 

For sufficiently large \(n\),
we have \(s(n) \le b\ceil{p \log n} \le
2bp\log n = q\fr{\log n}{\log {\card{A}}} = q\log_{\card{A}}n\),
and hence
\[
\left(1-\fr{1}{\card{A}^{s(n)}}\right)^{c(n)} 
\le \left(1-\fr{1}{n^q}\right)^{c(n)}
\le \left(1-\fr{1}{n^q}\right)^{\fr{n}{s(n)}-1}
\le \left(1-\fr{1}{n^q}\right)^{\fr{\log \card{A}}{q}\fr{n}{\log n}-1}.
\]
Also we have \(\lim_{n \rightarrow \infty}  \left(1-\fr{1}{n^{q}}\right)^{\fr{n}{\log n}} = 0\) (which we leave for the reader to check);
therefore \(\lim_{n \rightarrow \infty} \zps{n} = 0\).
\end{proof}

The key observations in the above proof were:
\begin{itemize}
  
\item[(W1)] Each word \(w=a_1\cdots a_n\) can be decomposed to
  \[(w_{\mathtt{rem}}', w'_1,\ldots,w'_{\new{c}(n)})\in 
  A^{n-c(n)s(n)}\times \Big(\prod_{i=1}^{c(n)} A^{s(n)}\Big) \] 
where \(w_{\mathtt{rem}}'=a_{c(n) s(n)+1} \cdots a_{n}\) and \(w'_i = a_{(i-1)s(n)+1}\cdots a_{is(n)}\).

\item[(W2)]
The word decomposition above induces the following
  decomposition of the set \(A^n\) of words of length \(n\): 
\begin{align*}
A^n \ &\cong\ A^{n-c(n)s(n)}\times \Big(\prod_{i=1}^{c(n)} A^{s(n)}\Big) 
\\
 \ &\cong\
\coprod_{w \in A^{n-c(n)s(n)}} \prod_{i=1}^{c(n)} A^{s(n)}
\end{align*}
Here, \(\cong\) denotes the existence of a bijection.
\item[(W3)] Choose \(s(n)\) and \(c(n)\) so that
(i) \(A^{s(n)}\) contains at least one element that contains \(w_n\) as a subword
and (ii) \(c(n)\) is sufficiently large. By condition (i), the probability that
an element of \(A^n\) does not contain \(w\) as a subword can be bounded above
the probability that none of \(w'_i\, (i\in\set{1,\ldots,c(n)})\) contains \(w_n\) as a subword,
i.e., \((1-\frac{1}{\card{A^{s(n)}}})^{c(n)}\), which converges to \(0\) by condition (ii).
\end{itemize}
The proof of Theorem~\ref{thm:imt} in the rest of this section is based on similar observations:
\begin{itemize}
\item[(T1)] Each tree \(T\) of size \(n\) can be decomposed to
  \[(E, U_1,\ldots,U_{c_E}),\]
  where \(U_1,\ldots,U_{c_E}\) are (affine) subcontexts, 
  \(E\), called a \emph{second-order context} (which will be formally defined later),
  is the ``remainder'' of \(T\) obtained by extracting \(U_1,\ldots,U_{c_E}\), and
  \(c_E\) is a number that depends on \(E\).
  For example,  the tree on the lefthand side of Figure~\ref{fig:partitioning}
  can be decomposed to the second-order context and affine contexts
  shown on the righthand side. By substituting each affine context for
  \(\Hhole{}{}\) in the preorder traversal order, we recover the tree on the lefthand side.
  This decomposition of a tree may be regarded as a generalization of the word decomposition
  above (by viewing a word as an affine context), where the part \(E\)
corresponds to the remainder
\(a_{c(n) s(n)+1} \cdots a_{n}\) of the word decomposition.
\item[(T2)] The tree decomposition above induces
the decomposition of
    the set \(\inhavn{\grm,N}\) in the following form:
\begin{align}
\inhavn{\grm,N} \cong \coprod_{E\in \mathcal{E}} \prod_{j=1}^{c_E} \mathcal{U}_{E,j}
\label{eq:bijection}
\end{align}
where
\(\mathcal{E}\) is a set of second-order contexts and
\(\mathcal{U}_{E,j}\) is a set of \affine{contexts}.
(At this point, the reader need not be concerned about
the exact definitions of \(\mathcal{E}\) and
\(\mathcal{U}_{E,j}\), which will be given later.)
\item[(T3)] \new{Design} the above decomposition so that
  (i) each \(\mathcal{U}_{E,j}\) contains at least one element that
  contains \(S_{\ceil{\cn \log n}}\) as a subcontext,
  and (ii) \(c_E\)  is sufficiently large.
By condition (i),
the probability that an element of \(\inhavn{\grm,N}\) does not contain \(S_{\ceil{\cn \log n}}\)
as a subcontext is bounded above by 
the probability that none of \(U_i\,(i\in\set{1,\ldots,c_E})\)  
contain \(S_{\ceil{\cn \log n}}\) as a subcontext,
which is further bounded above by
\(\displaystyle{\max_{E\in \mathcal{E}}\Big(\prod_{j=1}^{c_E}(1-\frac{1}{\card{\mathcal{U}_{E,j}}})\Big)}\).
The bound can be proved to converge to \(0\) by using condition (ii).
\end{itemize}

\begin{figure}[tbp]
\centering\includegraphics[scale=0.4,bb=0 0 842 391]{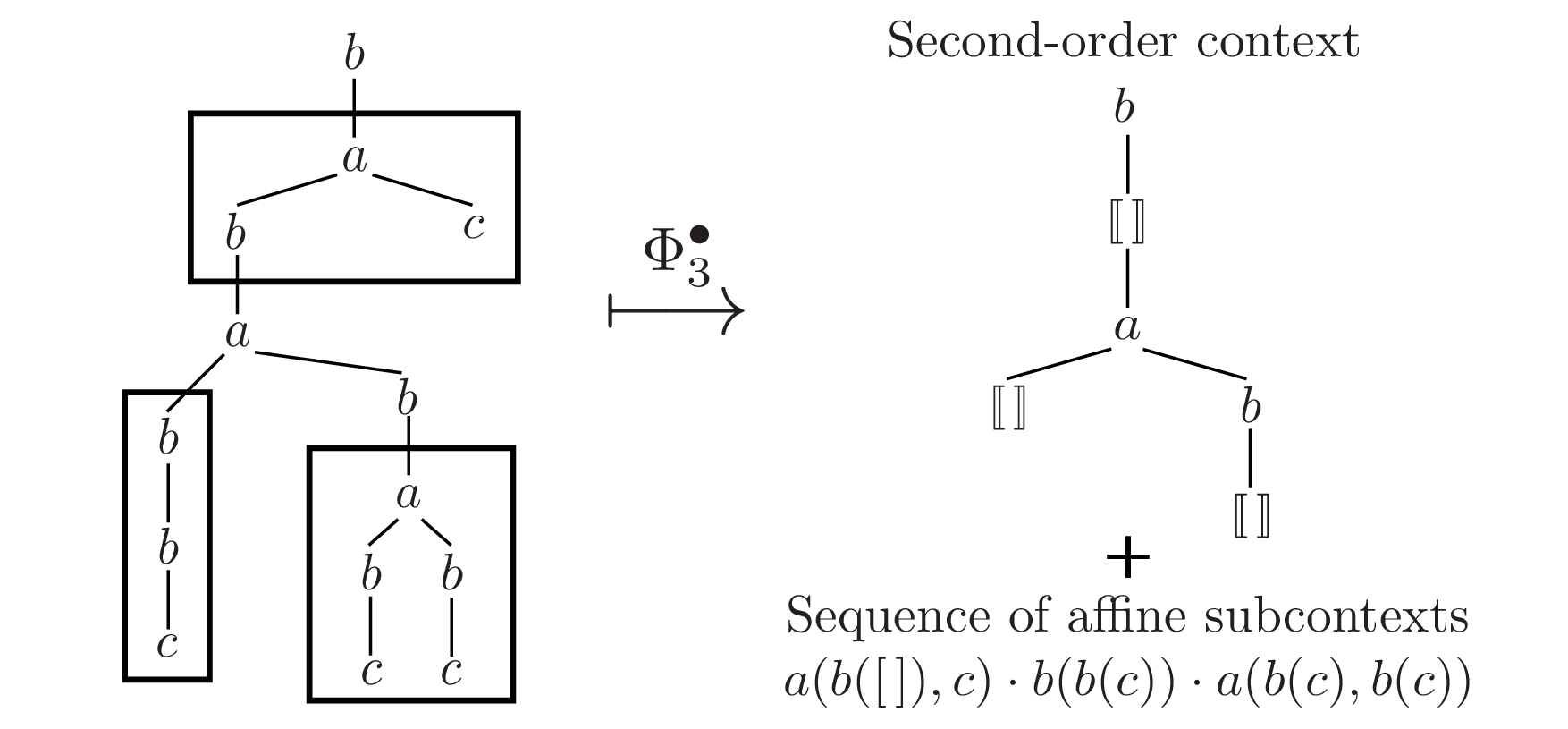}
\caption{An example of tree decomposition. The parts
 surrounded by rectangles
 on the lefthand side show the extracted affine contexts, and the remaining part of
 the tree is
 the second-order tree context.}
\label{fig:partitioning}
\end{figure}

The rest of the section is organized as follows.
Before we jump to the decomposition of \(\inhavn{\grm,N}\), we
first present 
a decomposition of \(\TRn{n}{\ralph}\), i.e., the set of \(\ralph\)-trees of size \(n\)
in Section~\ref{sec:decomposition}. The decomposition of
\(\TRn{n}{\ralph}\) may be a special case of that of \(\inhavn{\grm,N}\), where
\(\grm\) generates all the trees in \(\TRn{n}{\ralph}\).
For a technical convenience in extending the decomposition of
\(\TRn{n}{\ralph}\) to that of \(\inhavn{\grm,N}\), 
we normalize grammars to \emph{canonical form}
in Section~\ref{sec:canonical-grammar}. We then give the decomposition of
\(\inhavn{\grm,N}\) given in Equation (\ref{eq:bijection}) above,
and prove that it satisfies the required properties
 in Sections~\ref{sec:rtl-decomposition} and 
\ref{sec:adjusting}.
Finally, we prove Theorem~\ref{thm:imt} in Section~\ref{sec:proofimfmain}.

\subsection{Grammar-Independent Tree Decomposition}
\label{sec:decomposition}
In this subsection, we will define a decomposition
function $\closeddecomp$ (where $m$ is a parameter) that decomposes a tree $T$
into (i) a (sufficiently long) sequence $P=U_1\cdots U_k$ 
consisting of \affine{subcontexts} of size no less than \(m\), 
and (ii) a ``second-order'' context \(E\) (defined 
shortly in Section~\ref{sec:second-order-context}), which is the remainder of extracting \(P\) from \(T\).
Recall Figure~\ref{fig:partitioning}, which illustrates how a tree is decomposed by
$\closeddecompm{3}$. Here, the symbol $\hhole$ in the second-order context on the
right-hand side represents the original position of each subcontext.
By filling the \(i\)-th occurrence (counted in the depth-first, left-to-right pre-order) 
of \(\hhole\) with
the \(i\)-th \affine{context}, we can recover the original tree on the left hand side.

As formally stated later (in Corollary~\ref{cor:gindependent-tree-decomposition}),
the decomposition function \(\closeddecomp\) provides a witness for a bijection of the form
\[
\TRn{n}{\ralph} \cong \coprod_{E\in \mathcal{E}} \prod_{j=1}^{c_E} \mathcal{U}_{E,j},
\]
a special case of Equation \eqref{eq:bijection} mentioned in Section~\ref{sec:IMTforWords}.

\subsubsection{Second-Order Contexts}
\label{sec:second-order-context}
We first define the notion of second-order contexts and operations on them.

The set of \emph{second-order contexts} over $\ralph$,
ranged over by $E$, is defined by:
\[
E ::= \Hhole{k}{n}[E_1, \ldots, E_k] \;\mid\; a(E_1, \ldots,
E_{\ralph(a)}) \ \ (a \in \dom(\ralph)).
\]  
Intuitively, the second-order context is an expression having holes of
the form \(\Hhole{k}{n}\) (called \emph{second-order holes}), which should be filled with a \(k\)-context of size \( n \). 
By filling all the second-order holes, we obtain a \(\ralph\)-tree.
Note that
\(k\) may be \(0\). In the technical development below,
we only consider second-order holes \(\Hhole{k}{n}\) such that
\(k\) is  \(0\) or \(1\).
We write \(\shn(E)\) for the number of the second-order holes in \(E\).
Note that $\ralph$-trees can be regarded as second-order contexts \(E\) such that $\shn(E) = 0$, and vice versa.
For \(i \le \shn(E)\), we write \(\sh{E}{i}\) for the \(i\)-th second-order hole
(counted in the depth-first, left-to-right pre-order).
We define the \emph{size} \(|E|\) by:
\(|\hhole^n_{k}[E_1,\dots,E_k] | \defeq n + |E_1|+\cdots+|E_k|\)
and
\(|a(E_1,\dots,E_{\ralph(a)})| \defeq |E_1|+\cdots+|E_{\ralph(a)}| + 1\).
Note that \(|E|\) includes the size of contexts to fill the second-order holes in \(E\).
\begin{exa}
\label{ex:decomp}
The second-order context on the right hand side of Figure~\ref{fig:partitioning}
is expressed as \(E = \Tb(\Hhole{1}{3}[\Ta(\Hhole{0}{3}, \Tb(\Hhole{0}{5}))])\), where \(\shn(E)=3\),
\(|E|=14\), \(\sh{E}{1}=\Hhole{1}{3}\), \(\sh{E}{2}=\Hhole{0}{3}\), and \(\sh{E}{3}=\Hhole{0}{5}\).
\end{exa}

Next we define the substitution operation on second-order contexts.
For a context \(\pct\) and a second-order hole \(\hhole^n_k\),
we write \(\cbf{\pct}{\hhole^n_k}\) if \(\pct\) is a \(k\)-context of size \( n \).
Given \(E\) and \(\pct\) such that \(\shn(E) \ge 1\) and
\(\cbf{\pct}{\sh{E}{1}}\), we write $E\Hfill{\pct}$ for the second-order context obtained by replacing the leftmost second-order hole of $E$ (\ie $\sh{E}{1}$)
with $\pct$ (and by interpreting the syntactical bracket \([-]\) as the substitution operation).
Formally, it is defined by induction on $E$ as follows:
\[
\begin{array}{lcl}
\left(\hhole^n_k[E_1, \ldots, E_k]\right)
\Hfill{\pct} &\defeq&
\pct[E_1, \ldots, E_k]
\\
 (a(E_1, \ldots, E_{\ralph(a)})) \Hfill{\pct} &\defeq&
 a(E_1, \ldots, E_i \Hfill{\pct}, \ldots, E_{\ralph(a)})\\&&\hfill
 \text{ where } i = \min\{ j \mid \shn(E_j) \geq 1, 1 \leq j \leq \ralph(a) \}.
\end{array}
\]
In the first clause, $C[E_1,\ldots,E_k]$ is the second-order context obtained by
replacing $\hole_i$ in $C$ with $E_i$ for each $i \leq k$.
Note that 
we have \(|E\Hfill{\pct}| = |E|\) whenever \(E\Hfill{\pct}\) is well-defined,
i.e.,~if \(\cbf{\pct}{\sh{E}{1}}\) (cf. Lemma~\ref{lem:second-substitution-preserves-size} below).

We extend the substitution operation for a sequence of contexts.
We use metavariable \(P\) for sequences of contexts.
For $E$ and $P = \pct_1 \pct_2 \cdots \pct_{\shn(E)}$, we
write $\cbf{P}{E}$ if $\cbf{\pct_i}{\sh{E}{i}}$ for each $i \leq \shn(E)$.
Given \(E\) and a sequence of contexts \(P=\pct_1 \pct_2\cdots
\pct_\ell\) such that 
\(\ell \le \shn(E)\) and \(\cbf{\pct_i}{\sh{E}{i}}\) for each \(i \le \ell\), 
we define \(E\Hfill{P}\) by induction on \(P\):\tk{Is the partial application of length \( \ge 2 \) used anywhere?  I think it's not.}
\begin{align*}
&E\Hfill{\epsilon} \defeq E
&
&E\Hfill{\pct \cdot P} \defeq (E\Hfill{\pct})\Hfill{P}
\end{align*}
Note that \(\shn(E\Hfill{\pct}) = \shn(E) - 1\), so if \(\cbf{P}{E}\)
then \(E\Hfill{P}\) is a tree.

\begin{exa}
\label{ex:decomp-cont}
Recall the second-order context \(E = \Tb(\Hhole{1}{3}[\Ta(\Hhole{0}{3}, \Tb(\Hhole{0}{5}))])\) in
Figure~\ref{fig:partitioning} and Example~\ref{ex:decomp}.
Let \(P\) be the sequence of \affine{contexts} given in Figure~\ref{fig:partitioning}:
\[
\Ta(\Tb(\hole),\Tc)\cdot \Tb(\Tb(\Tc)) \cdot \Ta(\Tb(\Tc), \Tb(\Tc)).
\]
Then 
\[E\Hfill{\prj{P}{1}} = \Tb(\prj{P}{1}[\Ta(\Hhole03, \Tb(\Hhole05))]) = \Tb(\Ta(\Tb(\Ta(\Hhole{0}{3}, \Tb(\Hhole{0}{5}))), \Tc)),\]
and
\begin{align*}
E\Hfill{P} &= ((E\Hfill{\prj{P}{1}}])\Hfill{\prj{P}{2}})\Hfill{\prj{P}{3}}\\
&= \Tb(\Ta(\Tb(\Ta(\prj{P}{2}, \Tb(\prj{P}{3}))), \Tc)) = 
\Tb(\Ta(\Tb(\Ta(\Tb(\Tb(\Tc)), \Tb(\Ta(\Tb(\Tc), \Tb(\Tc))))), \Tc)),
\end{align*}
which is the tree shown on
the left hand side of Figure~\ref{fig:partitioning}.
\end{exa}

Thanks to the size annotation, the substitution operation preserves the size of a second-order context.
\begin{lem}\label{lem:second-substitution-preserves-size}
  Let \( E \) be a second-order context and \( P \) be a sequence of contexts such that \(P\COL E\).
  Then \( |E| = |E\Hfill{P}| \).
\end{lem}
\begin{proof}
The proof is given by a straightforward induction on \( E \). In the case \(E=\Hhole{k}{n}[E_1, \ldots, E_k]\),
we use the fact \(|C[E_1,\ldots,E_k]| = |C|+|E_1|+\cdots +|E_k|\) (which can be shown by induction on \(C\)).
\end{proof}

\subsubsection{Grammar-Independent Decomposition Function}
\label{sec:gindependent-decomposition}
Now we define the decomposition function \(\closeddecomp\).
Let \(\ralph\) be a ranked alphabet and \(m \ge 1\).
We define \(\mr_{\ralph} \defeq \max(\image{\ralph})\), which is also written as \(\mr\) for short.
We always assume that \(\mr \ge 1\), 
since it holds whenever \(\card{\inhav{\grm}}=\infty\),
which is an assumption of Theorem~\ref{thm:imt}. 
We shall define the decomposition function \(\closeddecompp{\ralph}\)
(we omit $\ralph$ and just write \(\closeddecomp\) for \(\closeddecompp{\ralph}\)below)
so that $\closeddecomp(T) = (E,P)$ where 
(i) \(E\) is the second-order context, and (ii)
\(P\) is a sequence of \affine{contexts},  (iii) \(E\Hfill{P} = T\),
and (iv) \(m \le |\prj{P}{i}| \le \mr(m-1)+1\) for each \(i\in \set{1,\ldots,\card{P}}\).

The function \(\closeddecomp\) is defined as follows, using an auxiliary decomposition function
\(\decomp\) given below.
\[
\closeddecomp(T) \defeq (U[E], P)\mbox{ where }(U,E,P)=\decomp(T).
\]
\noindent
The auxiliary decomposition function (just called ``decomposition function'' below) \(\decomp\)
traverses a given tree \(T\) in a bottom-up manner, 
extracts a sequence \(P\) of subcontexts,
and returns it along with a \linear{context} \(U\) and a second-order context \(E\); 
\(U\) and \(E\) together represent
the ``remainder'' of extracting \(P\) from \(T\).
During the bottom-up traversal, the \(U\)-component for each subtree \(T'\) represents a context
containing the root of \(T'\); 
whether it is extracted as (a part of) a subcontext or becomes a part of the remainder will be decided later based on the surrounding context. 
The \(E\)-component will stay as a part of the remainder during
the decomposition of the whole tree (unless the current subtree \(T'\) is too small, i.e., \(|T'|<m\)).

We define \(\decomp\) by:
\begin{itemize}
\item If \(|T|<m\), then \(\decomp(T) \defeq(\hole, T, \NIL)\).
\item If \(|T|\geq m\), $T = a(T_1, \ldots, T_{\ralph(a)})$, and
\(\decomp(T_i) = (\pU_i, E_i, P_i) \ (\mbox{for each }i \le \ralph(a))\), then:
\begin{equation}
\label{eq:defOfDecomp}
\decomp(T) \defeq 
\left\{
\begin{aligned}
& (\hole,\; a\bigl(\pU_1[E_1], \dots, \pU_{\ralph(a)}[E_{\ralph(a)}] \bigr),\; P_1 \cdots P_{\ralph(a)})\\
& \quad \mbox{if there exist $i,j$ such that \(1\leq i<j \le \ralph(a)\) 
and \(|T_i|,|T_j|\geq m\)}\\
& (\hole,\; \hhole^n_1 [E_i],\; a(T_1, \ldots, U_i, \ldots, T_{\ralph(a)}) \cdot P_i)\\
&  \quad \mbox{if $|T_j|<m$ for every \(j\neq i\), $|T_i| \geq m$, and}\\
&  \quad \mbox{\phantom{if} $ n \defeq |a(T_1, \ldots, U_i, \ldots,
 T_{\ralph(a)})| \geq m$}\\ 
& (a(T_1, \ldots, U_i, \ldots, T_{\ralph(a)}),\; E_i,\; P_i)\\
&  \quad \mbox{if $|T_j|<m$ for every \(j\neq i\), $|T_i| \geq m$, and}\\
&  \quad \mbox{\phantom{if} $|a(T_1, \ldots, U_i, \ldots,
 T_{\ralph(a)})| < m$} \\%
& (\hole,\; \hhole^{n}_0,\; T)\\
&  \quad \mbox{if $|T_i| < m$ for every $i \le \ralph(a)$, and \( n \defeq |T|\) }
\end{aligned}
		\right.
\end{equation}
\end{itemize}
As defined above, the decomposition is carried out by case analysis on the size of a given tree.
If \(T\) is not large enough (i.e., \(|T|<m\)), then \(\decomp(T)\) returns
an empty sequence of contexts, while keeping \(T\) in the \(E\)-component.
If \(|T|\geq m\), then \(\decomp(T)\) returns a non-empty sequence of contexts, by case analysis on the
sizes of \(T\)'s subtrees. If there are more than one subtree whose size is no less than \(m\) (the first case above), 
then \(\decomp(T)\) concatenates the sequences of contexts extracted from the subtrees, and returns the remainder
as the second-order context. If only one of the subtrees, say \(T_i\), is large enough (the second and third cases), then 
it basically returns the sequence \(P_i\) extracted from \(T_i\); however, if the remaining part
\(a(T_1, \ldots, U_i, \ldots,
 T_{\ralph(a)})\) is also large enough, then it is added to the sequence (the second case).
If none of the subtrees is large enough (but \(T\) is large enough), then \(T\) is returned as the \(P\)-component
 (the last case).

\begin{exa}
  \label{ex:decomposition-function}
Recall Figure~\ref{fig:partitioning}. Let \(T_0\) be the tree on the left hand side.
For some of the subtrees of \(T_0\),
\(\decompm{3}\) can be calculated as follows.
\[
\begin{array}{l}
\decompm{3}(\Tb(\Tc))=(\hole, \Tb(\Tc), \epsilon)\\
\decompm{3}(\Tb(\Tb(\Tc))) = (\hole, \Hhole{0}{3}, \Tb(\Tb(\Tc)))\hfill \mbox{ (by the last case of~\eqref{eq:defOfDecomp})}\\
\decompm{3}(\Ta(\Tb(\Tc),\Tb(\Tc))) = 
(\hole, \Hhole{0}{5}, \Ta(\Tb(\Tc), \Tb(\Tc))) \qquad \hfill\mbox{ (by the last case of~\eqref{eq:defOfDecomp})}\\
\decompm{3}(\Ta(\Tb(\Tb(\Tc)), \Tb(\cdots))) =
  (\hole, \Ta(\Hhole{0}{3}, \Tb(\Hhole{0}{5})),  \Tb(\Tb(\Tc))\cdot \Ta(\Tb(\Tc), \Tb(\Tc)))\\
\hfill \mbox{ (by the first case of~\eqref{eq:defOfDecomp})}\\
\decompm{3}(T_0) = (\Tb(\hole), \Hhole{1}{3}[\Ta(\Hhole{0}{3}, \Tb(\Hhole05))], 
\Ta(\Tb(\hole),\Tc)\cdot \Tb(\Tb(\Tc)) \cdot \Ta(\Tb(\Tc), \Tb(\Tc)))\\
\hfill \mbox{ (by the third case of~\eqref{eq:defOfDecomp})}\\
\end{array}
\]
From \(\decompm{3}(T_0)\) above, we obtain:
\[\closeddecompm{3}(T_0) = \left(\Tb(\Hhole13[\Ta(\Hhole03, \Tb(\Hhole05))]),\ \Ta(\Tb(\hole),\Tc)\cdot \Tb(\Tb(\Tc)) \cdot \Ta(\Tb(\Tc), \Tb(\Tc))\right).\]
\end{exa}

\subsubsection{Properties of the Decomposition Function}
\label{sec:gindependent-decomposition-properties}
We summarize important properties of \(\decomp\) in this subsection.

We say that \anaffine{context} $\pU$ is \emph{good for $m$}
if
\(|U| \ge m\) and 
$\pU$ is of the form \(a(\pU_1,\dots,\pU_{\ralph(a)})\)
where \(|\pU_i| < m\) for each \(i \le \ralph(a)\).
In other words, \(U\) is good if \(U\) is of an appropriate size:
it is large enough (i.e. \(|\pU|\ge m\)), and not too large (i.e. the size
of any proper subterm is less than \(m\)).
For example, \(\Ta(\Tb(\hole),\Tb(\Tc))\) is good for \(3\), but 
neither \(\Tb(\Tb(\hole))\) nor \(\Ta(\Tb(\hole),\Tb(\Tb(\Tc)))\) is.

The following are basic properties of the auxiliary decomposition function
\(\decomp\).
The property (\ref{charphi:split}) says that the original tree can be
recovered by composing the elements obtained by the decomposition,
and the property 
(\ref{charphi:componentsize})
ensures that \(\closeddecomp(T)\) extracts only good contexts from \(T\).

\begin{lem}\label{lemma:charphi}
  \tk{to do: check which property is used}
Let T be a tree. If $\decomp(T) = (U,E,P)$, then:
\begin{enumerate}
\item \label{charphi:split}
\(\cbf{P}{E}\),\, \(\hn{U}=1\), and \((U[E])\Hfill{P}=T\).
\item \label{charphi:garbagesize}
\(|\pU|<m\).
\item \label{charphi:componentsize}
For each \(i \in \set{1,\dots,\len{P}}\),
\(\prj{P}{i}\) is good for \(m\).
\item \label{charphi:size-of-e}\(|T| = |U[E]|\).
\end{enumerate}
\end{lem}
\begin{proof}
(\ref{charphi:split})--(\ref{charphi:componentsize}) follow by straightforward simultaneous induction on $|T|$.
(\ref{charphi:size-of-e}) follows from (\ref{charphi:split}) and
Lemma~\ref{lem:second-substitution-preserves-size}.
\end{proof}

The following lemma ensures that \(\shn(E)\) is sufficiently large whenever \(\closeddecomp(T)=(E,P)\).
Recall condition (ii) of (T3) in Section~\ref{sec:IMTforWords}; \(\shn(E)\) corresponds to \(c_E\).
 \begin{lem}
  \label{lemma:partition}
For any tree $T$ and $m$ such that $1 \leq m \leq |T|$, 
if 
$\decomp(T) = (U,E,P)$, then
 \[
 \len{P} \geq \frac{|T|}{2\mr m}.
 \]
 \end{lem}
\begin{proof}
Recall that \(\mr \defeq \max(\image{\ralph})\) and we assume \(\mr \ge 1\).
We show below that
\[
\begin{aligned}
|U| + 2\mr m\len{P}- \mr m \ge |T|
\end{aligned}
\]
by induction on \(T\).
Then it follows that
\[
 2\mr m\len{P} \ge |T|+ \mr m - |U|>|T|+ \mr m-m \geq |T|
\]
(where the second inequality uses Lemma~\ref{lemma:charphi}(\ref{charphi:garbagesize})),
which implies \(\len{P} > \fr{|T|}{2\mr m}\) as required.

Since \(|T| \ge m\), $\decomp(T)$ is computed by~Equation~\eqref{eq:defOfDecomp},
on which we perform a case analysis.
Let $T = a(T_1, \dots, T_{\ralph(a)})$ and \(\decomp(T_i) = (\pU_i, E_i, P_i) \ (i \le \ralph(a))\).
\begin{itemize}
 \item The first case: 
In this case, we have
\[
U = \hole \qquad P = P_{i_1} P_{i_2} \cdots P_{i_s}
\]
where $\{i_1, \dots, i_s\} = \{i \le\ralph(a) \mid |T_i| \geq m \}$ and $s \geq 2$,
since if \(|T_i|<m\) then \(P_i = \NIL\).
Note that we have \(\mr \ge 2\) in this case.
Then,
\begin{align*}&
|U|+2\mr m\len{P}-\mr m
\\ =\ &
2\mr m \left(\textstyle\sum_{j \le s} \len{P_{i_j}}\right) -\mr m
\\ \ge\ &
\left(\textstyle\sum_{j \le s}
(|T_{i_j}| + \mr m - |U_{i_j}|)\right) - \mr m 
\qquad(\because \text{by induction hypothesis for \(T_{i_j}\)})
\\ \ge\ &
\left(\textstyle\sum_{j \le s} (|T_{i_j}| + \mr m - (m-1))\right)   -\mr m
\qquad(\because |U_{i_j}|\le m-1 \mbox{ by Lemma~\ref{lemma:charphi}(\ref{charphi:garbagesize})})
\\ =\ &
\left(\textstyle\sum_{j \le s} |T_{i_j}|\right) + s \mr m - s m + s -\mr m
\\ \ge\ &
\left(\textstyle\sum_{j \le s} |T_{i_j}|\right) + 2 \mr m - s m + s -\mr m -r +1
\qquad(\because \text{\(s \ge 2\) and \(\mr \ge 
  1\)})
\\ =\ &
\left(\textstyle\sum_{j \le s} |T_{i_j}|\right) + (r-s)(m-1)+1
\\ \ge\ &
\left(\textstyle\sum_{j \le s} |T_{i_j}|\right) + (\ralph(a)-s)(m-1)+1
\qquad(\because r \ge \ralph(a))
\\ \ge\ & 
\left(\textstyle\sum_{j \le s} |T_{i_j}|\right) + (\textstyle\sum_{i\in \set{1,\ldots,\ralph(a)}\setminus\set{i_1,\ldots,i_s}}|T_i|)+1\\ &
\qquad(\because m-1 \ge |T_i| \mbox{ for }i\in \set{1,\ldots,\ralph(a)}\setminus\set{i_1,\ldots,i_s})
\\ =\ & 
|T|
\end{align*}
   as required.
 \item The second case: 
In this case, we have
\[
U = \hole
\qquad
P = a(T_1, \ldots, U_i, \ldots, T_{\ralph(a)}) P_i
\]
and \(|T_{j}| \ge m\) if and only if \(j=i\).
Also we have \(r \ge 1\).
Then,
\begin{align*}&
|U|+2\mr m\len{P}-\mr m
\\ =\ &
2\mr m(1+\len{P_i}) - \mr m
\\ =\ &
\mr m + 2\mr m \len{P_i}
\\ \ge\ &
\mr m + (|T_i| +\mr m -|U_i|)
\qquad(\because \text{by induction hypothesis for \(T_{i}\)})
\\ \ge\ &
|T_i| +\mr m -(m-1)
\qquad(\because \mr m \ge 0, |U_i| \le m-1\mbox{ by Lemma~\ref{lemma:charphi}(\ref{charphi:garbagesize})})
\\ \ge\ &
|T_i|+\mr m -\mr -m + 2
\qquad(\because \mr \ge 1)
\\ =\ &
|T_i|+(\mr-1)(m-1)+1
\\ \ge\ &
|T_i|+(\ralph(a)-1)(m-1)+1 \qquad(\because r \ge \ralph(a))
\\ \ge\ &
|T_i|+(\textstyle\sum_{j\in \set{1,\ldots,\ralph(a)}\setminus\set{i}}|T_j|)+1 
\\ =\ &
|T|
\end{align*}
as required.

 \item The third case of~\eqref{eq:defOfDecomp}:
In this case, we have
\[
U = a(T_1, \ldots, U_i, \ldots, T_{\ralph(a)})
\qquad
P = P_i
\]
and \(|T_{j}| \ge m\) if and only if \(j=i\).
Then,
\begin{align*}&
|U|+2\mr m\len{P}-\mr m
\\ =\ &
1+(\textstyle\sum_{j\neq i}|T_j|)+|U_i|+2\mr m \len{P_i} -\mr m
\\ \ge\ &
1+(\textstyle\sum_{j\neq i}|T_j|)+|T_i|
\qquad(\because \text{by induction hypothesis for \(T_{i}\)})
\\ =\ &
|T|
\end{align*}
as required.

 \item The fourth case of~\eqref{eq:defOfDecomp}:
In this case, we have
\[
U = \hole
\qquad
P = T.
\]
Then
\begin{align*}&
|U|+2\mr m\len{P}-\mr m
\\ =\ &
\mr m
\\ \ge\ &
\mr(m-1)+1
\qquad(\because r \ge 1) 
\\ \ge\ &
\ralph(a)(m-1)+1
\qquad(\because r \ge \ralph(a))
\\ \ge\ &
(\textstyle\sum_{j\in \set{1,\ldots,\ralph(a)}}|T_j|)+1
\\ =\ &
|T|
\end{align*}
as required.
\qedhere
\end{itemize}
\end{proof}

\subsubsection{Decomposition of \(\TRn{n}{\ralph}\)}
\label{sec:decomp-trees}
This subsection shows that 
the decomposition function \(\closeddecomp\) above provides a witness for a bijection of the form
\[
\TRn{n}{\ralph} \cong \coprod_{E\in \mathcal{E}} \prod_{j=1}^{c_{E}} \mathcal{U}_{E,\,j}.
\]
\asd{Here I removed a garbage: ``where \(\ell_i\) is sufficiently large for each \(i\).''}

We prepare some definitions to precisely state the bijection.
We define the set \(\uframe\) of second-order contexts and the set \(\ucomp\) of affine contexts
by:
\begin{align*}
  \uframe &\defeq \{ E \mid (E, P) = \closeddecomp(T) \text{ for some } T \in \TRn{n}{\ralph} \text{ and } P \}\\
  \ucomp &\defeq \set{ U \mid U : \hhole_k^n, U \text{ is good for } m }.
\end{align*}
Intuitively, \(\uframe\) is the set of second-order contexts obtained by decomposing a tree of
size \(n\), and \(\ucomp\) is the set of good contexts that match the second-order context 
\(\hhole_k^n\).

The bijection is then stated as the following lemma.
\begin{lem}
\label{cor:gindependent-tree-decomposition}
\begin{align}
  \TRn{n}{\ralph}\;\cong
  \coprod_{E \in \uframe} \prod_{i=1}^{\shn(E)} \ucomp[\sh{E}{i}].
\label{eq:bijection-for-trees}
  \end{align}
\end{lem}

The rest of this subsection is devoted to a proof of the lemma above;
readers may wish to skip the rest of this subsection upon the first reading.

\begin{lem}
  \label{lem:size-of-E-in-uframe}
  If \(E\in \uframe\), then \(|E|=n\).
\end{lem}
\begin{proof}
  Suppose \(E\in \uframe\). Then
  \(  (E, P) = \closeddecomp(T) \) { for some } \(T \in \TRn{n}{\ralph}\) and \(P\).
  By Lemma~
\ref{lemma:charphi}\eqref{charphi:size-of-e},
  \(|E|
=|T|=n\).
  \end{proof}

For a second-order context \(E\), 
and \(m \ge 1\), 
we define the set \(\useq\) of sequences of affine contexts:
\begin{align*}
  \useq &\defeq \set{ P \mid (E, P) = \closeddecomp(T) \text{ for some } T \in \TR{\ralph} }.
\end{align*}
The set \(\useq\) consists of sequences \(P\) of affine contexts that
match \(E\) and are obtained by the decomposition function \(\closeddecomp\).
In the rest of this subsection,
we prove the bijection in
Lemma~\ref{cor:gindependent-tree-decomposition}
in two steps.
We first show \(\TRn{n}{\ralph}\cong \coprod_{E\in \uframe} \useq\) (Lemma~\ref{lemma:union-sig},
called ``coproduct lemma''), and then show 
\(\useq = \prod_{j=1}^{\shn(E)} \ucomph{\sh{E}{j}}\) (Lemma~\ref{lemma:union-sig}, called ``product lemma'').

\begin{lem}[Coproduct Lemma (for Grammar-Independent Decomposition)]
  \label{lemma:union-sig}
  For any
  $n \geq 1$ and \(m \ge 1\), 
  there exists a bijection
  \begin{flalign*}
\TRn{n}{\ralph}\cong \coprod_{E\in \uframe} \useq
  \end{flalign*}
  that maps each  element \((E,P)\) of the set
  \(\coprod_{E\in \uframe} \useq\) to \(E\Hfill{P}\in \TRn{n}{\ralph}\).
\end{lem}
\begin{proof}
  We define a function
  \[
  f:
  \coprod_{E \in \uframe} \useq
  \longrightarrow
  \TRn{n}{\ralph}
  \]
  by \(f(E,P) \defeq E\Hfill{P}\),
  and a function
  \[
  g:
  \TRn{n}{\ralph}
  \longrightarrow
  \coprod_{E\in\uframe} \useq
  \]
  by \(g(T) \defeq \closeddecomp(T) \).
  
  Let us check that these are functions into the codomains:
  \begin{itemize}
  \item \( f(E,P) \in \TRn{n}{\ralph} \):
    Since \(P \in \useq\), there exists \( T \in \TR{\ralph} \) such that 
\( (E,P) = \closeddecomp(T) \).
 By Lemma~\ref{lemma:charphi}(\ref{charphi:split}),
    we have \(f(E,P) = E\Hfill{P}=T \in \TR{\ralph}\). 
    By the condition \(E\in\uframe\) and by
    Lemmas~\ref{lem:second-substitution-preserves-size} and \ref{lem:size-of-E-in-uframe},
\(|E\Hfill{P}|=|E|=n\). Thus, \(f(E,P)\in\TRn{n}{\ralph}\) as required.
  \item \( g(T) \in \coprod_{E \in \uframe} \useq \):
    Obvious from the definitions of \( \uframe \) and \( \useq \).
  \end{itemize}
  
  We have \(f(g(T)) = T\) by Lemmas~\ref{lemma:charphi}(\ref{charphi:split}). 
  Let \( (E,P) \in \coprod_{E \in \uframe} \useq \).
  By definition, there exists \( T \in \TR{\ralph} \) such that \( (E,P) = \closeddecomp(T)\).
 By using Lemmas~\ref{lem:second-substitution-preserves-size} and \ref{lem:size-of-E-in-uframe} again, 
we have \(|T|=|E\Hfill{P}|=|E| = n\). Thus, \((E,P)=g(T)\).
  Then
  \[
  g(f(E,P))
  =g(f(g(T))
  =g(T)
  =(E,P).
  \tag*{\qedhere}
  \]
\end{proof}

It remains to show the product lemma:
\(\useq = \prod_{j=1}^{\shn(E)} \ucomph{\sh{E}{j}}\).
To this end, we prove a few more properties about the auxiliary decomposition function \(\decomp\).

\begin{lem}\label{lem:size-of-E}
If \(|T|\geq m\) and \(\decomp(T)=(U,E,P)\), then \(|E|\geq m\).
\end{lem}
\begin{proof}
Straightforward induction on \(T\).
\end{proof}
The following lemma states that, given $\decomp(U[T]) = (U,E,P)$, $P$ is
determined only by $T$ ($U$ does not matter); this is because the decomposition is
performed in a bottom-up manner.

\begin{lem}
  \label{lemma:ehole}
  For $m \geq 1$, \(E\), \(P\), \(T\), and a \linear{context} \(U\) with {\(|T|\ge m\)} and \( |U|<m \),
  we have \( \decomp(T) = (\hole, E, P) \) if and only if \( \decomp(U[T]) = (U, E, P) \).
\end{lem}
\begin{proof}
  The proof proceeds by induction on \(|U|\).
  If \( U = \hole \) the claim trivially holds.
  If \(U \neq \hole\), \(U\) is of the form \(a(T_1,\dots,U_i,\dots,T_{\ralph(a)})\).
  Since \(|U| < m\), we have \( |U_i| < m \) and \(|T_j|<m\) for every \(j\neq i\).

  Assume that \( \decomp(T) = (\hole, E, P) \).
  By the induction hypothesis, we have \( \decomp(U_i[T]) = (U_i, E, P) \).
  Since \( |U| = |a(T_1, \dots, T_{i-1}, U_i, T_{i+1}, \dots, T_{\ralph(a)})| < m \), we should apply the third case of Equation~\eqref{eq:defOfDecomp} to compute \( \decomp(U[T]) \).
  Hence, we have \[ \decomp(U[T]) = (a(T_1, \dots, T_{i-1}, U_i, T_{i+1}, \dots, T_{\ralph(a)}),E,P) = (U, E, P). \]

  Conversely, assume that
\[
\decomp(U[T]) = (U, E, P) = (a(T_1, \dots, T_{i-1}, U_i, T_{i+1}, \dots, T_{\ralph(a)}), E, P) .
\]
  Let \( (U'_i, E'_i, P'_i) = \decomp(U_i[T]) \) and \( (U'_j, E'_j, P'_j) = \decomp(T_j) \) for each \( j \neq i \).
  Then the final step in the computation of \( \decomp(U[T]) \) must be the third case; otherwise \( U = \hole \), a contradiction.
  By the position of the unique hole in \( U \), it must be the case that \( U = a(T_1, \dots, T_{i-1}, U'_i, T_{i+1}, \dots, T_{\ralph(a)}) \), \( E = E'_i \) and \( P = P'_i \).
  So \( (U_i, E, P) = \decomp(U_i[T]) \).
  By the induction hypothesis, \( \decomp(T) = (\hole, E, P) \).
\end{proof}

The following is the key lemma for the product lemma, which says that
if \((U,E,P)=\decomp(T)\), the decomposition is actually independent of the \(P\)-part.

\begin{lem}\label{lem:decomposition-replace}
If \( (U, E, P) = \decomp(T) \), then 
 \( \decomp(U[E]\Hfill{P'}) = (U, E, P') \) for any \(P'\in \prod_{i=1}^{\shn(E)} \ucomp[\sh{E}{i}]\).
\end{lem}
\begin{proof}
The proof proceeds by induction on \(|T| (= |U[E]|)\). If \(|T|< m\), then \(E=T\), hence \(\shn(E)=0\). Thus, 
\(P=P'=\epsilon\), which implies 
 \( \decomp(U[E]\Hfill{P'}) = \decomp(T)=(U,E,P)=(U,E,P')\).

For the case \(|T|\ge m\), 
we proceed by the case analysis on which rule of
Equation~\eqref{eq:defOfDecomp} was
used to compute \( (U, E, P) = \decomp(T) \).
  Assume that \( T = a(T_1, \dots, T_{\ralph(a)}) \) and \( (U_i, E_i, P_{i}) = \decomp(T_i) \) for each \( i = 1, \dots, \ralph(a) \).
  By Lemma~\ref{lemma:charphi}, we have \( \cbf{P_{i}}{E_i} \) for each \( i \).
  \begin{itemize}
  \item
    The first case of Equation~\eqref{eq:defOfDecomp}:
    We have \( |T_j|, |T_{j'}|\ge m\)
    for some \( 1 \le j < j' \le \ralph(a) \), and:
    \[
    U = \hole
    \qquad
    E = a(U_1[E_1], \dots, U_{\ralph(a)}[E_{\ralph(a)}])
    \qquad
    P = P_{1} \cdots P_{\ralph(a)}.
    \]
Since \(P, P':E\), we can split
 \( P'=P'_{1} \cdots P'_{\ralph(a)} \) so that \( \cbf{P'_{i}}{E_i} \) for each \( i \).
    By the induction hypothesis,
    \[
    \decomp(U_i[E_i]\Hfill{P'_{i}}) = (U_i, E_i, P'_{i})
    \qquad\mbox{(for each \( i = 1, \dots, \ralph(a) \)).}
    \]
    Since \(|U_j[E_j]\Hfill{P'_{j}}|= |U_{j}[E_{j}]| =|T_j|\ge m \) and
    \(|U_{j'}[E_{j'}]\Hfill{P'_{j'}}|=  |U_{j'}[E_{j'}]|=|T_{j'}| \ge m \) for some \( 1 \le j < j' \le \ralph(a) \), we have
    \begin{align*}
      \decomp(U[E\Hfill{P'}])
      &= \decomp(a(U_1[E_1\Hfill{P'_1}], \dots, U_{\ralph(a)}[E_{\ralph(a)}\Hfill{P'_{\ralph(a)}}])) \\
      &= (\hole, a(U_1[E_1], \dots, U_{\ralph(a)}[E_{\ralph(a)}]), P'_{1} \cdots P'_{\ralph(a)})\\
      &= (U, E, P')
    \end{align*}
as required.
  \item
    The second case of Equation~\eqref{eq:defOfDecomp}:
    We have \( |T_j| \ge m \) for a unique \( j \in \{ 1, \dots, \ralph(a) \} \) and:
    \[
  n \defeq |U_0| \ge m \qquad
    U = \hole
    \qquad
    E = \hhole^n_1[E_{j}]
    \qquad
    P = U_0 \cdot P_j
    \]
    for \( U_0 \defeq a(T_1, \dots, T_{j-1}, U_{j}, T_{j+1}, \dots, T_{\ralph(a)}) \).
 Because \(|T_j|\geq m\),  we have \(|E_j|\ge m\) by Lemma~\ref{lem:size-of-E}.
    Now we have \(|E_j\Hfill{P_j}|=|E_j|\ge m\), \( |U_j|<m \), and \( \decomp(U_j[E_j\Hfill{P_j}]) =
       \decomp(T_j) = (U_j, E_j, P_j) \); hence
by Lemma~\ref{lemma:ehole}, we have \( \decomp(E_j\Hfill{P_j}) = (\hole, E_j, P_j) \).
    Let \( P' = U_0' \cdot P'' \).
    By the assumption, \( \cbf{U_0'}{\hhole^n_1} \) and thus \( U_0' \) is a \( 1 \)-context.
    Also, \( U_0' \) is good for \( m \); thus \( |U_0'| \ge m \)
    and \(U'_0\) is of the form \( a'(T_1', \dots, T'_{j'-1}, U'_{j'}, T'_{j+1}, \dots, T'_{\ralph(a')}) \)
    where \( |U'_{j'}| < m \) and \( |T'_i| < m \) for every \( i \neq j' \).
    Now we have
    \begin{enumerate}
    \item \( \decomp(U'_{j'}[E_j]\Hfill{P_j}) = (U'_{j'}, E_j, P_j) \) by Lemma~\ref{lemma:ehole} and \( \decomp(E_j\Hfill{P_j}) = (\hole, E_j, P_j) \);
    \item \(P''\in \prod_{i=1}^{\shn(E_j)} \ucomp[\sh{E_j}{i}]\) since 
\(P'=U_0'\cdot P''\in \prod_{i=1}^{\shn(E)} \ucomp[\sh{E}{i}] \) where \(E=\hhole^n_1[E_{j}]\);
    \item \( |U'_{j'}[E_j]| \;<\; |U_0'[E_j]| \;=\; n + |E_j| \;=\; |E| \;=\; |U[E]| \).
    \end{enumerate}
  Therefore,  we can apply the induction hypothesis to \( U'_{j'}[E_j] \), resulting in
    \[
    \decomp(U'_{j'}[E_j]\Hfill{P''}) = (U'_{j'}, E_j, P'').
    \]
    We have \( |U'_{j'}[E_j]|\ge |E_j| \ge m \) and \( |T'_i| < m \) for every \( i \neq j' \).
    Furthermore
    \[
    |a'(T'_1, \dots, T'_{j'-1}, U'_{j'}, T'_{j'+1}, \dots, T'_{\ralph(a')})| \;=\; |U_0'| \;=\; n \;\ge\; m.
    \]
    Hence we have
    \begin{align*}
      \decomp(U[E]\Hfill{P'})
      &=       
\decomp(U_0'[E_j]\Hfill{P''}) \\
      &= \decomp(a'(T'_1, \dots, T'_{j'-1}, U'_{j'}, T'_{j'+1}, \dots, T'_{\ralph(a')})[E_j]\Hfill{P''}) \\
      &= \decomp(a'(T'_1, \dots, T'_{j'-1}, U'_{j'}[E_j]\Hfill{P''}, T'_{j'+1}, \dots, T'_{\ralph(a')})) \\
      &= (\hole, \hhole^n_1[E_j], U'_0 \cdot P'')\\
      &= (U,E,P')\\
    \end{align*}
  \item
    The third case of Equation~\eqref{eq:defOfDecomp}:
    We have \( |T_j| = |U_j[E_{j}]| \ge m \) for a unique \( j \in \{ 1, \dots, \ralph(a) \} \) (and thus \( |T_i| < m \) for every \( i \neq j \)) and: 
    \[
    U = a(T_1, \dots, T_{j-1}, U_{j}, T_{j+1}, \dots, T_{\ralph(a)})
    \qquad
    |U|<m
    \qquad
    E = E_j
    \qquad
    P = P_j.
    \]
    By the induction hypothesis,
    \[
    \decomp(U_j[E_j]\Hfill{P'}) = (U_j, E_j, P').
    \]
    Since \( |T_i| < m \) for every \( i \neq j \) and \( |a(T_1, \dots, T_{j-1}, U_{j}, T_{j+1}, \dots, T_{\ralph(a)})| < m \), we have
    \begin{align*}
      \decomp(U[E]\Hfill{P'})
      &= \decomp(a(T_1, \dots, T_{j-1}, U_j[E_j]\Hfill{P'}, T_{j+1}, \dots, T_{\ralph(a)})) \\
      &= (a(T_1, \dots, T_{j-1}, U_{j}, T_{j+1}, \dots, T_{\ralph(a)}), E_j, P')\\
      &= (U, E, P')\\
    \end{align*}
  \item
    The fourth case of Equation~\eqref{eq:defOfDecomp}:
    Let \( n \) be \( |T| \); we have \( n \ge m \) and:
    \[
    U = \hole
    \qquad
    E = \hhole^n_0
    \qquad
    P = T.
    \]
    Since \( P' : E \), \(P'\) must be a singleton sequence consisting of a tree, 
say \( T'=a'(T_1', \dots, T_{\ralph(a')}') \).
    By the assumption, \( T' \) is good for \( m \).
    Hence \( |T'| \ge m \) and \( |T_i'| < m \) for every \( i = 1, \dots, \ralph(a') \).
    So
    \[
    \decomp(U[E]\Hfill{P'})
    \;=\; \decomp(T')
    \;=\; (\hole, \hhole^n_0, T')
    \;=\; (U,E,P').
    \tag*{\qedhere}
    \]
  \end{itemize}
\end{proof}

\begin{cor}\label{cor:closed-decomposition-replace}
If \( (E, P) = \closeddecomp(T) \), then 
 \( \closeddecomp(E\Hfill{P'}) = (E, P') \) for any \(P'\in \prod_{i=1}^{\shn(E)} \ucomp[\sh{E}{i}]\).

\end{cor}
\begin{proof}
If \( (E, P) = \closeddecomp(T) \), then \((U,E',P)=\decomp(T)\) and \(U[E']=E\) for some \(U, E'\).
Since \(\prod_{i=1}^{\shn(E)} \ucomp[\sh{E}{i}] = \prod_{i=1}^{\shn(E')} \ucomp[\sh{E'}{i}]\), Lemma~\ref{lem:decomposition-replace}
implies \(\decomp(E\Hfill{P'}) = \decomp(U[E']\Hfill{P'}) = (U,E',P')\). Thus, we have
\(\closeddecomp(E\Hfill{P'}) = (U[E'],P')=(E,P')\) as required.
\end{proof}

We are now ready to prove the product lemma.
\begin{lem}[Product Lemma (for Grammar-Independent Decomposition)]\label{lem:product}
  For \(E\) and \(m \ge 1\), 
  if \(\useq\) is non-empty, then
  \[
  \useq = \prod_{i=1}^{\shn(E)} \ucomp[\sh{E}{i}].
  \]
\end{lem}
\begin{proof}
  The direction \(\subseteq\) follows from
the fact that \( (E,P) = \closeddecomp(T) \) implies \( \cbf{P}{E} \) and \( \prj{P}{i} \) is good for \( m \)
(Lemma~\ref{lemma:charphi}).
  
  We show the other direction.
  Since \( \useq \neq \emptyset \), there exist \(P\) and \( T \) such that \( \closeddecomp(T) = (E, P) \).
  Let \( P' \in \prod_{i=1}^{\shn(E)} \ucomp[\sh{E}{i}] \). 
  By Corollary~\ref{cor:closed-decomposition-replace}, we have \( \closeddecomp(E\Hfill{P'}) = (E, P') \).
  This means that \( P' \in \useq \).
\end{proof}

Lemma~\ref{cor:gindependent-tree-decomposition} follows as
an immediate corollary of Lemmas~\ref{lemma:union-sig} and \ref{lem:product}.

\subsection{Grammars in Canonical Form}
\label{sec:canonical-grammar}

As a preparation for generalizing the decomposition of \(\TRn{n}{\ralph}\) (Lemma~\ref{cor:gindependent-tree-decomposition})
to that of \(\inhavn{\grm,N}\), 
we first transform a given regular tree grammar into \emph{canonical form}, which will be defined 
shortly (in Definition~\ref{def:canonical-gram}).
We prove that the transformation preserves unambiguity and (a weaker version of) strong connectivity.

\begin{defi}[Canonical Grammar]
\label{def:canonical-gram}
  A rewriting rule of a regular tree grammar \( \grm \) is in \emph{canonical form} if it is of the form
  \[
  N \rew_\grm a(N_1, \dots, N_{\ralph(a)}).
  \]
  A grammar \( \grm = (\ralph, \nont, \rules) \) is \emph{canonical} if every rewriting rule is in canonical form.
\end{defi}

We transform a given regular tree grammar \( \grm = (\ralph, \nont, \rules) \) to 
an equivalent one in canonical form.
The idea of the transformation is fairly simple: we replace a rewriting rule
\[
N \rew a(T_1, \dots, T_n) \qquad\mbox{(\( T_1, \dots, T_n \) are \( (\ralph \cup \nont) \)-trees)}
\]
such that \( T_i \notin \nont \) with rules
\[
\{\;\; N \rew a(T_1, \dots, T_{i-1}, N', T_{i+1}, \dots, T_n),
\qquad
N' \rew T_i \;\;\}
\]
where \( N' \) is a fresh nonterminal that does not appear in \( \nont \).
After iteratively applying the above transformation, we next replace a rewriting rule of the form \( N \rew N' \) with rules
\[
\{ N_0 \rew a(\dots, N', \dots) \mid (N_0 \rew a(\dots, N, \dots)) \in \rules \}.
\]
Again by iteratively applying this transformation, we finally obtain a grammar in canonical form.

The transformation, however, does not preserve strong connectivity.
For example, consider the grammar \( \grm = (\{ \Ta \mapsto 2, \Tb \mapsto 1, \Tc \mapsto 0 \}, \{ N \}, \rules) \) where
\[
\rules = \{\;\;
  N \rew \Tb(\Tc), \qquad
  N \rew \Ta(N, N)
\;\;\}.
\]
Then the above transformation introduces a nonterminal \( N' \) as well as rules
\[
N \rew \Tb(N')
\qquad\mbox{and}\qquad
N' \rew \Tc.
\]
Then \( N \) is not reachable from \( N' \).

Observe that the problem above was
caused by a newly introduced nonterminal that generates a single finite tree.
To overcome the problem, we introduce a weaker version of strong connectivity
called \emph{essential strong-connectivity}. It
requires strong connectivity only for nonterminals generating infinite languages;
hence, it is preserved by the above transformation.
\begin{defi}[Essential Strong-connectivity]
  Let \( \grm = (\ralph, \nont, \rules) \) be a regular tree grammar.
  We say that \( \grm \) is \emph{essentially strongly-connected} if for any nonterminals \( N_1, N_2 \in \nont \) with \( \card{\inhav{\grm,N_1}} = \card{\inhav{\grm,N_2}} = \infty \), \( N_2 \) is reachable from \( N_1 \).
\end{defi}
Note that by the definition, every strongly-connected grammar is
also essentially-strongly connected.
In the definition above, as well as in the arguments below, nonterminals \( N \) with \( \card{\inhav{\grm,N}} = \infty \) play an important role.
We write \( \nontinf \) for the subset of \( \nont \) consisting of such nonterminals.

\begin{rem}
  A regular tree grammar that is essentially strongly-connected can be easily transformed into a strongly-connected grammar, hence the terminology.
  Let \( \grm \) be an essentially strongly-connected grammar and \( N \in \nontinf \).
  We say that a nonterminal is \emph{inessential} if it generates a finite language.
  Let \( N_0 \) be an inessential nonterminal of a grammar \( \grm \) such that \( \inhav{\grm, N_0} = \{ T_1, \dots, T_n \} \).
  Then by replacing each rule
  \[
  N_1 \rew C[N_0,\ldots,N_0]
  \]
  (where \( C \) is a \( k \)-context possibly having nonterminals other than \(N_0\)) with rules
  \[
  \{ N_1 \rew C[T_{i_1},\ldots,T_{i_k}] \mid i_1,\ldots,i_k\in\set{1, \dots, n} \},
  \]
  one can remove the inessential nonterminal \( N_0 \) from the grammar.
  A grammar \( \grm \) is essentially strongly-connected if and only if the grammar \( \grm' \) obtained by removing all inessential nonterminals is strongly-connected.
  This transformation preserves the language in the following sense: 
writing \( \grm' = (\ralph, \nontinf, \rules') \) for the resulting grammar, 
we have \( \inhav{\grm, N} = \inhav{\grm', N} \) for each \(N \in \nontinf\).
  Note that the process of erasing inessential nonterminals breaks canonicity; in fact, 
  {the class of languages generated by strongly-connected canonical regular tree grammars
    is a \emph{proper} subset of that of essentially strongly-connected canonical regular tree grammars.}
\end{rem}

Recall that the second main theorem (Theorem~\ref{thm:imt}) takes a family \( (S_n)_{n \in \nat} \) 
from \( \ctx{\grm} = \bigcup_{N, N' \in \nont} \inhav{\grm, N \Ar N'} \).
In order to restate the theorem for essentially strongly connected grammars,
we need to replace
\(\ctx{\grm}\) with the ``essential'' version, namely,
\[
\ctxinf{\grm} \defeq \bigcup_{N, N' \in \nontinf} \inhav{\grm, N \Ar N'}.
\]
\begin{lem}[Canonical Form]
  \label{lem:canonical-grammar}
  Let \( \grm = (\ralph, \nont, \rules) \) be a regular tree grammar that is unambiguous and strongly connected.
  Then one can (effectively) construct a grammar \( \grm' = (\ralph, \nont', \rules') \) and 
  a family \( ( \nset_N )_{N \in \nont} \) of subsets \( \nset_N \subseteq \nont' \) that satisfy the following conditions:
  \begin{itemize}
  \item \( \grm' \) is canonical, unambiguous and essentially strongly-connected.
  \item 
	\( \inhav{\grm,N} = \biguplus_{N' \in \nset_N} \inhav{\grm',N'} \) for every \(N \in \nont\).
  \item If \( \card{\inhav{\grm}} = \infty \), then \( \ctx{\grm} \subseteq \ctxinf{\grm'} \).
  \end{itemize}
\end{lem}
\begin{proof}
  See Appendix~\ref{appx:canonical}.
\end{proof}
 
\subsection{Decomposition of Regular Tree Languages}
\label{sec:rtl-decomposition}

This subsection generalizes the decomposition of \(\TRn{n}{\ralph}\)
in Section~\ref{sec:decomposition}:
\[  \TRn{n}{\ralph}\;\cong
  \coprod_{E \in \uframe} \prod_{i=1}^{\shn(E)} \ucomp[\sh{E}{i}]
  \]
to
that of \(\inhavn{\grm,N}\), and proves a bijection of the following form:
  \[
  \inhavn{\grm,N}\;\cong
  \coprod_{\gctx \in \tframe(\grm,N)} \prod_{i=1}^{\shn(\gctx)} \tcomp[\sh{\gctx}{i}]. 
  \]
  Here, \(\gctx\) denotes a \emph{typed} second-order context
  (which 
  will be defined shortly in Section~\ref{sec:typed-context}),
  each of whose second-order holes
\(\sh{\gctx}{i}=\Hhole{\f}{n}\) carries not only the size \(n\) but the \emph{context type} \(\f\) of a context to be substituted for the hole.
Accordingly, we have replaced 
\(\ucomp[\sh{E}{i}]\) with  \(\tcomp[\sh{\gctx}{i}]\), which denotes the set of contexts that respect
the context type specified by \(\sh{\gctx}{i}\).

In the rest of this subsection, we first define the notion of typed second-order contexts in
Section~\ref{sec:typed-context}, extend the decomposition function accordingly in Section~\ref{sec:grespecting-decomposition}, and use it to prove the above bijection in Section~\ref{sec:grespecting-bijection}.
Throughout this subsection (i.e.,~Section~\ref{sec:rtl-decomposition}), 
we assume that \(\grm\) is a canonical and unambiguous grammar.
We emphasize here that the discussion in this subsection 
essentially relies on both unambiguity and canonicity of the grammar.
The essential strong connectivity is not required for the results in this
subsection; it will be used in Section~\ref{sec:adjusting}, to show
that each component \(\tcomp[\sh{\gctx}{i}]\) contains an affine
context that has \(S_{\ceil{\cn \log n}}\) as a subcontext
(recall condition (i) of (T3) in Section~\ref{sec:IMTforWords}).

\begin{rem}
  Note that any deterministic bottom-up tree automaton
  (without any \(\epsilon\)-rules) ~\cite{tata2007}
can be considered an unambiguous canonical tree grammar, by regarding
each transition rule \(a(q_1,\ldots,q_k) \rew q\) as a rewriting rule
\(N_q \rew a(N_{q_1},\ldots, N_{q_k})\). Thus, by the equivalence between the
class of tree languages generated by regular tree grammars and the class of those
accepted by deterministic bottom-up tree automata, any regular tree grammar can be
converted to an unambiguous canonical tree grammar.
\end{rem}
\subsubsection{Typed Second-Order Contexts}
\label{sec:typed-context}

The set of \emph{\( \grm \)-typed second-order contexts}, ranged over by \(\gctx\), is defined by:
\[
\gctx ::= \hhole_{N_1 \cdots N_k \Ar N}^n[\gctx_1, \ldots, \gctx_k] \;\mid\; a(\gctx_1, \ldots,
\gctx_{\ralph(a)}) \ \ (a \in \dom(\ralph)).
\]
The subscript \( \kappa = (N_1 \cdots N_k \Ar N) \) describes the type of
first-order contexts that can be substituted for this second-order hole; 
the superscript \( n \) describes the size as before. 
Hence a (first-order) context \( C \) is suitable for 
filling a second-order hole \( \hhole_{\kappa}^n \) if \( C : \kappa \) and \( |C| = n \).
We write \( \cbf{C}{\hhole_{\kappa}^n} \) if \( C : \kappa \) and \( |C| = n \).
The operations such as \( \sh{\gctx}{i} \), \( \gctx\Hfill{C} \) and \( \gctx\Hfill{P} \) are defined 
analogously.
For a sequence of contexts \(P=C_1C_2\cdots C_{\ell}\),
we write \( \cbf{P}{\gctx} \)
if \( \len{P} (= \ell)=\shn(\gctx) \) and \( \cbf{\pct_i}{\sh{\gctx}{i}} \) for each \( i \le \shn(\gctx) \).

We define the second-order context typing relation \(\p \gctx:N\) inductively
by the rules in Figure~\ref{fig:TypeRuleE}. Intuitively, 
\(\p \gctx:N\) means that \(\gctx\Hfill{P}\in 
\inhav{\grm,N}\) holds for any \(P\) such that \(P\COL \gctx\) (as confirmed in Lemma~\ref{lem:second-order-substitution-lemma} below).
As in the case of untyped second-order contexts,
 we actually use only typed second-order contexts with holes
of the form \(\hhole_{N_1 \cdots N_k \Ar N}^n\) where \(k\) is \(0\) or \(1\).

\begin{figure}[t]
\[
\begin{array}{cr}
  \dfrac{
    \p \gctx_i: N_i
    \quad
    \mbox{(for each \( i = 1, \dots, k \))}
  }{
    \p \hhole_{N_1 \cdots N_k \Ar N}^{n}[\gctx_1,\dots,\gctx_k] : N
  }
&
  (\text{SC-Hole})
  \\[25pt]
  \dfrac{
    (N \rew a(N_1,\dots,N_{\ralph(a)})) \in \rules
    \qquad
    \p \gctx_i : N_i \quad\mbox{(for each $i = 1,\dots,\ralph(a)$)}
  }{
    \p a(\gctx_1,\dots,\gctx_{\ralph(a)}) : N
  }
  &
  (\text{SC-Term})
\end{array}
\]
\caption{Typing rules for \(\gctx\)}\label{fig:TypeRuleE}
\end{figure}

\begin{exa}
\label{ex:gctx}
Recall the second-order context
\(E = \Tb(\Hhole{1}{3}[\Ta(\Hhole{0}{3}, \Tb(\Hhole{0}{5}))])\) in 
Figure~\ref{fig:partitioning} and Example~\ref{ex:decomp}.
Given the grammar consisting of the rules:
\[A\to \Ta(B,B) \qquad B\to \Tb(A)\qquad B\to \Tb(B)\qquad B\to \Tc,\]
the corresponding \emph{typed} second-order context \(\gctx\) is:
\[\Tb(\Hhole{A\Ar A}{3}[\Ta(\Hhole{\Ar B}{3}, \Tb(\Hhole{\Ar A}{5}))])\]
and we have \(\p \gctx : B\).
For \(P = \Ta(\Tb(\hole),\Tc)\cdot \Tb(\Tb(\Tc)) \cdot \Ta(\Tb(\Tc), \Tb(\Tc))\), we have \(P\COL\gctx\),
and:
\[\gctx\Hfill{P} = \Tb(\Ta(\Tb(\Ta(\Tb(\Tb(\Tc)), \Tb(\Ta(\Tb(\Tc), \Tb(\Tc))))), \Tc)) \in \inhav{\grm,B}.\]
\end{exa}

\begin{lem}
  The following rule is derivable:
  \[
  \begin{array}{cr}
    \dfrac{
      C \in \inhav{\grm, N_1 \dots N_k \Ar N}
      \qquad
      \p \gctx_i : N_i
      \quad
      \mbox{\rm (for each \( i = 1, \dots, k \))}
    }{
      \p C[\gctx_1, \dots, \gctx_k] : N
    }
    &
    (\mbox{\rm SC-Ctx})
  \end{array}
  \]
\end{lem}
\begin{proof}
  The proof proceeds by induction on \(C\).
   Assume that the premises of the rule hold.
  If \(C=\hole\), then \(C[\gctx_1, \dots, \gctx_k] =\gctx_1\)
  and \(N=N_1\); thus the result follows immediately from the assumption.
  If \(C=a(C_1,\ldots,C_\ell)\), then by the assumption that the grammar is canonical
  and \(C \in \inhav{\grm, N_1 \dots N_k \Ar N}\), 
  we have \(N\to a(N'_1,\ldots,N'_\ell)\) with \(C_i\in
  \inhav{\grm,N_{k_{i-1}+1}\dots N_{k_i}\Ar N'_i}\)
  for \(i=1,\ldots,\ell\) where \(k_0=0\) and \(k_\ell=k\).
  By the induction hypothesis,
  we have \(\p C_i[\gctx_{k_{i-1}+1},\ldots,\gctx_{k_{i}}]:N'_i\).
  Thus, by using rule (\rm{SC-Term}), we have
\(\p  C[\gctx_1, \dots, \gctx_k] : N\) as required.
\end{proof}

\begin{lem}
  \label{lem:second-order-substitution-lemma}
  Assume that \( \p \gctx : N \).
  \begin{enumerate}
  \item\label{item:second-order-substitution-lemma-tree}
  If \( \shn(\gctx) = 0 \), then \( \gctx \) is a tree and \( \gctx \in \inhav{\grm, N} \).
  \item\label{item:second-order-substitution-lemma-sbst}
  If \( \shn(\gctx) \ge 1 \) and \( \cbf{C}{\sh{\gctx}{1}} \), then \( \p \gctx\Hfill{C} : N \).
  \item\label{item:second-order-substitution-lemma-sbstall}
  If \( \cbf{P}{\gctx} \), then \( \gctx\Hfill{P} \in \inhav{\grm, N} \).
  \end{enumerate}
\end{lem}
\begin{proof}
  \((1)\) By induction on the structure of \( \p \gctx : N \).
  \((2)\) By induction on the structure of \( \p \gctx : N \).
  The base case is when \( \gctx = \hhole_{N_1 \dots N_k \Ar N}^n[\gctx_1, \dots, \gctx_k] \).
  Since \( \p \gctx : N \), we have \( \p \gctx_i : N_i \) for each \( i \).
  Then by the derived rule (\textsc{SC-Ctx}), we have \( \p C[\gctx_1, \dots, \gctx_k] : N \).
  \((3)\) By induction on \( \shn(\gctx) \) (using \((1)\) and \((2)\)).
\end{proof}

\subsubsection{Grammar-Respecting Decomposition Function \(\closedgdecomp\)}
\label{sec:grespecting-decomposition}

We have defined in Section~\ref{sec:gindependent-decomposition}
the function \(\closeddecomp\) that decomposes a tree \(T\) and 
returns a pair \((E,P)\) of  
 a second-order context \( E \) and
 a sufficiently long sequence \(P\) of (first-order) contexts.
The aim here is to extend \(\closeddecomp\) to a grammar-respecting one
\(\closedgdecomp\) that takes a pair \((T,N)\) such that 
\(T\in \inhav{\grm,N}\) as an input, and returns a pair \((\gctx,P)\),
which is the same as \(\closeddecomp(T)=(E,P)\), except that
 \(\gctx\) is a ``type-annotated'' version of  \(E\).
For example, for the tree 
in Figure~\ref{fig:partitioning} and the grammar in Example~\ref{ex:gctx},
we expect that \(\closedgdecomp(T) = (\gctx,P)\) where:
\[
\begin{array}{l}
T = \Tb(\Ta(\Tb(\Ta(\Tb(\Tb(\Tc)), \Tb(\Ta(\Tb(\Tc), \Tb(\Tc))))), \Tc))\\
\gctx = \Tb(\Hhole{A\Ar A}{3}[\Ta(\Hhole{\Ar B}{3}, \Tb(\Hhole{\Ar A}{5}))])\qquad
P = \Ta(\Tb(\hole),\Tc)\cdot \Tb(\Tb(\Tc)) \cdot \Ta(\Tb(\Tc), \Tb(\Tc)).
\end{array}
\]

We say that \emph{\( \gctx \) refines \( E \)}, written \( \gctx \lhd E \), if \( E \) is obtained by simply forgetting type annotations, i.e.,~replacing \( \hhole_{N_1 \cdots N_k \Ar N}^n \) in \( \gctx \) with \( \hhole_k^n \).
This relation is formally defined by induction on the structures of \( E \) and \( \gctx \) by the rules in Figure~\ref{fig:refinement-of-second-order-contexts}.
\begin{figure}[t]
\[
\begin{array}{cr}
  \dfrac{
    \gctx_i \lhd E_i
    \quad
    \mbox{(for each \( i = 1, \dots, k \))}
  }{
    \hhole_{N_1 \cdots N_k \Ar N}^{n}[\gctx_1,\dots,\gctx_k] \lhd \hhole_k^n[E_1,\dots,E_n]
  }
&
  (\text{Ref-Hole})
  \\[25pt]
  \dfrac{
    \gctx_i \lhd E_i \quad\mbox{(for each $i = 1,\dots,\ralph(a)$)}
  }{
    a(\gctx_1,\dots,\gctx_{\ralph(a)}) \lhd a(E_1, \dots, E_n)
  }
  &
  (\text{Ref-Term})
\end{array}
\]
\caption{Refinement relation}\label{fig:refinement-of-second-order-contexts}
\end{figure}
The following lemma is obtained by straightforward induction on \(\gctx\lhd E\).
\begin{lem}
\label{lem:ref-basic}
  Assume \( \gctx \lhd E \).
  \begin{enumerate}
  \item\label{item:ref-basic-size}
  \( |\gctx| = |E| \).
  \item\label{item:ref-basic-type}
  If \( \cbf{P}{\gctx} \), then \( \cbf{P}{E} \).
  \item\label{item:ref-basic-sbst}
  If \( \cbf{P}{\gctx} \), then \( \gctx\Hfill{P} = E\Hfill{P} \).
  \end{enumerate}
\end{lem}

Given \( T \in \inhav{\grm, N} \) and \( (E, P) = \closeddecomp(T) \), 
the value of \( \closedgdecomp(T, N) \) should be \( (\gctx, P) \)
where \( \gctx \) is a \( \grm \)-typed second-order context that satisfies the following conditions:
\[
\gctx \lhd E,
\qquad
\p \gctx : N
\quad\mbox{and}\quad
\cbf{P}{\gctx}.
\]
We prove that, for every \(T \in \inhav{\grm,N}\), 
there exists exactly one typed second-order context \( \gctx \) that satisfies the above condition.
Hence the above constraints define the function \( \closedgdecomp(T, N) \).

We first state and prove a similar 
result for first-order contexts (in Lemma~\ref{lem:first-unique-decomposition}), 
and then prove the second-order version (in Lemma~\ref{lem:second-unique-decomposition}), 
using the former.
\begin{lem}
\label{lem:first-unique-decomposition}
  Let \( C \) be a (first-order) \(k\)-context and \( T_1, \dots, T_k \) be trees. 
  Assume that \( C[T_1, \dots, T_k] \in \inhav{\grm, N} \).
  Then there exists a unique family \( (N_i)_{i \in \set{1, \dots, k}} \) such that 
  \( C \in \inhav{\grm, N_1 \dots N_k \Ar N} \) and \( T_i \in \inhav{\grm, N_i} \) for every \( i = 1, \dots, k \).
\end{lem}
\begin{proof}
  By induction on \( C \).
  \begin{itemize}
  \item Case \( C = \hole \):
    Then \( k = 1 \).
    The existence follows from \( \hole \in \inhav{\grm, N \Ar N} \) (as \( N \rew^* N \)) and \( T_1 = C[T_1] \in \inhav{\grm, N} \).
    The uniqueness follows from the fact that \( \hole \in \inhav{\grm, N_1 \Ar N} \) implies \( N_1 = N \).
  \item Case \( C = a(C_1, \dots, C_{\ralph(a)}) \):
    Let \( \ell = \ralph(a) \), \( k_i = \hn{C_i} \) for each \( i = 1, \dots, \ell \) and
    \[
    (T_{1,1}, \dots, T_{1,k_1}, T_{2,1}, \dots, T_{2,k_2}, \dots, T_{\ell,1}, \dots, T_{\ell,k_\ell}) = (T_1, \dots, T_k).
    \]
    Then \( C[T_1, \dots, T_k] = a(C_1[T_{1,1}, \dots, T_{1,k}], \dots, C_{\ell}[T_{\ell,1}, \dots T_{\ell,k_{\ell}}]) \).
    Since \( C[T_1, \dots, T_k] \in \inhav{\grm, N} \), there exists a rewriting sequence
    \[
    N \rew a(N_1, \dots, N_\ell) \rew^* a(C_1[T_{1,1}, \dots, T_{1,k_1}], \dots, C_{\ell}[T_{\ell,1}, \dots T_{\ell,k_{\ell}}]).
    \]
    Thus \( N_i \rew^* C_i[T_{i,1}, \dots, T_{i,k_i}] \), i.e.,~\( C_i[T_{i,1}, \dots, T_{i,k_i}] \in \inhav{\grm, N_i} \), for each \( i = 1, \dots, \ell \).
    By the induction hypothesis, there exist \( N_{i,1}, \dots, N_{i,k_i} \) such that \( C_i \in \inhav{\grm, N_{i,1} \dots N_{i,k_i} \Ar N_i} \) and \( T_{i,j} \in \inhav{\grm, N_{i,j}} \) for each \( j = 1, \dots, k_i \).
    Then (the rearrangement of) the family \( (N_{i,j})_{i \le \ell, j \le k_i} \) satisfies the requirement.

    We prove the uniqueness.
    Assume that both \( (N^1_{i,j})_{i,j} \) and \( (N^2_{i,j})_{i,j} \) satisfy the requirement.
    Then, for each \( m = 1, 2 \), there exist \(N^m_i\) (\(i=1,\dots,\ell\)) such that
    \begin{align*}
    N
    \rew a(N^m_1, \dots, N^m_\ell)
    & \rew^* a(C_1[N^m_{1,1}, \dots, N^m_{1,k_1}], \dots, C_\ell[N^m_{\ell,1}, \dots, N^m_{\ell,k_\ell}]) \\
    & \rew^* a(C_1[T_{1,1}, \dots, T_{1,k_1}], \dots, C_{\ell}[T_{\ell,1}, \dots, T_{\ell,k_{\ell}}]).
    \end{align*}
    Since \( \grm \) is unambiguous, \( N^1_i = N^2_i \) for each \( i = 1, \dots, \ell \).
    By the induction hypothesis, we have \( N^1_{i,j} = N^2_{i,j} \) for each \( i \) and \( j \).
\qedhere
  \end{itemize}
\end{proof}

\begin{lem}
\label{lem:second-unique-decomposition}
  Let \( \grm \) be a canonical unambiguous regular tree grammar and \( T \in \inhav{\grm, N} \).
Given \(E\) and \(P\), assume that \( \cbf{P}{E} \) and \( E\Hfill{P} = T \).
  Then \( E \) has a unique refinement \( \gctx \lhd E \) such that \( \p \gctx : N \) and \( \cbf{P}{\gctx} \).
\end{lem}
\begin{proof}
  We prove by induction on \( E \).
  \begin{itemize}
  \item
    Case \( E = \hhole_k^n[E_1, \dots, E_k] \):
    The sequence \( P \) can be decomposed as \( P = C \cdot P_1 \cdots P_k \) so that \( \cbf{P_i}{E_i} \) for each \( i = 1, \dots, k \).
    Furthermore \( \hn{C} = k \) and \( |C| = n \).
    We have \( E\Hfill{P} = C[E_1\Hfill{P_1}, \dots, E_k\Hfill{P_k}] \in \inhav{\grm, N} \).

    We prove the existence.
    By Lemma~\ref{lem:first-unique-decomposition}, there exists a family \( (N_i)_{i = 1, \dots, k} \) such that \( C \in \inhav{\grm, N_1 \dots N_k \Ar N} \) (i.e.,~\( C : \hhole_{N_1 \dots N_k \Ar N}^n \)) and \( E_i\Hfill{P_i} \in \inhav{\grm, N_i} \) for each \( i = 1, \dots, k \).
    By the induction hypothesis, for each \( i = 1, \dots, k \), there exists \( \gctx_i \lhd E_i \) such that \( \p \gctx_i : N_i \) and \( \cbf{P_i}{\gctx_i} \).
    Let \( \gctx \defeq \hhole_{N_1\dots N_k \Ar N}^n[\gctx_1, \dots, \gctx_k] \).
    Then \( \gctx \lhd E \), \( \p \gctx : N \), and \( \cbf{P}{\gctx} \).

    We prove the uniqueness.
    Assume \( \gctx^1 \) and \( \gctx^2 \) satisfy that \( \gctx^j \lhd E \), \( \p \gctx^j : N \) and \( \cbf{P}{\gctx^j} \) for \( j = 1, 2 \).
    Since \( \gctx^j \lhd E \) and \( \p \gctx^j : N \), \(\gctx^j\) must be of the form: \( \gctx^j =
       \hhole_{N_1^j \dots N_k^j \Ar N}^n[\gctx^j_1, \dots, \gctx^j_k] \) with \( \gctx^j_i \lhd E_i \) and \( \p \gctx^j_i : N^j_i \).
    Since \( \cbf{P}{\gctx^j} \), we have \( C \in \inhav{\grm, N_1^j \dots N_k^j \Ar N} \)
    and \(P_i : \gctx^j_i\) for \(i=1,\dots,k\).
    By Lemma~\ref{lem:second-order-substitution-lemma}, 
    \( \gctx^j_i\Hfill{P_i} \in \inhav{\grm, N_i^j} \) for each \( i = 1, \dots, k \) and \( j = 1, 2 \).
    By Lemma~\ref{lem:ref-basic}\eqref{item:ref-basic-sbst},
    \(E\Hfill{P} = \gctx^j\Hfill{P} = C[\gctx^j_1\Hfill{P_1}, \dots, \gctx^j_k\Hfill{P_k}] \in \inhav{\grm,N}\).
    Hence by Lemma~\ref{lem:first-unique-decomposition}, \( N_i^1 = N_i^2 \) for each \( i = 1, \dots, k \).
    By the induction hypothesis, \( \gctx_i^1 = \gctx_i^2 \) for each \( i \).
    Hence \( \gctx^1 = \gctx^2 \).
  \item
    Case \( E = a(E_1, \dots, E_{\ralph(a)}) \):
    The sequence \( P \) can be decomposed as \( P = P_1 \cdots P_\ralph(a) \) so that \( \cbf{P_i}{E_i} \) for each \( i = 1, \dots, \ralph(a) \).
    Then \( E\Hfill{P} = a(E_1\Hfill{P_1}, \dots, E_{\ralph(a)}\Hfill{P_{\ralph(a)}}) \).

    We prove the existence.
    Since \( N \rew^* a(E_1\Hfill{P_1}, \dots, E_{\ralph(a)}\Hfill{P_{\ralph(a)}}) \),
    there exists a rule \( N \rew a(N_1, \dots, N_{\ralph(a)}) \) such that
    \[
    N \rew a(N_1, \dots, N_{\ralph(a)}) \rew^* a(E_1\Hfill{P_1}, \dots, E_{\ralph(a)}\Hfill{P_{\ralph(a)}}).
    \]
    So \( N_i \rew^* E_i\Hfill{P_i} \) for each \( i = 1, \dots, \ralph(a) \).
    By the induction hypothesis, there exists \( \gctx_i \lhd E_i \) such that \( \p \gctx_i : N_i \) and \( \cbf{P_i}{\gctx_i} \).
    Let \( \gctx \defeq a(\gctx_1, \dots, \gctx_{\ralph(a)}) \).
    Then \( \gctx \lhd E \), \( \p \gctx : N \), and \( \cbf{P}{\gctx} \).

    We prove the uniqueness.
    Assume \( \gctx^1 \) and \( \gctx^2 \) satisfy that \( \gctx^j \lhd E \), \( \p \gctx^j : N \) and \( \cbf{P}{\gctx^j} \) for \( j = 1, 2 \).
    Since \( \gctx^j \lhd E \), \(\gctx^j\) must be of the form: 
    \( \gctx^j = a(\gctx^j_1, \dots, \gctx^j_{\ralph(a)}) \) with \( \gctx^j_i \lhd E_i \).
    By \( \p \gctx^j : N \), there exists a rule \( N \rew a(N_1^j, \dots, N_{\ralph(a)}^j) \) 
    such that \( \p \gctx^j_i : N_i^j \) for each \( i \).
    Since \(P : \gctx^j\), we have \(P_i : \gctx^j_i\) for each \(i\).
    By Lemmas~\ref{lem:second-order-substitution-lemma} and~\ref{lem:ref-basic}\eqref{item:ref-basic-sbst},
    we have 
    \( \gctx_i^j\Hfill{P_i} = E_i\Hfill{P_i} \in \inhav{\grm, N_i^j} \).
    Now we have
    \[
    N \rew a(N_1^j, \dots, N_{\ralph(a)}^j) \rew^* a(E_1\Hfill{P_1}, \dots, E_{\ralph(a)}\Hfill{P_{\ralph(a)}})
    \]
    for \( j = 1, 2 \).
    Since \( \grm \) is unambiguous, \( N_i^1 = N_i^2 \) for each \( i = 1, \dots, \ralph(a) \).
    By the induction hypothesis, \( \gctx_i^1 = \gctx_i^2 \) for each \( i \).
    Hence \( \gctx^1 = \gctx^2 \).
\qedhere
  \end{itemize}
\end{proof}

\subsubsection{Decomposition of \(\inhavn{\grm,N}\)}
\label{sec:grespecting-bijection}

To formally state the decomposition lemma, we prepare some definitions.
For a canonical unambiguous grammar $\grm = (\ralph, \nont, \rules)$, $N \in \nont$, $n \ge 1$, and $m \ge 1$, 
we define \(\tframe(\grm, N)\),
 \(\tseq(\grm,N)\), and \(\tcomp\) by: 
\begin{align*}
  \tframe(\grm, N) &\defeq \{ \gctx \mid (\gctx, P) = \closedgdecomp(T,N) \text{ for some } T \in \inhavn{\grm, N} \text{ and } P \}.
\end{align*}
\begin{align*}
  \tseq(\grm,N) &\defeq \set{ P \mid (\gctx, P) = \closedgdecomp(T, N) \text{ for some } T \in \inhav{\grm, N} }.
\end{align*}
\begin{align*}
  \tcomp &\defeq \set{ U \mid U : \hhole_{\kappa}^n, U \text{ is good for } m }.
\end{align*}
The set \(\tframe(\grm, N)\) consists of second-order contexts that are obtained by
decomposing trees of size \(n\), and \(\tseq(\grm,N)\) consists of
affine context sequences that match \(\gctx\). The set \(\tseq(\grm,N)\) is the set of contexts
that match the hole \(\hhole_{\kappa}^n\).

The following is the main result of this subsection.
\begin{lem}
\label{lem:grm-bijection}
  \begin{flalign*}
    \inhavnn{n}{\grm, N} \quad \cong &
    \coprod_{\gctx \in \tframe(\grm, N)} \prod_{i=1}^{\shn(\gctx)} \tcomp[\sh{\gctx}{i}].
  \end{flalign*}
\end{lem}

The lemma above 
is a direct consequence of 
typed versions of coproduct and product lemmas (Lemmas~\ref{lemma:union-grm} and
\ref{lemma:product} below).
The following coproduct lemma can be shown in a manner similar to Lemma~\ref{lemma:union-sig}:
\begin{lem}[Coproduct Lemma (for Grammar-Respecting Decomposition)]
  \label{lemma:union-grm}
  For any
  $n \geq 1$ and \(m \ge 1\), 
  there exists a bijection
  \begin{flalign*}
    \inhavnn{n}{\grm, N} \quad \cong &
    \coprod_{\gctx \in \tframe(\grm, N)} \tseq(\grm,N).
  \end{flalign*}
  such that \((\gctx,P)\) in the right hand side is mapped to \(\gctx\Hfill{P}\).
\end{lem}
\begin{proof}
  We define a function
  \[
  f:
  \coprod_{\gctx \in \tframe(\grm, N)} \tseq(\grm,N)
  \longrightarrow
  \inhavnn{n}{\grm, N}
  \]
  by \(f(\gctx,P) \defeq \gctx\Hfill{P}\),
  and a function
  \[
  g:
  \inhavnn{n}{\grm, N}
  \longrightarrow
  \coprod_{\gctx \in \tframe(\grm, N)} \tseq(\grm,N)
  \]
  by \(g(T) \defeq \closedgdecomp(T, N) \).
  
  Let us check that these are functions into the codomains:
  \begin{itemize}
  \item \( f(\gctx,P) \in \inhavn{\grm,N} \):
    Since \(P \in \tseq(\grm,N)\), there exists \( T \in \inhavn{\grm,N} \) such that \( (\gctx,P) = \closedgdecomp(T,N) \).
    By Lemmas~\ref{lemma:charphi} and \ref{lem:ref-basic}, we have \( f(\gctx,P) = \gctx\Hfill{P} = T \in \inhavn{\grm,N} \).
  \item \( g(T) \in \coprod_{\gctx \in \tframe(\grm, N)} \tseq(\grm,N) \):
    Obvious from the definitions of \( \tframe(\grm,N) \) and \( \tseq(\grm,N) \).
  \end{itemize}
  
  We have \(f(g(T)) = T\) by Lemmas~\ref{lemma:charphi}(\ref{charphi:split}) and \ref{lem:ref-basic}.
  Let \( (\gctx,P) \in \coprod_{\gctx \in \tframe(\grm, N)} \tseq(\grm,N) \).
  By definition, there exists \( T \in \inhavn{\grm,N} \) such that \( (\gctx,P) = \closedgdecomp{(T,N)} = g(T) \).
  Then
  \[
  g(f(\gctx,P))
  =g(f(g(T))
  =g(T)
  =(\gctx,P).
  \tag*{\qedhere}
  \]
\end{proof}

The following is a key lemma used for proving a typed version of the product lemma.
\begin{lem}\label{lem:ppind}
  For 
  a nonterminal \(N\),
  $n \geq 1$, \(m \ge 1\) and $ \gctx \in \tframe(\grm, N)$, let \( E \) be the unique second-order context such that \( \gctx \lhd E \).
  Then we have
  \begin{align*}
    \tseq(\grm,N) = \useq \cap \{ P \mid \cbf{P}{\gctx} \}.
  \end{align*}
\end{lem}
\begin{proof}
  The direction \(\subseteq\) is clear.
  We prove the converse.

  Let \(P\) be in the right hand side.
  Since \( \gctx \in \tframe(\grm, N)\), there exist \( T' \) and \( P' \) such that
\[
    \closedgdecomp(T',N) = (\gctx, P').
\]
  By the definition of \( \closedgdecomp(T',N) \), we have \( \p \gctx : N \).
  Since \(P \in \useq\), there exists \( T \) such that \( (E,P) = \closeddecomp(T) \) and thus, by Lemma~\ref{lemma:charphi},
\[
    \closeddecomp(E\Hfill{P}) = (E,P).
\]
  Since \( \closeddecomp(E\Hfill{P}) = (E,P) \), \( \gctx \lhd E \), \( \p \gctx : N \),
  and \( \cbf{P}{\gctx} \), by the definition of \( \closedgdecomp \), we have \( \closedgdecomp(E\Hfill{P}, N) = (\gctx,P) \).
  By Lemmas~\ref{lem:second-order-substitution-lemma}\eqref{item:second-order-substitution-lemma-sbstall}
  and~\ref{lem:ref-basic}\eqref{item:ref-basic-sbst}, \(E\Hfill{P} = \gctx\Hfill{P} \in \inhav{\grm,N}\).
  Hence \( P \in \tseq(\grm,N) \).
\end{proof}

The following is the typed version of the product lemma, which follows from
Lemmas~\ref{lem:product} and~\ref{lem:ppind}.
\begin{lem}[Product Lemma (for Grammar-Respecting Decomposition)]
  \label{lemma:product}
  For 
  any nonterminal \(N\),
  $n \geq 1$,
  \(m \ge 1\) and
  \( \gctx \in \tframe(\grm,N) \), we have
  \[
  \tseq(\grm,N) = \prod_{i=1}^{\shn(\gctx)} \tcomp[\sh{\gctx}{i}].
  \]
\end{lem}
\begin{proof}
Let \( E \) be the unique second-order context such that \( \gctx \lhd E \).
By Lemmas~\ref{lem:product} and~\ref{lem:ppind}, we have:
\begin{align*}
\tseq(\grm,N) &= \useq \cap \{ P \mid \cbf{P}{\gctx} \}\\
&= (\textstyle\prod_{i=1}^{\shn(E)} \ucomp[\sh{E}{i}])\cap (\textstyle\prod_{i=1}^{\shn(\gctx)} \set{U \mid U\COL \sh{\gctx}{i}})\\
&= \textstyle\prod_{i=1}^{\shn(\gctx)} (\ucomp[\sh{E}{i}]\cap \set{U \mid U\COL \sh{\gctx}{i}})\\
&= \textstyle\prod_{i=1}^{\shn(\gctx)} \tcomp[\sh{\gctx}{i}].
\tag*{\qedhere}
\end{align*}
\end{proof}

Lemma~\ref{lem:grm-bijection} is an immediate 
corollary of Lemmas~\ref{lemma:union-grm} and \ref{lemma:product}.
 
\subsection{Each Component Contains the Subcontext of Interest}
\label{sec:adjusting}
In Section~\ref{sec:rtl-decomposition},
we have shown that the set \( \inhavn{\grm, N} \) of trees can be decomposed as:
\[
\inhavnn{n}{\grm, N} \quad \cong 
\coprod_{\gctx \in \tframe(\grm, N)} \prod_{i=1}^{\shn(\gctx)} \tcomp[\sh{\gctx}{i}],
\]
assuming that \(\grm\) is canonical and unambiguous.
In this subsection, we further assume that \(\grm\) is essentially strongly connected,
and prove that, for each tree context \( S \in \ctxinf{\grm} \), every component ``contains'' \( S \), i.e., there exists \(U\in\; \tcomp[\sh{\gctx}{i}]\) such that
\(S\subc U\)
if \( m \) is sufficiently large, say \( m \ge m_0 \) (where \( m_0 \) depends on \( |S| \)).
More precisely, the goal of this subsection is to prove the following lemma.
\begin{lem}
  \label{lemma:adjust}
  Let \( \grm = (\ralph,\nont,\rules) \) be an unambiguous, essentially strongly-connected grammar in canonical form and \((S_n)_{n \in \nat} \) be a family of \linear{contexts} in \(\ctxinf{\grm}\) such that \( |S_{n}| = \order(n) \).
  Then there exist integers \(b,c \ge 1\) that satisfy the following:
  For any \(N \in \nontinf\), \( n\ge 1 \), \( m \ge b \), \( \gctx \in \tframeNM{n}{cm}(\grm,N) \)
  and \(i \in \set{1,\dots,\shn(\gctx)}\), 
  there exists \( \pU \in \tcompCM{\sh{\gctx}{i}}{cm}\) such that $S_{m} \subc \pU$.
\end{lem}
The rest of this \new{subsection} is devoted to a proof of the lemma above.
The idea of the proof is as follows.
Assume \( S \in \inhav{\grm, N_1 \Ar N_2} \) (\( N_1, N_2 \in \nontinf \)) and \( \hhole^n_{\f} = \sh{\gctx}{i} \).
Recall that\tk{where is the first occurrence of the claim below?}
\[
\tcomp[\hhole^n_{\f}] = \{ U \mid U \in \inhav{\grm, \f}, \; |U| = n, \mbox{ and \( U \) is good for \( m \)} \}.
\]
It is not difficult to find a context \( U \) that satisfies both \( U \in \inhav{\grm, \f} \) and \( S \subc U \).
For example, assume that \( \f = N_0 \Ar N_3 \) (\( N_0, N_3 \in \nontinf \)).
Then, since the grammar is assumed to be essentially strongly-connected, 
there exist \( S_{0,1} \in \inhav{\grm, N_0 \Ar N_1} \) and \( S_{2,3} \in \inhav{\grm, N_2 \Ar N_3} \) 
and then \( U \defeq S_{2,3}[S[S_{0,1}]] \) satisfies \( U \in \inhav{\grm, \f} \) and \( S \subc U \).
What is relatively difficult is to show that \( S_{2,3} \) and \( S_{0,1} \) can be chosen so that they meet the required size constraints (i.e.,~\( |U| = n \) and \( U \) is good for \( m \)).

The following is a key lemma, which states 
that any essentially strongly connected grammar \(\grm\) is
\emph{periodic} in the sense that there is a constant \(c\) (that depends on \(\grm)\)
and a family of constants \(d_{N,N'}\) such that,
 for each \(N,N'\in \nontinf\), 
and  sufficiently large \(n\),
  \( \inhavn{\grm, N \Ar N'} \neq \emptyset \) if and only if
  \(n\equiv d_{N,N'} \mod c\).

  \begin{lem}
\label{lem:period-grammar2}
  Let \( \grm = (\ralph, \nont, \rules) \) be a regular tree grammar.
  Assume that \( \grm \) is essentially strongly-connected and \( \card{\inhav{\grm}} = \infty \).
  Then there exist constants \( n_0, c > 0 \) and a family \( ( d_{N,N'} )_{N, N' \in \nontinf} \) of natural numbers \( 0 \le d_{N,N'} < c \) that satisfy the following conditions:
 \begin{enumerate}
  \item\label{item:period-grammar2-every}
       For every \( N, N' \in \nontinf \), if \( \inhavn{\grm, N \Ar N'} \neq \emptyset \), then \( n \equiv d_{N,N'} \mod c \).
  \item\label{item:period-grammar2-suff}
       The converse of \eqref{item:period-grammar2-every} holds for sufficiently large \( n \):
       for any \( N, N' \in \nontinf \) and \( n \ge n_0 \), 
       if \( n \equiv d_{N,N'} \mod c \) then \( \inhavn{\grm, N \Ar N'} \neq \emptyset \).
  \item\label{item:period-grammar2-id}
       \( d_{N,N} = 0 \) for every \( N \in \nontinf \).
  \item\label{item:period-grammar2-comp}
       \( d_{N,N'} + d_{N', N''} \equiv d_{N,N''} \mod c \) for every \( N, N', N'' \in \nontinf \).
 \end{enumerate}
\end{lem}

The proof of the above lemma is rather involved;
we defer it to  Appendix~\ref{sec:periodicity}.
We give some examples below,  to clarify what the lemma means.
\begin{exa}
  Consider the grammar \(\grm_1\) consisting of the following rewriting rules:
  \[
  A \rew \Ta(B)\qquad B \rew \Tb(A)\qquad A \rew \Tc(C)\qquad C\rew \Tc
\]
Then \(\nontinf=\set{A,B}\) and the conditions of the lemma above hold for:
\[n_0=0\qquad c=2\qquad d_{A,A}=d_{B,B}=0\qquad d_{A,B}=d_{B,A}=1.\]
In fact, \(\inhav{\grm_1,A\Ar A}=\set{(\Ta\Tb)^k\hole \mid k\geq 0}\),
\(\inhav{\grm_1,B\Ar B}=\set{(\Tb\Ta)^k\hole \mid k\geq 0}\),
\(\inhav{\grm_1,B\Ar A}=\set{(\Ta\Tb)^k\Ta\hole \mid k\geq 0}\),
and 
\(\inhav{\grm_1,A\Ar B}=\set{(\Tb\Ta)^k\Tb\hole \mid k\geq 0}\).
Here, since the arities of \(\Ta\) and \(\Tb\) are \(1\), we have
used regular expressions to denote \linear{contexts}.

Consider the grammar \(\grm_2\),
obtained by adding the following rules to the grammar above.
\[ A\rew \Ta(A_1) \qquad A_1\rew \Ta(A_2) \qquad A_2\rew \Ta(A).\]
Then, the conditions of the lemma above hold for
\(n=6\), \(c=1\),  \(d_{N,N'}=0\) for \(N,N'\in \set{A,A_1,A_2,B}\). Note that
\(\inhav{\grm_2,A_2\Ar A_1} = \set{\Ta^2(\Ta^3|\Ta\Tb)^k\Ta^2\hole\mid k\geq 0}\).
\qed
\end{exa}
\begin{rem}
  With some additional assumptions on a grammar,
  Lemma~\ref{lem:period-grammar2} above can be easily proved.
  For example, consider a canonical, unambiguous and essentially strongly-connected grammar \( \grm \) 
and assume that 
(i) \(\grm\) is \(N\)-aperiodic for some \(N\)
and (ii) 
there exist a \( 2 \)-context \( C \) and nonterminals \( N, N_1, N_2 \in \nontinf \) such that \( N \rew_{\grm}^* C[N_1,N_2] \). 
  Then Lemma~\ref{lem:period-grammar2} for the grammar \( \grm \) trivially holds with \( c = 1 \).
  In fact, this simpler approach is essentially what we adopted in the conference version~\cite{SAKT17FOSSACS} of this article.
\end{rem}

Using the lemma above, we prove that
for every \( S \in \ctxinf{\grm} \) and any sufficiently large
  \(n\) such that \(\inhavnn{n}{\grm, \f}\neq \emptyset\),
  we can find a context \(U\in \inhavnn{n}{\grm, \f}\) such that \( S \subc U \)
  (Lemma~\ref{lem:period-sub} below).
\begin{lem}
\label{lem:period-sub}
  Let \( \grm = (\ralph, \nont, \rules) \) be a regular tree grammar in canonical form.
  Assume that \( \grm \) is unambiguous and essentially strongly-connected and \( \card{\inhav{\grm}} = \infty \).
  Then there exists a constant \( n_0 \in \nat \) that satisfies the following condition:
  For every
  \begin{itemize}
  \item \( S \in \ctxinf{\grm} \),
  \item \( n \ge n_0  + |S| \), and
  \item \( \f = ({}\Ar N') \) or \( (N \Ar N') \) where \( N, N' \in \nontinf \),
  \end{itemize}
  if \( \inhavnn{n}{\grm, \f} \neq \emptyset \), then there exists \( U \in \inhavnn{n}{\grm, \f} \) with \( S \subc U \).
\end{lem}
\begin{proof}
  First, let us choose the following constants:
  \begin{itemize}
  \item \( m_0 \in \nat \) such that, for every \( N, N' \in \nont^{\mathrm{inf}} \),
	there exists \( S_{N,N'} \in \inhav{\grm, N \Ar N'} \) with \( |S_{N, N'}| \le m_0 \).
	The existence of \(m_0\) is a consequence of essential strong-connectivity and of finiteness of \( \nont \).
  \item \( m_1 \) which is the constant \( n_0 \) of Lemma~\ref{lem:period-grammar2}.
  \item \( m_2 \in \nat \) such that, for every \( N \in \nont \backslash \nont^{\mathrm{inf}} \)
	and every \( T \in \inhav{\grm, N} \), we have \( |T| < m_2 \).
	The existence of \(m_2\) follows from the fact that 
	\( \bigcup_{N \in \nont \backslash \nont^{\mathrm{inf}}} \inhav{\grm, N} \) is a finite set.
  \end{itemize}
  Let \( c \) 
  and \( (d_{N, N'})_{N, N' \in \nont^{\mathrm{inf}}} \) be the constant and the family obtained by Lemma~\ref{lem:period-grammar2}.

  We define \( n_1 \defeq m_0 + m_1 \) and \( n_2 \defeq m_0 + m_1 + \mr m_2 + 1 \), where \(\mr \defeq \max(\image{\ralph})\). 
  Below we shall show:
  (i) the current lemma for the case \( \f = (N \Ar N') \), by setting
   \(n_0 \defeq n_1\) 
   and then (ii) the lemma for the case \( \f = (\Ar N') \),
   by setting 
   \(n_0 \defeq n_2\); 
 we use (i) to show (ii).
 The whole lemma
 then follows immediately  from (i) and (ii) with \(n_0 \defeq \max\set{n_1,n_2} \, (=n_2)\).
\begin{itemize}
\item
  Case (i):
  We define \( n_0 \defeq n_1 = m_0 + m_1 \).
  Assume that: \( S \in \inhav{\grm, N_1 \Ar N_2} \) where \( N_1, N_2 \in \nont^{\mathrm{inf}} \);
  \( n \ge n_0 + |S| \); \( \f = (N \Ar N') \) where \( N, N' \in \nont^{\mathrm{inf}} \);
  and \( \inhavnn{n}{\grm, \f} \neq \emptyset \).
  Let \( S_1 \in \inhav{\grm, N_2 \Ar N'} \) with \( |S_1| \le m_0 \).
  Then \( S_1[S] \in \inhav{\grm, N_1 \Ar N'} \).
  It suffices to show that \( \inhavnn{n - |S_1[S]|}{\grm, N \Ar N_1} \neq \emptyset \).
  By Lemma~\ref{lem:period-grammar2}\eqref{item:period-grammar2-every}\eqref{item:period-grammar2-comp},
\[
\mspace{20mu}
  \begin{aligned}
    n &\equiv d_{N, N'} \equiv d_{N, N_1} + d_{N_1, N_2} + d_{N_2, N'} && \bmod c
  &&(\because \inhavnn{n}{\grm, \f} = \inhavnn{n}{\grm, N\Ar N'} \neq \emptyset)
  \\
    |S_1| &\equiv d_{N_2, N'} && \bmod c
  &&(\because (S_1 \in)\, \inhavnn{|S_1|}{\grm, N_2 \Ar N'} \neq \emptyset)
  \\
    |S| &\equiv d_{N_1, N_2} && \bmod c
  &&(\because (S \in)\, \inhavnn{|S|}{\grm, N_1 \Ar N_2} \neq \emptyset)
\end{aligned}
\]
  and thus
  \[
  n - |S_1[S]| \equiv d_{N, N_1} \mod c.
  \]
  Since \(n \ge m_0 + m_1 + |S|\), we have \( n - |S_1[S]| \ge m_1 \),
  and hence by Lemma~\ref{lem:period-grammar2}\eqref{item:period-grammar2-suff}, we have
  \( \inhavnn{n - |S_1[S]|}{\grm, N \Ar N_1} \neq \emptyset \).
\item
  Case (ii):
  We define \( n_0 \defeq n_2 = m_0 + m_1 + \mr m_2 + 1 \).
  Assume that: \( S \in \inhav{\grm, N_1 \Ar N_2} \) where \( N_1, N_2 \in \nont^{\mathrm{inf}} \);
  \( n \ge n_0 + |S| \); \( \f = ({} \Ar N') \) where \( N' \in \nont^{\mathrm{inf}} \);
  and \( \inhavnn{n}{\grm, \f} \neq \emptyset \).
  Let \( T \in \inhavnn{n}{\grm, \f} \).
  Since arities of terminal symbols are bounded by \( \mr \) and \( |T| \ge \mr m_2 + 1 \),
  there exists a subtree \( T_0 \subc T \) such that \( m_2 \le |T_0| \le \mr m_2 + 1 \),
  which can be shown by induction on tree \(T\).
  Let \( U \) be a \linear{context} such that \( T = U[T_0] \).
  Since \( \grm \) is canonical and unambiguous, by Lemma~\ref{lem:first-unique-decomposition},
  \( U \in \inhav{\grm, N \Ar N'} \) and \( T_0 \in \inhav{\grm, N} \) for some \( N \in \nont \).
  Since \( m_2 \le |T_0| \le \mr m_2 + 1 \), we have \( N \in \nont^{\mathrm{inf}} \) and \( |U| \ge m_0 + m_1 + |S| \).
  By using the case \((1)\), since \((U \in)\, \inhavnn{|U|}{\grm,N \Ar N'} \neq \emptyset\),
  there exists \(U' \in \inhavnn{|U|}{\grm,N \Ar N'}\) with \(S \subc U'\).
  Let \(T' \defeq U'[T_0]\); then \(T' \in \inhavn{\grm,\Ar N'}\) and \(S \subc T'\).
\qedhere
\end{itemize}
\end{proof}

\newcommand{\widehatu}[3]{\widehat{\mathcal{U}}_{{#1},{#2}}^{#3}}

We prepare another lemma.
\begin{lem}
\label{lem:remainsHigh}
If \(\decomp(T) = (U,E,P)\), \(\ell \in \set{1,\dots,\shn(E)}\),
\(\prj{P}{\ell}\) is a \linear{context}, and
the second-order hole \(\sh{E}{\ell}\) occurs in \(E\) in the form \((\sh{E}{\ell})[E']\),
then \(|E'| \ge m\).
\end{lem}
\begin{proof}
  By straightforward induction on \(|T|\) and case analysis of
  Equation~\eqref{eq:defOfDecomp} in the definition of \(\decomp\).
  We can use Lemma~\ref{lem:size-of-E} in the second case of
  Equation~\eqref{eq:defOfDecomp} when \(\ell=1\);
all the other cases immediately follow from induction hypothesis.
\end{proof}

We are now ready to prove the main lemma of this subsection.
\nk{I have removed an intermediate lemma and merged it with
  the proof of the main lemma below; I do understand the motivation for the intermediate lemma,
  but renaming the constants \(m_0,n_0,...\) so many times seemed to
  make the proof more difficult to read.}
\begin{proof}[Proof of Lemma~\ref{lemma:adjust}]
  Since \(S_0 \in \ctxinf{\grm}\), we have
\(\nontinf \neq \emptyset\) and hence \(\card{\inhav{\grm}} = \infty\).
  Let \( n_0\) be the constant of Lemma~\ref{lem:period-sub}.
  Let \( c_1 \) be a positive integer such that \( |S_n| \le c_1 n \) for every \( n \in \nat \).
  We define \( c \defeq (2r+1)c_1 \) (recall that
  \(r\) is the largest arity of \(\Sigma\)) and choose \(b\) so that \(b \ge n_0\) and
  \( b > |T| \) for any \( T \in \bigcup_{N \in \nont\backslash\nontinf} \inhav{\grm,N} \).

    Assume that \(N \in \nontinf\), \( n \ge 1 \), \( m \ge b \), \( \gctx \in \tframeNM{n}{cm}(\grm,N) \)
  and \( i \in \set{1, \dots, \shn(\gctx)} \).
  Let \( \sh{\gctx}{i} = \hhole^{n'}_{\f'} \).
  We need to show that there exists
  \( \pU \in \tcompCM{\hhole^{n'}_{\f'}}{cm}\) such that $S_{m} \subc \pU$.

    Since \( \gctx \in \tframeNM{n}{cm}(\grm,N) \), 
there exist \(T \in \inhavn{\grm, N}\) and \(P\) such that
  \((\gctx, P) = \closedgdecompm{cm}(T,N)\);
 hence \(\prj{P}{i} \in \tcompCM{\sh{\gctx}{i}}{cm} = \tcompCM{\hhole^{n'}_{\f'}}{c m}\).
 The affine context \(\prj{P}{i}\) must be of the form
 \( a(U_1, \dots, U_{\ralph(a)}) \). Since \(\prj{P}{i}\)  is good for \(cm\),
 we have \(|\prj{P}{i}|\geq cm = (2r+1)c_1m\geq 2rc_1m+1\).
 Since \(\Sigma(a)\leq r\), there exists \(j\leq \Sigma(a)\) such that
 \(|U_j|\geq 2c_1m\).
\new{We have \(\f' = (N' \Ar N'')\) or \((\Ar N'')\) for some \(N', N'' \in \nont\).}
   Since \( \prj{P}{i} = a(U_1, \dots, U_{\ralph(a)}) \in \inhav{\grm, \f'} \), there exist \(N_1, \dots, N_{\ralph(a)} \in \nont\) such that
  \[
  N'' \rew a(N_1, \dots, N_{\ralph(a)}) \rew^* T'
  \]
  where \( T' = \prj{P}{i}[N'] \) if \( \f' = (N' \Ar N'') \),
 and \( T' = \prj{P}{i} \) if \( \f' = (\Ar N'') \).
\new{Let \(\f\) be \(N'\Ar N_j\) if \(U_j\) is a \linear{context}, and \(\Ar N_j\) otherwise.}
Then \(U_j\in \inhavnn{|U_j|}{\grm,\f}\neq \emptyset\).\asd{
To show this sentence, I changed the definition of \(\f\).
I could not check if \(U_j\) is always a \linear{context} in the case that \(\f' = (N' \Ar N'')\).
If \(U_j\) is always a \linear{context}, the change of the definition is meaningless but also harmless,
and if \(U_j\) is not necessarily a \linear{context}, the change is necessary;
anyway, the change is harmless and the current proof is correct.}
In order to apply Lemma~\ref{lem:period-sub} (for \(S=S_m\) \new{and \(n= |U_j|\)}),
we need to check the conditions: (i) \(|U_j| \geq  n_0+|S_m|\)
and (ii) \(\f\) consists of only nonterminals in \(\nontinf\).
Condition (i) follows immediately from \(|U_j|\geq 2c_1m
\geq m+c_1m \geq n_0+|S_m|\). As for (ii), it suffices to check that
\(\inhav{\grm,N_j}\) and \(\inhav{\grm,N'}\) 
contain a tree whose size is no less than \(b\)
(where the condition on \(\inhav{\grm,N'}\) is required only if
\new{\(U_j\) is a \linear{context}}).
The condition on \(\inhav{\grm,N_j}\) follows 
 from \(|U_j|\geq 2c_1m\geq m\geq b\). 
 If 
\new{\(U_j\) is a \linear{context},}
 by Lemma~\ref{lem:remainsHigh}, 
 there exist \(S\) and \(T''\)
  such that:
\[
T=S[(\prj{P}{i})[T'']]
\qquad
|T''| \ge cm\ge b
\qquad
T'' \in \inhav{\grm,N'}
\]
as required.

Thus, we can apply Lemma~\ref{lem:period-sub}
and obtain \(U_j' \in \inhav{\grm,\f}\) such that \(S_m \subc U_j'\) and \(|U_j'|=|U_j|\).
  Since \(\prj{P}{i} = a(U_1,\ldots,U_{\ralph(a)})\) is good for \( c m \),
  \( U \defeq a(U_1, \dots, U_{j-1}, U'_j, U_{j+1}, \dots, U_{\ralph(a)}) \) is also good for \( c m \).
  Obviously \( |U| = |\prj{P}{i}| \) and thus \( U \in
\tcompCM{\hhole^{n'}_{\f'}}{cm}\). 
  Since \( S_m \subc U_j' \) and \( U_j' \subc U \), we have \( S_m \subc U \) as required.
\end{proof}

\subsection{Main Proof}\label{sec:proofimfmain}
Here, we give a proof of Theorem~\ref{thm:imt}.
Before the proof, we prepare a simple lemma.
We write \( \cszon{\ralph} \) for the set of \affine{contexts} over \( \ralph \) of size at most \( n \).
Lemma~\ref{lem:upBound} below gives an upper bound of \( \card{\cszon{\ralph}} \).
A more precise bound can be obtained by using a technique of analytic combinatorics
such as Drmota--Lalley--Woods theorem~(\cf \cite{Flajolet}, Theorem~VII.6), 
but the rough bound provided by the lemma below is sufficient for our purpose.

\begin{lem}
\label{lem:upBound}
  For every ranked alphabet \( \ralph \), there exists a real constant \( \gamma > 1 \) such that
  \[
  \card{\cszon{\ralph}} \le \gamma^n
  \]
  for every \( n \ge 0 \).
\end{lem}
\begin{proof}
Let \(A\) be the set of symbols:
  \(
\{ \grave{a}, \acute{a} \mid a \in \dom(\ralph) \} \cup \{\, \square \}
  \).
  Intuitively \( \grave{a} \) and \( \acute{a} \) are opening and closing tags of an XML-like language.

  We can transform an affine context \(U\) to its XML-like string representation \(U^\dagger\in A^*\) by:
  \begin{align*}
    (a(U_1, \dots, U_{\ralph(a)}))^\dagger &\defeq \grave{a}\, U_1^\dagger \cdots U_{\ralph(a)}^\dagger \,\acute{a} \\
    \hole^\dagger &\defeq \square.
  \end{align*}
Obviously, \((\cdot)^\dagger\) is injective. Furthermore, \(|U^\dagger| = 2|U|\) if \(U\) is a
\(0\)-context (i.e., a tree), and \(|U^\dagger| = 2|U|+1\) if \(U\) is a \linear{context}
(note that the size of the hole \(\hole\) is zero, but its word representation is of length \( 1 \)). 
Thus, for \(n>0\), we have
\[
  \card{\cszon{\ralph}} \le 
  \sum_{i = 0}^{2n+1} (\card{A})^i 
\le 
  \sum_{i = 0}^{3n} (\card{A})^i 
\le (\card{A} + 1)^{3n} = ((\card{A}+1)^3)^n.
\]
If \(n=0\), then   \(\card{\cszon{\ralph}}=1=((\card{A}+1)^3)^n\),
as \(\cszon{\ralph}\) is the singleton set \(\set{\hole}\). 
Thus, the required result holds for \(\gamma=(\card{A}+1)^3\).
\end{proof}

The following lemma is a variant of 
Theorem~\ref{thm:imt}, specialized to a canonical grammar.
\begin{lem}
\label{lem:imt}
Let $\grm = (\ralph, \nont, \rules)$ be a canonical, unambiguous, and essentially strongly-connected regular tree grammar
such that \(\card{\inhav{\grm}}=\infty\),
and \((S_n)_{n \in \nat}\) be a family of \linear{contexts} in \(\ctxinf{\grm}\)
such that $|S_n| = \order(n)$.
Then there exists a real constant \(\cn > 0\) such that 
for any \(N \in \nontinf\), 
\[
 \limd_{n \rightarrow \infty} \frac{\card{\{ \tr \in \inhavn{\grm,N} \mid
 S_{\ceil{\cn \log n}} \subc \tr
 \}}}{\card{\inhavn{\grm,N}}} = 1.
\]
\end{lem}
\begin{proof}
  \nk{The previous proof is found in ``oldimdproof.tex''.}
    The overall structure of the proof is the same as that of Proposition~\ref{prop:monkey}.
  Let \(Z_n\) be
  \[
1- \frac{\card{\{ \tr \in \inhavn{\grm,N} \mid
 S_{\ceil{\cn \log n}} \subc \tr
 \}}}{\card{\inhavn{\grm,N}}}
=
\frac{\card{\{ \tr \in \inhavn{\grm,N} \mid
 S_{\ceil{\cn \log n}} \not\subc \tr
 \}}}{\card{\inhavn{\grm,N}}}.
  \]
It suffices to show that \(Z_n\) converges to \(0\).  

By Lemma~\ref{lem:grm-bijection},
we have
\[ \inhavn{\grm,N}\cong
\coprod_{\gctx \in \tframeNM{n}{m}(\grm, N)}
\prod_{i=1}^{\shn(\gctx)} \tcompCM{\sh{\gctx}{i}}{m}
  \]
  for any \(m>0\).
Thus, we have
\begin{equation}
Z_n \leq
\frac{\sum_{\gctx \in \tframeNM{n}{m}(\grm, N)}
  \prod_{i=1}^{\shn(\gctx)}
\card{\set{U\in \tcompCM{\sh{\gctx}{i}}{m}\mid
S_{\ceil{\cn \log n}}    \not\subc U}}
}
     {\sum_{\gctx \in \tframeNM{n}{m(n)}(\grm, N)} \prod_{i=1}^{\shn(\gctx)} \card{\tcompCM{\sh{\gctx}{i}}{m}}}. \label{zn}
\end{equation}
     Let \(b, c\geq 1\) be the numbers in Lemma~\ref{lemma:adjust}.
     Then, by the lemma, for any \(n\) \new{and \(p\)} such that \(\ceil{\cn \log n}\geq b\),
     each \(\tcompCM{\sh{\gctx}{i}}{c\ceil{\cn \log n}}\) contains at least one \(U\)
     that has \(S_{\ceil{\cn \log n}}\) as a subcontext.
     Thus, for \(m=c\ceil{\cn\log n}\), 
 \(    \card{\set{U\in \tcompCM{\sh{\gctx}{i}}{m}\mid
     S_{\ceil{\cn \log n}}    \not\subc U}}\) is bounded above  by:
   \[
   \begin{array}{ll}
  \card{\tcompCM{\sh{\gctx}{i}}{m}}-1 
  =   (1-\displaystyle\frac{1}{\card{\tcompCM{\sh{\gctx}{i}}{m}}})\card{\tcompCM{\sh{\gctx}{i}}{m}}
  \leq
   (1-\displaystyle\frac{1}{\gamma^{r m}})\card{\tcompCM{\sh{\gctx}{i}}{m}}.
   \end{array}
   \]
   Here \(\gamma\) is the constant (that only depends on \(\Sigma\))
   of Lemma~\ref{lem:upBound}, and \(r\) is the largest arity of \(\Sigma\). In
   the last inequality, we have used the fact that 
   \(\tcompCM{\sh{\gctx}{i}}{m}\subseteq \cszon[\mr m]{\ralph}\).

   By using the upper-bound above, Equation \eqref{zn}, and \(\shn(\gctx)\geq\frac{n}{2rm}\)
   (Lemma~\ref{lemma:partition}),
    we have:
   \[
   \begin{array}{l}
Z_n \leq
\dfrac{\sum_{\gctx \in \tframeNM{n}{m}(\grm, N)}
  \prod_{i=1}^{\shn(\gctx)}
(1-\displaystyle\frac{1}{\gamma^{r m}})\card{\tcompCM{\sh{\gctx}{i}}{m}}}
     {\sum_{\gctx \in \tframeNM{n}{m(n)}(\grm, N)} \prod_{i=1}^{\shn(\gctx)} \card{\tcompCM{\sh{\gctx}{i}}{m}}} \\
= \dfrac{\sum_{\gctx \in \tframeNM{n}{m}(\grm, N)}
(1-\displaystyle\frac{1}{\gamma^{r m}})^{\shn(\gctx)}  \prod_{i=1}^{\shn(\gctx)}
\card{\tcompCM{\sh{\gctx}{i}}{m}}}
     {\sum_{\gctx \in \tframeNM{n}{m(n)}(\grm, N)} \prod_{i=1}^{\shn(\gctx)} \card{\tcompCM{\sh{\gctx}{i}}{m}}} \\
\leq \dfrac{\sum_{\gctx \in \tframeNM{n}{m}(\grm, N)}
(1-\displaystyle\frac{1}{\gamma^{r m}})^{\frac{n}{2rm}} \prod_{i=1}^{\shn(\gctx)}
\card{\tcompCM{\sh{\gctx}{i}}{m}}}
     {\sum_{\gctx \in \tframeNM{n}{m(n)}(\grm, N)} \prod_{i=1}^{\shn(\gctx)} \card{\tcompCM{\sh{\gctx}{i}}{m}}}
= 
(1-\dfrac{1}{\gamma^{r m}})^{\frac{n}{2rm}}
   \end{array}
   \]
   for any \(n\) \new{and \(p\)} such that \(\ceil{\cn \log n}\geq b\) and \(m=c\ceil{\cn\log n}\).
   
   It remains to choose \(\cn\) so that
   \((1-\dfrac{1}{\gamma^{r m}})^{\frac{n}{2rm}} =
   (1-\dfrac{1}{\gamma^{r c\ceil{\cn\log n}}})^{\frac{n}{2rc\ceil{\cn\log n}}}\)
   converges to \(0\).
Let us choose positive real numbers \(a\), \(p\), and \(q\) so that 
\(p\) and \(q\) satisfy the following conditions for every \(n \ge a\):
\begin{align}
  {\cn \log n} &\ge b
\label{eq:foradjust}
\\
\gamma^{\mr c\ceil{\cn \log n}} &\le n^q
\label{eq:toconv}
\\
q &< 1.
\label{eq:forconv}
\end{align}
For example,
we can choose \(a,\, p,\, q\) as follows:
\begin{align*}&
a = \max\set{\gamma^{b({\mr c+2})}, \gamma^{\mr c({\mr c+2})}}
\qquad
\cn = \fr{1}{(\mr c+2)\log\gamma}
\qquad
q = \fr{\mr c+1}{\mr c+2}
\end{align*}
In fact, condition \eqref{eq:toconv} follows from:
\[
\begin{array}{l}
\mr c\ceil{\cn \log n} - \log_\gamma n^q
\leq
\mr c(\fr{1}{(\mr c+2)\log\gamma} \log n+1) - \fr{\mr c+1}{\mr c+2}\frac{\log n}{\log{\gamma}}\\
= \frac{1}
{(rc+2)\log\gamma}
(rc\log n+ rc(rc+2)\log \gamma - (rc+1)\log n)\\
= \frac{1}
{(rc+2)\log\gamma}
(rc(rc+2)\log \gamma - \log n)\\
\leq \frac{1}
{(rc+2)\log\gamma}
(rc(rc+2)\log \gamma - \log a)\leq 0.\\
\end{array}
\]

Thus, for \(n\geq a\), we have:
\[
(1-\dfrac{1}{\gamma^{r c\ceil{\cn\log n}}})^{\frac{n}{2rc\ceil{\cn\log n}}}
\leq
(1-\dfrac{1}{n^q})^{\frac{n}{2rc\ceil{\cn\log n}}}
\leq 
\left(\left(1-\fr{1}{n^{q}}\right)^{\fr{n}{\log n}}\right)^{\fr{1}{3\mr c p}}.
\]
Since \(\lim_{n \rightarrow \infty}  \left(1-\fr{1}{n^{q}}\right)^{\fr{n}{\log n}} = 0\),
we have \(\limd\limits_{n\rightarrow\infty} \zps{n} =0\) as required.
\end{proof}
\begin{rem}
\label{rem:convergence}
In the proof above, we used
the fact that if \(0<q<1\) then
\[
\lim_{n\tends \infty}\left(1-\fr{1}{n^q}\right)^{\fr{n}{\log n}} = 0.
\]
We also remark that if \(q\ge 1\) then
\[
\lim_{n\tends \infty}\left(1-\fr{1}{n^q}\right)^{\fr{n}{\log n}} = 1.
\]
Thus, \(p\) should be chosen to be sufficiently small so that Equation (\ref{eq:toconv}) in the proof
holds for some \(q<1\).
\end{rem}

We are now ready to prove Theorem~\ref{thm:imt}.
We restate the theorem.
\renewcommand{\thethmNN}{\ref{thm:imt}}
\begin{thmNN}[\thmstatementIMTcap]
\thmstatementIMT
\end{thmNN}
\begin{proof}
By Lemma~\ref{lem:canonical-grammar}, there exists a canonical, unambiguous and essentially
strongly-connected grammar \(\grm' = (\ralph, \nont', \rules')\) and 
a family \( ( \nset_N )_{N \in \nont} \) of subsets \( \nset_N \subseteq \nont' \)
such that
\( \inhav{\grm,N} = \biguplus_{N' \in \nset_N} \inhav{\grm',N'} \) for every \(N \in \nont\)
and \( \ctx{\grm} \subseteq \ctxinf{\grm'} \).
Let \(\nset'_N \defeq \nset_N \cap \nontinfp{\nont'} = \set{N'\in\nset_N\mid \inhav{\grm',N'}=\infty}\). 
For any \(N \in \nont\),
since \(\inhav{\grm}=\infty\) and \(\grm\) is strongly connected, 
we have \(\inhav{\grm,N}=\infty\), and hence
\(\nset'_N\neq \emptyset\). 
By Lemma~\ref{lem:imt}, there exists a real constant \(p>0\) such that, for each 
\(N'\in \nset'_N\),
\begin{align}
 \limd_{n \rightarrow \infty} \frac{\card{\{ \tr \in \inhavn{\grm',N'} \mid
 S_{\ceil{\cn \log n}} \subc \tr
 \}}}{\card{\inhavn{\grm',N'}}} = 1.
\label{eq:lim}
\end{align}
Thus, we have
\begin{align*}
&\
 \limd\limits_{n \rightarrow \infty} \frac{\card{\{ \tr \in \inhavn{\grm,N} \mid
 S_{\ceil{\cn \log n}} \subc \tr
 \}}}{\card{\inhavn{\grm,N}}}
\\ =&\
 \limd\limits_{n \rightarrow \infty} \frac{\sum_{N'\in\nset_N}\card{\{ \tr \in \inhavn{\grm',N'} \mid
 S_{\ceil{\cn \log n}} \subc \tr
 \}}}{\sum_{N'\in\nset_N}\card{\inhavn{\grm',N'}}}
\\ =&\
 \limd\limits_{n \rightarrow \infty} \frac{\sum_{N'\in\nset'_N}\card{\{ \tr \in \inhavn{\grm',N'} \mid
 S_{\ceil{\cn \log n}} \subc \tr
 \}}}{\sum_{N'\in\nset'_N}\card{\inhavn{\grm',N'}}}
\\ =&\
1
\end{align*}
as required; we have used Lemma~\ref{lem:convSum} and
Equation~\eqref{eq:lim} in the last step.
\end{proof}

  \section{Proof of the Main Theorem on \(\lambda\)-calculus}\label{sec:proofbeta}

This section proves our main theorem (Theorem~\ref{mainthm}). We first prepare
a regular tree grammar that generates the set of tree representations of 
elements of $\terms{\D}{\I}{\K}$ in Section~\ref{sec:grammar-for-lambda}, and then
apply Corollary~\ref{cor:imtTree} to obtain Theorem~\ref{mainthm},
where \((T_n)_{n \in \nat}\) in the corollary are set to (the tree representations of)
 the terms in the introduction that have long \(\beta\)-reduction sequences.

\subsection{Regular Tree Grammar $\lgrm{\D}{\I}{\K}$ of $\terms{\D}{\I}{\K}$}
\label{sec:grammar-for-lambda}

Recall that $\terms{\D}{\I}{\K}$ is the set of (\(\alpha\)-equivalence classes of)
closed well-typed terms, whose order, internal arity,
and number of variables are bounded above by \(\D\), \(\I\), and \(\K\)
(consult Definition~\ref{df:set-of-terms} for the precise definition).
The set $\terms{\D}{\I}{\K}$ can be generated by the following grammar (up to isomorphism).
\begin{defi}[Grammar of $\terms{\D}{\I}{\K}$]
Let $\D,\I,\K \geq 0$ be integers and $X_\K = \{x_1, \ldots, x_\K\}$ be a subset of $V$.
 The regular tree grammar \(\lgrm{\D}{\I}{\K}\) is defined as
\((\lralph{\D}{\I}{\K}, \lnont{\D}{\I}{\K}, \lrules{\D}{\I}{\K})\) where:
\begin{align*}
\lralph{\D}{\I}{\K} \defeq& \ \{ x \mapsto 0 \mid x \in X_\K \} \cup \{@ \mapsto
 2\}\\
  \cup & \ \{ \lambda \var^\tau \mapsto 1 \mid
       \var \in \{\uv\} \cup X_\K,\  \tau \in \types{\D-1}{\I}
       \}\\
 \lnont{\D}{\I}{\K} \defeq & \ \{\NT{\Gamma}{\tau} \mid
 \tau \in \types{\D}{\I}, \dom(\Gamma) \subseteq X_\K,\  \image{\Gamma} \subseteq \types{\D-1}{\I},
 \\
 & \ \qquad\qquad\termsg{\Gamma}{\tau}{\D}{\I}{\K}\neq\emptyset\} \\
 \lrules{\D}{\I}{\K} \defeq & \ \{\NT{\{x_i:\tau\}}{\tau} \rew x_i \}
 \cup \{ \NT{\Gamma}{\sigma \ft \tau} \rew \lambda \uv^{\sigma}
 (\NT{\Gamma}{\tau}) \}\\
   \cup & \  \{ \NT{\Gamma}{\sigma \ft \tau} \rew \lambda x_i^{\sigma}
 (\NT{\Gamma\cup\set{x_i:\sigma}}{\tau}) \mid i = \min\{ j \mid x_j\in X_\K\setminus \dom(\Gamma)\}\}
\\ \cup & \
   \{\NT{\Gamma}{\tau} \rew @(\NT{\Gamma_1}{\sigma \ft \tau},
 \NT{\Gamma_2}{\sigma}) \mid \Gamma = \Gamma_1 \cup \Gamma_2\}
\end{align*}
\end{defi}
The grammar above generates the tree representations of elements of $\terms{\D}{\I}{\K}$,
where a variable \(x\), a lambda-abstraction, and an application
are represented respectively as the nullary tree constructor \(x\), unary tree constructor
\(\lambda \var\), and binary tree constructor \(@\).
The nonterminal \(\NT{\Gamma}{\tau}\) is used to generate (the tree representations of) the elements of 
\(\termsg{\Gamma}{\tau}{\D}{\I}{\K}\); the condition 
\(\termsg{\Gamma}{\tau}{\D}{\I}{\K}\neq\emptyset\) on nonterminal \(\NT{\Gamma}{\tau}\) ensures that
every nonterminal generates at least one tree.
To guarantee that the grammar generates at most one tree for each
\(\alpha\)-equivalence class \([t]_\alpha\), (i) variables are chosen from the fixed set \(X_\K\),
and (ii) in the rule for generating a \(\lambda\)-abstraction, a variable is chosen in a deterministic manner.
Note that \(\lralph{\D}{\I}{\K}\), \(\lnont{\D}{\I}{\K}\) and \(\lrules{\D}{\I}{\K}\) are finite.
The finiteness of \(\lnont{\D}{\I}{\K}\) follows from that of 
\(X_\K\), \(\types{\D-1}{\I}\), and 
\(\set{\Gamma \mid \dom(\Gamma) \subseteq X_\K,\  \image{\Gamma} \subseteq
\types{\D-1}{\I}}\). The finiteness of \(\lrules{\D}{\I}{\K}\) also follows
immediately from that of \(\lnont{\D}{\I}{\K}\).

\begin{exa}\label{exam:lambdaGrammar}
Let us consider the case where $\D = \I = \K = 1$.
The grammar $\lgrm{1}{1}{1}$ consists of the following components. 
\begin{align*}
 \lralph{1}{1}{1} &=  \{ x_1, @, \lambda x_1^\T, \lambda \uv^\T \}
 \qquad
 \lnont{1}{1}{1} = \{ \NT{\ee}{\T \ft \T}, \NT{\{x_1: \T\}}{\T},
 \NT{\{x_1: \T\}}{\T \ft \T} \}\\
 \lrules{1}{1}{1} &=  
 \begin{cases}
  \NT{\ee}{\T \ft \T} \rew \lambda x_1^\T(\NT{\{x_1:\T\}}{\T})\\
  \NT{\{x_1: \T \}}{\T} \rew 
x_1 \mid @(\NT{\{x_1: \T\}}{{\T \ft \T}},\NT{\{x_1: \T\}}{\T}) \mid
  @(\NT{\ee}{\T \ft \T},\NT{\{x_1: \T\}}{\T})
\\
  \NT{\{x_1: \T\}}{\T \ft \T} \rew \lambda \uv^\T(\NT{\{x_1: \T \}}{\T}).
  \end{cases}
\end{align*}
\end{exa}

There is an obvious embedding $\embp{\D}{\I}{\K}$ ($\emb$ for short) from trees in
$\TR{\lralph{\D}{\I}{\K}}$ into (not necessarily well-typed) $\lambda$-terms.
For $\NT{\Gamma}{\tau} \in \lnont{\D}{\I}{\K}$ we define
\[
 \tecp{\D}{\I}{\K}{\Gamma}{\tau} \defeq [-]_\alpha \circ \emb:
 \inhav{\NT{\Gamma}{\tau}} \fa \termsg{\Gamma}{\tau}{\D}{\I}{\K}
\]
where $[-]_\alpha$ maps a term to its $\alpha$-equivalence class.
We sometimes omit the superscript and/or the subscript and may write just \(\pi\) for 
\(\tecp{\D}{\I}{\K}{\Gamma}{\tau}\). 

The following lemma says that
$\lgrm{\D}{\I}{\K}$ gives a complete representation system of the \(\alpha\)-equivalence classes.

\begin{lem}\label{prop:cong}
For $\D,\I,\K \geq 0$, 
$\ \tecp{\D}{\I}{\K}{\Gamma}{\tau}:
 \inhav{\lgrm{\D}{\I}{\K},\NT{\Gamma}{\tau}} \fa \termsg{\Gamma}{\tau}{\D}{\I}{\K}$ is a size-preserving bijection.
\end{lem}
\begin{proof}
It is trivial that the image of \(\tec_{\typing{\Gamma}{\tau}}\) is contained
 in \(\termsg{\Gamma}{\tau}{\D}{\I}{\K}\) and \(\tec\) preserves the size.

The injectivity, 
\ie \(\emb(\tr) \aeq \emb(\tr')\) implies \(\tr=\tr'\) for \(\tr, \tr' \in \inhav{\NT{\Gamma}{\tau}}\),
is shown by induction on 
the length of the leftmost rewriting
sequence \(\NT{\Gamma}{\tau} \rew^* T\)
and by case analysis of the rewriting rule used in the first step of the reduction sequence. 
In the case analysis below, we use the fact that if \(T\in \inhav{\NT{\Gamma}{\tau}}\), then 
\(\fv(\emb(T))=\dom(\Gamma)\), 
which can be proved by straightforward induction.
\begin{itemize}
\item Case for the rule  $\NT{\{x_i:\tau\}}{\tau} \rew x_i$: In this case, \(T=x_i\). By the assumption
\(\emb(\tr)\aeq\emb(\tr')\), \(T=T'\) follows immediately.
\item Case for the rule \(\NT{\Gamma}{\sigma \ft \tau} \rew
 \lambda \uv^{\sigma} (\NT{\Gamma}{\tau}) \): In this case, \(\tr= \lambda \uv^{\sigma} (\tr_1)\) with
\(\NT{\Gamma}{\tau}\rew^* \tr_1\). By the assumption 
\(\emb(\tr)\aeq\emb(\tr')\), \(\tr'\) must be of the form \(\lambda \var^\sigma(\tr_1')\)
where (i) \(\var\) does not occur free in \(\emb(\tr_1')\), and (ii) \(\emb(\tr_1)\aeq\emb(\tr_1')\).
Then \(\var = \uv\) and hence we have 
\(\NT{\Gamma}{\sigma \ft \tau} \rew \lambda \uv^{\sigma} (\NT{\Gamma}{\tau}) \) with \(\NT{\Gamma}{\tau}\rew^* \tr_1'\),
because, if \(\var \neq \uv\), we have 
\(\NT{\Gamma}{\sigma \ft \tau} \rew \lambda \var^{\sigma} (\NT{\Gamma\cup\set{\var:\sigma}}{\tau}) \) with
      \(\NT{\Gamma\cup\set{\var:\sigma}}{\tau}\rew^* \tr_1'\), and hence
      \(\fv(\emb(T'_1))=\dom(\Gamma\cup\set{\var:\sigma})\), which contradicts the condition (i).
Therefore, by the induction hypothesis, we have \(\tr_1=\tr_1'\), which implies \(\tr=\tr'\) as required.
\item Case for the rule:
 \( \NT{\Gamma}{\sigma \ft \tau} \rew \lambda x_i^{\sigma}
 (\NT{\Gamma\cup\set{x_i:\sigma}}{\tau})\), where \(i = \min
\set{ j \mid x_j\in X_\K\setminus \dom(\Gamma)}\).
By the assumption
\(\emb(\tr)\aeq\emb(\tr')\), \(\tr'\) must be of the form \(\lambda x^{\sigma}(\tr_1')\) where \(x\) occurs free in \(\emb(\tr_1')\),
and \([x'/x_i]\emb(\tr_1)\aeq [x'/x]\emb(\tr_1')\) for a fresh variable \(x'\).
Thus, by the definition of \(\lrules{\D}{\I}{\K}\), 
\(\tr'\) must have also been generated by the same rule, i.e., \(\tr'=\lambda x_i^{\sigma}(\tr_1')\)
with \(\NT{\Gamma\cup\set{x_i:\sigma}}{\tau}\rew^* \tr_1'\).
By the induction hypothesis, we have \(\tr_1=\tr_1'\), which also implies \(\tr=\tr'\) as required.
\item 
Case for the rule: \(\NT{\Gamma}{\tau} \rew @(\NT{\Gamma_1}{\sigma \ft \tau}, \NT{\Gamma_2}{\sigma})\)
where \(\Gamma = \Gamma_1 \cup \Gamma_2\).
Then \(\tr = @(\tr_1,\tr_2)\) with \(\NT{\Gamma_1}{\sigma\ft\tau}\rew^* \tr_1\)
and \(\NT{\Gamma_2}{\sigma}\rew^*\tr_2\). By the assumption 
\(\emb(\tr)\aeq\emb(\tr')\), \(T'\) must also be of the form \(@(\tr_1',\tr_2')\), with
\(\emb(\tr_1)\aeq \emb(\tr_1')\) and \(\emb(\tr_2)\aeq \emb(\tr_2')\).
Therefore, \(\tr'\) must also have been generated by a rule for applications;
hence 
\(\NT{\Gamma'_1}{\sigma' \ft \tau}\rew^*\tr_1'\) and \(\NT{\Gamma'_2}{\sigma'}\rew^*\tr_2'\) for some
\(\Gamma_1', \Gamma_2', \sigma'\) such that \(\Gamma = \Gamma'_1 \cup \Gamma'_2\).
By the condition \(\emb(\tr_i)\aeq\emb(\tr_i')\), 
 \(\dom(\Gamma_i)=\fv(\emb(\tr_i))=\fv(\emb(\tr_i'))=\dom(\Gamma'_i)\) for each \(i\in\set{1,2}\). Thus,
since \(\Gamma_1 \cup \Gamma_2 = \Gamma'_1\cup\Gamma'_2\),
we have \(\Gamma_i=\Gamma_i'\), which also implies \(\sigma=\sigma'\) (since the type of a simply-typed \(\lambda\)-term
is uniquely determined by the term and the type environment).
Therefore, by the induction hypothesis, we have \(\tr_i=\tr_i'\) for each \(i\in\set{1,2}\), which implies
\(\tr=\tr'\) as required.
\end{itemize}

Next we show the surjectivity, \ie for any \([t]_\alpha \in \termsg{\Gamma}{\tau}{\D}{\I}{\K}\) there exists
\(\tr \in \inhav{\NT{\Gamma}{\tau}}\) such that \(\emb(\tr) \aeq t\).
For \(\D, \I, \K\) with \(X_\K = \{x_1, \ldots, x_\K\}\) and \(\NT{\Gamma}{\tau} \in \lnont{\D}{\I}{\K}\),
we define a ``renaming'' function 
\(\rnp{\D}{\I}{\K}{\Gamma}{\tau}\) (\(\rn_{\typing{\Gamma}{\tau}}\) or \(\rn\) for short) from \(\set{t\mid [t]_\alpha \in
 \termsg{\Gamma}{\tau}{\D}{\I}{\K}}\)
 to \(\inhav{\NT{\Gamma}{\tau}}\) by induction on the size of \(t\),
so that \(\tec{}(\rn(t)) = [t]_\alpha\) holds.

\[
\newcommand\localspace{\mspace{-50mu}}
\newcommand\localspacee{\mspace{-220mu}}
\begin{array}{ll}
\rn_{\typing{\set{x:\tau}}{\tau}}(x) &\defeq x
\\
\rn_{\typing{\Gamma}{\sigma\ft\tau}}(\lambda \uv^{\sigma}\!.t) &\defeq \lambda
 \uv^\sigma(\rn_{\typing{\Gamma}{\tau}}(t))
\\
\rn_{\typing{\Gamma}{\sigma\ft\tau}}(\lambda x^{\sigma}\!.t) &\defeq \lambda
 \uv^\sigma(\rn_{\typing{\Gamma}{\tau}}(t)) 
\qquad\qquad\hfill(\text{if
 }x\notin\fv(t)) 
\\
\rn_{\typing{\Gamma}{\sigma\ft\tau}}(\lambda x^{\sigma}\!.t) &\defeq
\lambda
 x_i^\sigma(\rn_{\typing{\Gamma\cup\set{x_i:\sigma}}{\tau}}(\swap{t}{x_i}{x}))
\text{ where \(i \defeq \min\{ j \mid x_j\in X_\K\setminus
 \dom(\Gamma)\}\)}
 \\
&\hfill (\text{if }x\in\fv(t)) 
\\
\rn_{\typing{\Gamma_1\cup\Gamma_2}{\tau}}(t_1t_2) &\defeq 
@(\rn_{\typing{\Gamma_1}{\sigma\ft\tau}}(t_1),
 \rn_{\typing{\Gamma_2}{\sigma}}(t_2))
\hfill(\text{if }
 \Gamma_1 \p t_1: \sigma \ft \tau \text{ and } \Gamma_2 \p t_2: \sigma)
\end{array}
\]
Here, \(\swap{t}{x}{y}\) represents the term obtained by swapping every occurrence of \(x\) with that of \(y\);
for example, \(\swap{(\lambda x.xy)}{x}{y} = \lambda y.yx\). Note that in the last clause,
if \(\Gamma\p t_1t_2:\tau\), then there exists a unique triple \((\Gamma_1,\Gamma_2,\sigma)\) such that
\(\Gamma=\Gamma_1\cup\Gamma_2\), 
\( \Gamma_1 \p t_1: \sigma \ft \tau\), and \(\Gamma_2 \p t_2: \sigma\). 
We can prove that \(\tec{}(\rn(t)) = [t]_\alpha\) holds for every \(\lambda\)-term \(t\) such that
\([t]_\alpha\in \termsg{\Gamma}{\tau}{\D}{\I}{\K}\),
 by straightforward induction on the size of $t$.
 \end{proof}

\begin{lem}\label{lem:unamb}
For $\D,\I,\K \geq 0$, \(\lgrm{\D}{\I}{\K}\) is unambiguous.
\end{lem}
\begin{proof}
The proof is similar to that of the injectivity of Lemma~\ref{prop:cong}.
We show that, for any \(\NT{\Gamma}{\tau} \in \lnont{\D}{\I}{\K}\), \(T \in \TR{\lralph{\D}{\I}{\K}}\),
and two leftmost rewriting sequences of the form
\begin{align*}&
\NT{\Gamma}{\tau} \rew T^{(1)} \rew \dots \rew T^{(n)} = T \qquad(n \ge 1)
\\&
\NT{\Gamma}{\tau} \rew T'^{(1)} \rew \dots \rew T'^{(n')} = T \qquad(n' \ge 1),
\end{align*}
\(n=n'\) and \(T^{(k)}=T'^{(k)}\) hold for \(1 \le k < n\).
The proof proceeds by induction on \(n\), with
 case analysis on the rule \(\NT{\Gamma}{\tau} \rew T^{(1)}\).
As in the proof of Lemma~\ref{prop:cong},
we use the fact that if \(T\in \inhav{\NT{\Gamma}{\tau}}\), then 
\(\fv(\emb(T))=\dom(\Gamma)\).
\begin{itemize}
\item Case for the rule  $\NT{\{x_i:\tau\}}{\tau} \rew x_i$: In this case, \(T=x_i\) and \(n=1\).
Since the root of \(T'^{(1)}\) must be \(x_i\), we have \(T'^{(1)} = T\) and \(n'=1=n\).

\item Case for the rule \(\NT{\Gamma}{\sigma \ft \tau'} \rew \lambda \uv^{\sigma} (\NT{\Gamma}{\tau'}) \):
In this case, \(n \ge 2\),
\(T^{(k)} = \lambda \uv^{\sigma}(T^{(k)}_1)\) (\(k \le n\)),
\(T= \lambda \uv^{\sigma} (T_1)\), and we have the following leftmost rewriting sequences:
\[
\NT{\Gamma}{\tau'} = T^{(1)}_1 \rew \dots \rew T^{(n)}_1 = T_1.
\]
Since the root of \(T'^{(k)}\) is the same as that of \(T\),
\(T'^{(k)} = \lambda \uv^{\sigma}(T'^{(k)}_1)\) for \(k \le n'\),
and  we have the following leftmost rewriting sequence:
\[
T'^{(1)}_1 \rew \dots \rew T'^{(n')}_1 = T_1.
\]
By the definition of rules, \(T'^{(1)}\) must be of the form \(\lambda \uv^{\sigma} (\NT{\Gamma}{\tau})\)
and so \(T'^{(1)}_1 = \NT{\Gamma}{\tau}\).
By the induction hypothesis for \(\NT{\Gamma}{\tau} \rew T^{(2)}_1\),
we have \(n=n'\) and \(T^{(k)}_1=T'^{(k)}_1\) for \(2 \le k < n\).
Thus we have also \(T^{(k)}=T'^{(k)}\) for \(1 \le k < n\), as required.

\item Case for the rule:
 \( \NT{\Gamma}{\sigma \ft \tau'} \rew \lambda x_i^{\sigma}
 (\NT{\Gamma\cup\set{x_i:\sigma}}{\tau'})\), where \(i = \min
\set{ j \mid x_j\in X_\K\setminus \dom(\Gamma)}\).
In this case, \(n \ge 2\), 
\(T^{(k)} = \lambda x_i^{\sigma}(T^{(k)}_1)\) (\(k \le n\)),
\(T= \lambda x_i^{\sigma} (T_1)\), and we have the following leftmost rewriting sequences:
\[
\NT{\Gamma\cup\set{x_i:\sigma}}{\tau'} = T^{(1)}_1 \rew \dots \rew T^{(n)}_1 = T_1.
\]
Since the root of \(T'^{(k)}\) are the same as that of \(T\),
\(T'^{(k)} = \lambda x_i^{\sigma}(T'^{(k)}_1)\) (\(k \le n'\)),
and we have the leftmost rewriting sequence:
\[
T'^{(1)}_1 \rew \dots \rew T'^{(n')}_1 = T_1.
\]
By the definition of rules, 
\(\NT{\Gamma}{\sigma \ft \tau'} \rew \lambda x_i^{\sigma} (T'^{(1)}_1)\)
implies that \(T'^{(1)}_1 = \NT{\Gamma\cup\set{x_i:\sigma}}{\tau'}\).
Hence by the induction hypothesis for \(\NT{\Gamma\cup\set{x_i:\sigma}}{\tau} \rew T^{(2)}_1\),
we have \(n=n'\) and \(T^{(k)}_1=T'^{(k)}_1\) for \(2 \le k < n\).
Thus we also have \(T^{(k)}=T'^{(k)}\) for \(1 \le k < n\).

\item 
Case for the rule: \(\NT{\Gamma}{\tau} \rew @(\NT{\Gamma_1}{\sigma \ft \tau}, \NT{\Gamma_2}{\sigma})\)
where \(\Gamma = \Gamma_1 \cup \Gamma_2\).
In this case, \(n=n_1+n_2-1\), \(n_1,n_2 \ge 2\), \(T= @(T_1,T_2)\),
\begin{align*}
T^{(k)} &= @(T^{(k)}_1, \NT{\Gamma_2}{\sigma}) && (1 \le k \le n_1)
\\
T^{(n_1-1+k)} &= @(T_1, T^{(k)}_2) && (1 \le k \le n_2),
\end{align*}
 and we have the following leftmost rewriting sequences:
\begin{align*}
\NT{\Gamma_1}{\sigma \ft \tau} &= T^{(1)}_1 \rew \dots \rew T^{(n_1)}_1 = T_1
\\
\NT{\Gamma_2}{\sigma} &= T^{(1)}_2 \rew \dots \rew T^{(n_2)}_2 = T_2.
\end{align*}
Since the root of \(T'^{(1)}\) is \(@\), 
\(T'^{(1)}\) must be of the form \(@(\NT{\Gamma'_1}{\sigma' \ft \tau}, \NT{\Gamma'_2}{\sigma'})\)
with \(\Gamma'_1\cup\Gamma'_2 = \Gamma\).
Also let us consider \(T'^{(k)}\):
we have \(n=n'_1+n'_2-1\), \(n'_1,n'_2 \ge 2\),
\begin{align*}
T'^{(k)} &= @(T'^{(k)}_1, \NT{\Gamma'_2}{\sigma'}) && (1 \le k \le n'_1)
\\
T'^{(n'_1-1+k)} &= @(T_1, T'^{(k)}_2) && (1 \le k \le n'_2),
\end{align*}
 and we have the following leftmost rewriting sequences:
\begin{align*}
\NT{\Gamma'_1}{\sigma' \ft \tau} &= T'^{(1)}_1 \rew \dots \rew T'^{(n'_1)}_1 = T_1
\\
\NT{\Gamma'_2}{\sigma'} &= T'^{(1)}_2 \rew \dots \rew T'^{(n'_2)}_2 = T_2.
\end{align*}

Since \(\dom(\Gamma_i)=\fv(\emb(T_i))=\dom(\Gamma'_i)\) for each \(i\in\set{1,2}\)
and \(\Gamma_1 \cup \Gamma_2 = \Gamma'_1\cup\Gamma'_2\),
we have \(\Gamma_i=\Gamma'_i\),
which also implies \(\sigma=\sigma'\).
Hence by the induction hypothesis for \(\NT{\Gamma_1}{\sigma \ft \tau} \rew T^{(2)}_1\)
and for \(\NT{\Gamma_2}{\sigma} \rew T^{(2)}_2\),
we have \(n_i=n'_i\) and \(T^{(k)}_i=T'^{(k)}_i\) for each \(i\in\set{1,2}\) and \(2 \le k < n_i\).
Thus we also have \(T^{(k)}=T'^{(k)}\) for \(1 \le k < n\).
\qedhere
\end{itemize}
\end{proof}

\subsection{Strong Connectivity and Aperiodicity} 

In this section, we restrict the grammar \(\lgrm{\D}{\I}{\K}\) to
$\egrm{\D}{\I}{\K}$ by removing unnecessary nonterminals, and 
 show the strong connectivity and aperiodicity of $\egrm{\D}{\I}{\K}$
 for $\D,\I, \K \geq 2$ (Lemma~\ref{prop:irr} below).
Recall that the strong connectivity and aperiodicity is required to apply Corollary~\ref{cor:imtTree} and
Remark~\ref{rem:avoidLimDef}, respectively.

We define the restricted grammar \(\egrm{\D}{\I}{\K}\) by:
\[
\begin{aligned}
\enont{\D}{\I}{\K} &\defeq \set{\N{\z} \in \lnont{\D}{\I}{\K} \mid 
\text{\(\N{\z}\) is reachable in \(\lgrm{\D}{\I}{\K}\)
from some \(\NT{\emptyset}{\sigma} \in \lnont{\D}{\I}{\K}\)}}
\\
\erules{\D}{\I}{\K} &\defeq \set{\N{\z} \rew T \in \lrules{\D}{\I}{\K} \mid 
\N{\z} \in \enont{\D}{\I}{\K}}
\\
\egrm{\D}{\I}{\K} &\defeq
(\lralph{\D}{\I}{\K}, \enont{\D}{\I}{\K}, \erules{\D}{\I}{\K}).
\end{aligned}
\]
For \(\N{\z} \in \enont{\D}{\I}{\K}\), clearly
\(\inhav{\egrm{\D}{\I}{\K},\N{\z}} = \inhav{\lgrm{\D}{\I}{\K},\N{\z}}\).
Through the bijection \(\tec\),
we can show that, for any \(\NT{\Gamma}{\tau} \in \lnont{\D}{\I}{\K}\), \(\NT{\Gamma}{\tau}\) also belongs to \(\enont{\D}{\I}{\K}\)
if and only if there exists a term in \(\terms{\D}{\I}{\K}\) whose type
derivation contains a type judgment of the form \(\Gamma \p t:
\tau\).

The strong connectivity of \(\egrm{\D}{\I}{\K}\)
follows from the following facts: (i) each \(\NT{\Gamma}{\tau}\in \enont{\D}{\I}{\K}\) is reachable from
some \(\NT{\emptyset}{\tau'}\in \enont{\D}{\I}{\K}\) (by the definition of $\egrm{\D}{\I}{\K}$ above),
(ii) each $\NT{\emptyset}{\tau}\in \enont{\D}{\I}{\K}$ is 
reachable from $\NT{\emptyset}{\T\ft\T}$ (Lemma~\ref{lem:sc1} below),
and (iii) 
$\NT{\emptyset}{\T\ft\T}$ is reachable from every \(\NT{\Gamma}{\tau}\in \enont{\D}{\I}{\K}\) 
(Lemma~\ref{lem:sc2} below).

\begin{lem}
\label{lem:sc1}
Let $\D, \I \geq 2$ and $\K \geq 1$ be integers.
Then for any nonterminal $\NT{\emptyset}{\tau} \in \enont{\D}{\I}{\K}$,
$\NT{\emptyset}{\tau}$ is reachable from $\NT{\emptyset}{\T\ft\T}$.
\end{lem}
\begin{proof}
Let \(\tau = \tau_1 \ft \dots \ft \tau_n \ft \T\) and
\(\tau_i = \tau^i_1 \ft \dots \ft \tau^i_{\ell_i} \ft \T\)
for \(i=1,\dots,n\).
For \(i=1,\dots,n\), let \(T_{\tau_i} \defeq \lambda *^{\tau^i_1}(\dots\lambda *^{\tau^i_{\ell_i}}(x_1)\dots)\).
Then, we have:
\[
\begin{aligned}
\NT{x_1:\T}{\tau_i}
\rew
\lambda *^{\tau^i_1}(\NT{x_1:\T}{\tau^i_2 \ft \dots \ft \tau^i_{\ell_i} \ft \T})
&\rew^*
\lambda *^{\tau^i_1}(\dots\lambda *^{\tau^i_{\ell_i}}(\NT{x_1:\T}{\T})\dots)
\\
&\rew
\lambda *^{\tau^i_1}(\dots\lambda *^{\tau^i_{\ell_i}}(x_1)\dots)
= T_{\tau_i}
\end{aligned}
\]
and hence
\[
\begin{aligned}[b]
\NT{\emptyset}{\T\ft\T}
&\rew\phantom{{}^*}
\lambda x_1^{\T}(\NT{x_1:\T}{\T})
\\
&\rew\phantom{{}^*}
\lambda x_1^{\T}(
 @(\NT{x_1:\T}{\tau_n\ft\T},\NT{x_1:\T}{\tau_n})
)
\rew^*
\lambda x_1^{\T}(
 @(\NT{x_1:\T}{\tau_n\ft\T},T_{\tau_n})
)
\\
&\rew^*
\lambda x_1^{\T}(
 @({@(
 \dots
 @(
\NT{x_1:\T}{\tau_2\ft\cdots\ft\tau_n\ft\T}
,\tr_{\tau_2}) ,
 \dots)},\tr_{\tau_n})
)
\\
&\rew\phantom{{}^*}
\lambda x_1^{\T}(
 @({@(
 \dots
 @(
@(\NT{\emptyset}{\tau_1\ft\tau_2\ft\cdots\ft\tau_n\ft\T},\NT{x_1:\T}{\tau_1})
,\tr_{\tau_2}) ,
 \dots)},\tr_{\tau_n})
)
\\
&\rew^*
\lambda x_1^{\T}(
 @({@(
 \dots
 @(
@(\NT{\emptyset}{\tau},\tr_{\tau_1})
,\tr_{\tau_2}) ,
 \dots)},\tr_{\tau_n})
).
\end{aligned}
\tag*{\qedhere}
\]
\end{proof}

\begin{lem}
\label{lem:sc2}
Let $\D, \I, \K \geq 2$ be integers.
Then for any nonterminal $\NT{\Gamma}{\tau} \in \enont{\D}{\I}{\K}$,
$\NT{\emptyset}{\T\ft\T}$ is reachable from \(\NT{\Gamma}{\tau}\).
\end{lem}
\begin{proof}
Suppose $\NT{\Gamma}{\tau} \in \enont{\D}{\I}{\K}$.
By the definition of \(\enont{\D}{\I}{\K}\) (and \(\lnont{\D}{\I}{\K}\)), there exists \(t\) such that
\([t]_\alpha\in \termsg{\Gamma}{\tau}{\D}{\I}{\K}\).
Let \(\tr \defeq \tec^{-1}([t]_\alpha) \in \inhav{\NT{\Gamma}{\tau}}\).
Now \(\tr\) contains at least one (possibly bound) variable, say \(x\), and let \(\sigma\) be the type of \(x\).
Since \(\tr \in \inhav{\NT{\Gamma}{\tau}}\), there exists a \linear{context} \(\ctr\) such that
\[
\NT{\Gamma}{\tau} \rew^{*} \ctr[\NT{\set{x:\sigma}}{\sigma}] \rew \ctr[x] = \tr.
\]

We show, by the induction on the structure of \(\sigma\), that \(\NT{\set{x:\sigma}}{\sigma}\rew^* 
\ctr'[\NT{\set{y:\T}}{y}]\) for some \(y\) and \linear{context} \(\ctr'\).
The case for \(\sigma=\T\) is obvious. If \(\sigma=\sigma'\to\tau'\), then since \(\K\geq 2\),
we obtain:
\[
\begin{array}{l}
\NT{\set{x:\sigma}}{\sigma}\rew \lambda x_i^{\sigma'}\NT{\set{x:\sigma,x_i:\sigma'}}{\tau'}
\rew \lambda x_i^{\sigma'}(@(\NT{\set{x:\sigma}}{\sigma}, \NT{\set{x_i:\sigma'}}{\sigma'}))\\
\rew \lambda x_i^{\sigma'}(@(x, \NT{\set{x_i:\sigma'}}{\sigma'}))
\end{array}
\]
by ``\(\eta\)-expansion''. By the induction hypothesis, we have \(
\NT{\set{x_i:\sigma'}}{\sigma'} \rew^* \ctr''[\NT{\set{y:\T}}{\T}]\) for some \(y\) and 1-context \(\ctr''\).
Thus, the result holds for \(\ctr'=\lambda x_i^{\sigma'}(@(x,\ctr''))\).

By using the property above, we obtain
\[
\begin{array}{l}
\NT{\Gamma}{\tau} \rew^{*} \ctr[\NT{\set{x:\sigma}}{\sigma}] \rew^*
\ctr[\ctr'[\NT{\set{y:\T}}{\T}]] \\\rew
\ctr[\ctr'[@(\NT{\emptyset}{\T\to\T},\NT{\set{y:\T}}{\T})]]
\rew 
\ctr[\ctr'[@(\NT{\emptyset}{\T\to\T},y)]].
\end{array}
\]
Thus, \(\NT{\emptyset}{\T\ft\T}\) is reachable from \(\NT{\Gamma}{\tau}\).
\end{proof}

\begin{lem}\label{lem:aperiodicity}
Let $\D, \I \geq 2$ and $\K \geq 1$ be integers.
 Then for any integer $n \geq 5$, the nonterminal $\NT{\ee}{\T \ft \T}$
 of $\egrm{\D}{\I}\K$
 satisfies \(\inhavn{\NT{\ee}{\T \ft \T}} \neq \emptyset\).
\end{lem}
\begin{proof}
For simplicity of the presentation, here we identify trees with terms by
the size-preserving bijection \(\tec\).
The proof proceeds by induction on \(n\).
For \(n=5,6,7\), the following terms
\[
\begin{aligned}&
\lambda x^{\T}\!. (\lambda \uv^{\T}\!.x) x
\\&
\lambda x^{\T}\!. (\lambda \uv^{\T\ft\T}\!.x) (\lambda \uv^{\T}\!.x)
\\&
\lambda x^{\T}\!. ( \lambda \uv^{\T\ft\T\ft\T}\!.x) \lambda \uv^{\T}\!.\lambda \uv^{\T}\!.x
\end{aligned}
\]
belong to \(\inhavn{\NT{\ee}{\T \ft \T}}\), respectively.
For $n \geq 8$, by the induction hypothesis, there exists
 $t \in \inhavnn{n-3}{\NT{\ee}{\T \ft \T}}$. Thus we have
 $\lambda x. tx \in \inhavn{\NT{\ee}{\T \ft \T}}$ as required.
\end{proof}

The following is the main result of this subsection.
\begin{lem}\label{prop:irr}
$\egrm{\D}{\I}{\K}$ is strongly connected and aperiodic for any $\D,\I, \K \geq 2$.
\end{lem}

\begin{proof}
Strong connectivity follows from Lemmas~\ref{lem:sc1} and \ref{lem:sc2},
and the definition of \(\egrm{\D}{\I}{\K}\) (as stated at the beginning of this subsection).
The aperiodicity follows from Lemma~\ref{lem:aperiodicity} and strong connectivity.
(Note that for every nonterminal \(N\), there exists a linear context \(U\in \inhav{\grm,
\NT{\emptyset}{\T\ft\T}\Ar N}\). Thus, for any \(n\ge |U|+5\), 
\(\inhavn{\grm,N} \supseteq \set{U[T]\mid T\in\inhavnn{n-|U|}{\grm,
\NT{\emptyset}{\T\ft\T}}}\neq \emptyset\).)
\end{proof}
 
\subsection{Explosive Terms}
\label{sec:explosive}
In this section, we define a family \((\expl{n}{k})_n\) 
of \(\lambda\)-terms, which have long \(\beta\)-reduction sequences.
They play the role of \((T_n)_n\) in 
Corollary~\ref{cor:imtTree}.

We define a ``duplicating term'' $\dt \defeq \lambda x^\T. (\lambda x^\T. \lambda
\uv^\T. x)x\,x$, and \(\idt \defeq \lambda x^\T. x\).
For two terms $t$, $t'$ and integer $n \geq 1$, we define the ``$n$-fold
application'' operation $\fold{}{n}{}$ by 
$\fold{t}{0}{t'} \defeq t'$ and $\fold{t}{n}{t'} \defeq t (\fold{t}{n-1}{t'})$.
For an integer $k \geq 2$, we define an order-\(k\) term
\[
 \ctwok{k} \defeq \lambda f^{\Ttwice{k-1}}. \lambda x^{\Ttwice{k-2}}.
 f (f x)
\]
where $\Ttwice{i}$ is defined by $\Ttwice{0} \defeq \T$ and $\Ttwice{i+1}
\defeq \Ttwice{i} \ft \Ttwice{i}$.

\begin{defi}[Explosive Term]
Let $m \geq 1$ and $k \geq 2$ be integers.
We define the \emph{explosive term $\expl{m}{k}$} 
by:
\[
\expl{m}{k} \defeq
\lambda x^\T.\bigl( (\fold{\ctwok{k}}{m}{\ctwok{k-1}})\ctwok{k-2} \cdots
 \ctwok{2} \, \dt (\idt\, x) \bigr).
\]
\end{defi}

We state key properties of $\expl{m}{k}$ below.
\begin{lem}[Explosive]
\label{lemma:expl}
\begin{enumerate}
\item\label{item:expltyped} $\ee\p\expl{m}{k} : \T\ft\T$ is derivable.
\item\label{item:expldix}
$\ord(\expl{m}{k}) = k$,
$\ar(\expl{m}{k}) = k$ and
$\card{\smash{\V(\expl{m}{k})}}=2$.
\item\label{item:explsize} $|\expl{m}{k}| = 8m + 8k - 2$.
\item\label{item:expluctx} \([\expl{m}{k}]_\alpha \in \termsg{\ee}{\T\ft\T}{\D}{\I}{\K}\) if \(\D\), \(\I \ge k\) and
     \(\K \ge 2\).
\item\label{item:explbeta}
      If a term $t$ satisfies $\expl{m}{k} \subc t$, then $\beta(t) \geq \Exp{k}{m}$ holds.
 \end{enumerate}
\end{lem}

\begin{proof}
First, observe that by the definition of $\ctwok{k}$, we have:
\[
 |\ctwok{k}| = 7, \quad \p \ctwok{k}: \Ttwice{k}, \quad \ord(\ctwok{k}) =
  k, \quad \ar(\ctwok{k}) = k. 
\]

(\ref{item:expltyped}) is obvious, and in the derivation of 
$\ee\p\expl{m}{k} : \T\ft\T$, 
\(\ctwok{k}\) is 
a subterm that has a type of the largest order and internal arity. Thus,
(\ref{item:expldix}) follows also immediately from 
\(\ord(\ctwok{k}) =\ar(\ctwok{k}) = k\), and 
\(\smash{\V(\expl{m}{k})}=\set{f,x}\).

(\ref{item:explsize}) follows from the following calculation.
\[
\begin{array}{ll}
|\expl{m}{k}| 
&= 1+ 
 |\fold{\ctwok{k}}{m}{\ctwok{k-1}}| + \Sigma_{i=2}^{k-2}(1+|\ctwok{i}|)+(1+|\dt|)+(1+|\idt\,x|)\\
& \hfill\mbox{(by $ |t_0t_1\cdots t_n| = |t_0|+\Sigma_{i=1}^n (1+|t_i|)$)}\\
&= 1+ 
 ((|\ctwok{k}|+1)m+|\ctwok{k-1}|) + \Sigma_{i=2}^{k-2}(1+|\ctwok{i}|)+(1+|\dt|)+(1+|\idt\,x|)\\
&= 1+ (8m+7) + 8(k-3)+9+5 
=8m + 8k -2.
\end{array}
\]

(\ref{item:expluctx}) follows from (\ref{item:expltyped}) and (\ref{item:expldix}).

The proof for (\ref{item:explbeta}) is the same as the corresponding proof in~\cite{Beckmann01}.
\end{proof}
 
\subsection{Proof of the Main Theorem}
We are now ready to prove Theorem~\ref{mainthm}.

\begin{proof}[Proof of Theorem~\ref{mainthm}]
We apply Corollary~\ref{cor:imtTree} with Remark~\ref{rem:avoidLimDef} to
the grammar \(\egrm{\D}{\I}{\K}\),
which is unambiguous by Lemma~\ref{lem:unamb},
and aperiodic and strongly connected by Lemma~\ref{prop:irr}.
We define \(T_n \defeq \tec^{-1}([\expl{n}{k}]_\alpha)\);
then by  Lemma~\ref{lemma:expl}(\ref{item:explsize}),
we have \(|T_n| = 8n + 8k - 2 = \order(n)\).
Thus, by Corollary~\ref{cor:imtTree} (with Remark~\ref{rem:avoidLimDef},
which allows us to replace \(\limd_{n \rightarrow \infty}\) with
\(\lim_{n \rightarrow \infty}\)), 
there exists \(\cn>0\) such that,
for \(\nset \defeq \set{ \NT{\ee}{\tau} \mid \tau \in \types{\D}{\I} }\),
\begin{equation}
\label{eq:mainCor}
 \lim_{n \rightarrow \infty} \frac{\card{\{ (N,\tr) \in \coprod_{N \in \nset}\inhavn{\egrm{\D}{\I}{\K},N} \mid
 \tr_{\ceil{\cn \log n}} \subc \tr
 \}}}{\card{\coprod_{N \in \nset}\inhavn{\egrm{\D}{\I}{\K},N}}} = 1.
\end{equation}
 
For \(n \in \nat\), we have
\begin{equation}
\label{eq:MAINsubset}
\{ [t]_\alpha \in \termsn{\D}{\I}{\K} \mid \expl{\ceil{\cn \log n}}{k} \subc t \}
\subseteq
\{ [t]_\alpha \in \termsn{\D}{\I}{\K} \mid \beta(t) \geq \Exp{k-1}{n^{\cn}} \}
\end{equation}
by Lemma~\ref{lemma:expl}(\ref{item:explbeta}), and
\begin{equation}
\label{eq:MAINtermGram}
\termsn{\D}{\I}{\K}
\ =\ 
\, \biguplus_{\mathclap{\tau \in
 \types{\D}{\I}}}\, \termsng{\ee}{\tau}{\D}{\I}{\K}
\ \cong\ 
\, \coprod_{\mathclap{N \in \nset}}\, 
\inhavnn{n}{\egrm{\D}{\I}{\K},N}
\end{equation}
by Lemma~\ref{prop:cong}.

Therefore, we have:
\begin{align*}
1 \ge \ &
\frac{\card{\{ [t]_\alpha \in
 \termsn{\D}{\I}{\K}
 \mid \beta(t) \geq \Exp{k-1}{n^{\cn}}
 \}}}{\card{\termsn{\D}{\I}{\K}}}
\\ \ge \ &
 \frac{\card{\{ [t]_\alpha \in
 \termsn{\D}{\I}{\K}
 \mid \expl{\ceil{\cn \log n}}{k} \subc t
 \}}}{\card{\termsn{\D}{\I}{\K}}}
&& \reason{by~\eqref{eq:MAINsubset}}
\\ =\ &
\frac{\card{\{ (N, \tr) \in 
\coprod_{N \in \nset}\, \inhavnn{n}{\egrm{\D}{\I}{\K},N}
 \mid
 \tr_{\ceil{\cn \log n}} \subc \tr
 \}}}{\card{\coprod_{N \in \nset}\, \inhavnn{n}{\egrm{\D}{\I}{\K},N}}}
&& \reason{by~\eqref{eq:MAINtermGram}}
\end{align*}
Since the right hand side converges to \(1\) by~\eqref{eq:mainCor}, we have
\[  \lim_{n \rightarrow \infty} 
\frac{\card{\{ [t]_\alpha \in
 \termsn{\D}{\I}{\K}
 \mid \beta(t) \geq \Exp{k-1}{n^{\cn}}
 \}}}{\card{\termsn{\D}{\I}{\K}}}
=1\]
as required.
\end{proof}

 \section{Related Work}\label{sec:relatedwork}
As mentioned in Section~\ref{sec:introduction},
there are several pieces of work on probabilistic properties of
untyped $\lambda$-terms~\cite{SN,grygiel2013counting,Bendkowski}.
David et al.~\cite{SN} have shown that almost all untyped
$\lambda$-terms are strongly normalizing, whereas the result is opposite
for terms expressed in SK combinators (the latter result has later been generalized for arbitrary Turing complete combinator bases~\cite{Bendkowski17}).
Their former result implies that untyped $\lambda$-terms
do \emph{not} satisfy the infinite monkey theorem, \ie for any term
$t$, the probability that a randomly chosen term of size $n$ contains
$t$ as a subterm tends to \emph{zero}.

Bendkowski et al.~\cite{Bendkowski} proved that almost all terms in de
Brujin representation are not strongly normalizing, by regarding the
size of an index $i$ is $i+1$, instead of the constant $1$.
The discrepancies among those results suggest that this kind of probabilistic property
is quite fragile and depends on the definition of the syntax and the size of terms.
Thus, the setting of our paper, especially the assumption on
the boundedness of internal arities and the number of variables is a matter of debate,
and it would be interesting to study how the result changes for different assumptions.

We are not aware of similar studies on \emph{typed} $\lambda$-terms. 
In fact, in their paper
about combinatorial aspects of $\lambda$-terms, 
Grygiel and Lescanne~\cite{grygiel2013counting} 
pointed out that
the combinatorial study of typed $\lambda$-terms is difficult, due to 
the lack of (simple) recursive definition of typed terms.
In the present paper, we have avoided the difficulty
by making the assumption on the boundedness of internal arities and the number
of variables (which is,  as mentioned above, subject to a debate though).

Choppy et al.~\cite{Choppy89} proposed a method to evaluate
the average number of reduction steps for
a restricted class of term rewriting systems called \emph{regular} rewriting systems.
In our context, the average number of reduction steps is not of much interest;
note that, as the worst-case number of reduction steps
is \(k\)-fold exponential for order-\(k\) terms, 
the average is also \(k\)-fold exponential, even if
it were the case that the number of reduction steps is small for almost all terms.

In a larger context, our work may be viewed as an instance of the studies of
average-case complexity~\cite[Chapter 10]{CCbook}, which discusses ``typical-case
feasibility''. We are not aware of much work on the average-case complexity of
problems with hyper-exponential complexity.

As a result related to our parameterized infinite monkey theorem for trees
(Theorem~\ref{thm:imt}), 
Steyaert and Flajolet~\cite{STEYAERT198319} studied the probability that a given pattern
(which is, in our terminology, a tree context of which every leaf is a hole)
occurs at a randomly chosen node of a randomly chosen tree. Their result immediately yields
a \emph{non-parameterized} infinite monkey theorem for trees (which says that 
the probability that a given pattern occurs in a sufficiently large tree tends to \(1\)),
but their technique does not seem directly applicable to obtain our \emph{parameterized} version.

 \section{Conclusion}\label{sec:conclusion}

We have shown that almost every simply-typed $\lambda$-term of order $k$
has a $\beta$-reduction sequence as long as $(k-1)$-fold exponential in the term size,
under a certain assumption. To our knowledge, this is the first result of this kind for
{typed} $\lambda$-terms. We obtained the result through the parameterized infinite monkey theorem 
for regular tree grammars, which may be of independent interest.

A lot of questions are left for future work, such as (i) whether
our assumption (on the boundedness of arities and the number of
variables) is reasonable, and how the result changes for different assumptions, (ii)
 whether our result is optimal (e.g., whether almost every term has a $k$-fold exponentially long reduction sequence), and (iii) whether similar results hold for Terui's decision
 problems~\cite{Terui12RTA} and/or the higher-order model checking problem~\cite{Ong06LICS}.
To resolve the question (ii) above, it may be useful to conduct
experiments to count the number of reduction steps for randomly generated terms.

 \subsubsection*{Acknowledgment}
We would like to thank anonymous referees for useful comments.
This work was supported by JSPS KAKENHI Grant Number JP15H05706.

\bibliographystyle{../splncs.bst}

\appendix
\section{A List of Notations}
\label{sec:symbols}

\begin{tabular}{|l|l|l|}
  \hline
  notation & explanation & introduced at:\\
  \hline
  $\card{X}$ & the cardinality of set \(X\), or the length of sequence \(X\)
  & Section~\ref{sec:mainresult}\\
  \hline
  $s.i$ & the \(i\)-th element of sequence \(s\)
  & Section~\ref{sec:mainresult}\\
  \hline
  $\dom(f)$ & the domain of function \(f\)
  & Section~\ref{sec:mainresult}\\
  \hline
  $\image{f}$ & the image of function \(f\)
  & Section~\ref{sec:mainresult}\\
  \hline
  $\tau$ & a type (of a simply-typed \(\lambda\)-term);
  & Section~\ref{sec:mainresult}\\
  \hline
  $\ord(\tau)$ & the order of type \(\tau\)
  & Section~\ref{sec:mainresult}\\
  \hline
  $\ar(\tau)$ & the internal arity of type \(\tau\)
  & Section~\ref{sec:mainresult}\\
\hline
$\types{\D}{\I}$ & \mline{the set of types whose order and internal arity are bounded by
\(\D\) and \(\I\) respectively}
  & Section~\ref{sec:mainresult}\\
  \hline
  \(t\) & a \(\lambda\)-term
  & Section~\ref{sec:mainresult}\\
  \hline
  \(\beta(t)\) & the maximal length of \(\beta\)-reduction sequences of \(t\)
  & Section~\ref{sec:mainresult}\\
  \hline
  \(\D\) & a bound on the order of \(\lambda\)-terms
  & Section~\ref{sec:mainresult}\\
  \hline
  \(\I\) & a bound on the arity of \(\lambda\)-terms
  & Section~\ref{sec:mainresult}\\
  \hline
  \(\K\) & a bound on the number of variables used in \(\lambda\)-terms
  & Section~\ref{sec:mainresult}\\
  \hline
  \(\termsg{\Gamma}{\tau}{\D}{\I}{\K}\) &
  \mline{
  the set of \(\alpha\)-equivalence classes of terms \(t\) such that
  \(\Gamma\p t:\tau\), with the given bounds}
  & Section~\ref{sec:mainresult}\\
  \hline
  \(\termsng{\Gamma}{\tau}{\D}{\I}{\K}\) &
  \mline{Like \(\termsg{\Gamma}{\tau}{\D}{\I}{\K}\), but
    the term size is restricted to \(n\)}
  & Section~\ref{sec:mainresult}\\
  \hline
  \(\terms{\D}{\I}{\K}\) &
  \mline{
    the set of \(\alpha\)-equivalence classes of closed well-typed
    terms, \new{with the given bounds}}
  & Section~\ref{sec:mainresult}\\
  \hline
  \(\termsn{\D}{\I}{\K}\) &
  \mline{Like \(\terms{\D}{\I}{\K}\), but
    the term size is restricted to \(n\)}
  & Section~\ref{sec:mainresult}\\
  \hline
  \(T\) & a tree & Section~\ref{sec:RegTreeGram}\\
  \hline
  \(|T|\) & the size of tree \(T\) & Section~\ref{sec:RegTreeGram}\\
  \hline
  \(C\) & a  tree context & Section~\ref{sec:RegTreeGram}\\
  \hline
  \(|C|\) & the size of context \(C\) & Section~\ref{sec:RegTreeGram}\\
  \hline
  \(C\subc C'\) & \(C\) is a subcontext of \(C'\) & Section~\ref{sec:RegTreeGram}\\
  \hline
  \(S\) & a \linear{context} & Section~\ref{sec:RegTreeGram}\\
  \hline
  \(U\) & an \affine{context} & Section~\ref{sec:RegTreeGram}\\
  \hline
  \(N\) & a nonterminal & Section~\ref{sec:RegTreeGram}\\
  \hline
  \(a\) & a terminal symbol & Section~\ref{sec:RegTreeGram}\\
  \hline
  \(\grm\) & a grammar & Section~\ref{sec:RegTreeGram}\\
  \hline
  \(\f\) & a context type & Section~\ref{sec:RegTreeGram}\\
  \hline
  \(\inhav{\grm,N}\) & the set of trees generated by \(N\) & Section~\ref{sec:RegTreeGram}\\
  \hline
    \(\inhavn{\grm,N}\)&
  \mline{Like \(\inhav{\grm,N}\), but
    the tree size is restricted to \(n\)}
  & Section~\ref{sec:RegTreeGram}\\
  \hline
  \(\inhav{\grm,\f}\) &  \mline{ the set of contexts \(C\)
    such that \(N\rew^* C[N_1,\ldots,N_k]\)
where \(\f=N_1\cdots N_k\Ar N\)
  }  & Section~\ref{sec:RegTreeGram}\\
  \hline
    \(\inhavn{\grm,\f}\)&
  \mline{Like \(\inhav{\grm,\f}\), but
    the context size is restricted to \(n\)}
  & Section~\ref{sec:RegTreeGram}\\
\hline
\end{tabular}

\begin{tabular}{|l|l|l|}
  \hline
  notation & explanation & introduced at:\\
  \hline
    \(\inhav{\grm}\)&
  \mline{the set of trees generated from some nonterminals}
  & Section~\ref{sec:RegTreeGram}\\
  \hline
    \(\inhavn{\grm}\)&
  \mline{Like \(\inhavn{\grm}\), but the tree size is restricted to \(n\)}
  & Section~\ref{sec:RegTreeGram}\\
  \hline
    \(\ctx{\grm}\)&
  \mline{the set of \linear{contexts} generated from some nonterminals}
  & Section~\ref{sec:RegTreeGram}\\
  \hline
    \(\ctxn{\grm}\)&
  \mline{Like \(\ctxn{\grm}\), but the size is restricted to \(n\)}
  & Section~\ref{sec:RegTreeGram}\\
  \hline
    \(\octx{\grm}\)&
  \mline{the set of \affine{contexts} generated from some nonterminals}
  & Section~\ref{sec:RegTreeGram}\\
  \hline
    \(\octxn{\grm}\)&
  \mline{Like \(\octxn{\grm}\), but the size is restricted to \(n\)}
  & Section~\ref{sec:RegTreeGram}\\
  \hline
  $E$ & a second-order context & Section~\ref{sec:decomposition}\\
  \hline
  \(\shn(E)\) & the number of second-order holes of \(E\) &
		 Section~\ref{sec:decomposition}\\
 \hline
  $P$ & a sequences of affine contexts & Section~\ref{sec:decomposition}\\
    \hline
  $\len{P}$ & the length of the sequence \(P\) of affine contexts & Section~\ref{sec:decomposition}\\
    \hline
    \(r\) & the largest arity of symbols in \(\ralph\) & Section~\ref{sec:gindependent-decomposition}
    \\
    \hline
  $\closeddecomp$ & \mline{grammar-independent decomposition function}
  & Section~\ref{sec:gindependent-decomposition}\\
  \hline
  $\decomp$ & an auxiliary decomposition function
  & Section~\ref{sec:gindependent-decomposition}\\
  \hline
  \(\uframe\) & \mline{the set of second-order contexts obtained by decomposing trees of size \(n\)}  
  & Section~\ref{sec:decomp-trees}\\
  \hline
  \(\ucomp\) & \mline{the set of affine contexts that match the hole
    \(\hhole_k^n\)} 
  & Section~\ref{sec:decomp-trees}\\
  \hline
  \(\useq\) & \mline{the set of sequences of affine contexts
    that match the second-order context \(E\)} 
  & Section~\ref{sec:decomp-trees}\\
  \hline
  \(\gctx\) & a typed second order context & Section~\ref{sec:rtl-decomposition}\\
  \hline
  \(\Hhole{\f}{n}\) & a second-order hole & Section~\ref{sec:rtl-decomposition}\\
  \hline
  \(\shn(\gctx)\) & the number of second-order holes in \(\gctx\)   & Section~\ref{sec:rtl-decomposition}\\
  \hline
  \(\sh{\gctx}{i}\) & the \(i\)-th second-order hole of \(\gctx\)
  & Section~\ref{sec:rtl-decomposition}\\
  \hline
  \(\closedgdecomp\)& grammar-respecting decomposition function for trees &
  Section~\ref{sec:grespecting-decomposition}\\
  \hline
  \(  \tframe(\grm, N)\)
&  typed-version of \(\uframe\) &
  Section~\ref{sec:grespecting-bijection}\\
  \hline
  \(  \tcomp\)
&   typed-version of \(\ucomp\) &
  Section~\ref{sec:grespecting-bijection}\\
  \hline
  \(  \tseq(\grm, N)\)
&   typed-version of \(\useq\) &
  Section~\ref{sec:grespecting-bijection}\\
  \hline
  \(\lralph{\D}{\I}{\K}\)
  & the set of terminals of the grammar for \(\lambda\)-terms
  & Section~\ref{sec:grammar-for-lambda}\\
  \hline
  \(\lnont{\D}{\I}{\K}\)
  &the set of nonterminals of the grammar for \(\lambda\)-terms
  & Section~\ref{sec:grammar-for-lambda}\\
  \hline
  \(\lrules{\D}{\I}{\K}\)
  &the set of rules of the grammar for \(\lambda\)-terms
  & Section~\ref{sec:grammar-for-lambda}\\
  \hline
  $\expl{m}{k}$ & an explosive term & Section~\ref{sec:explosive}\\
  \hline
\end{tabular}
 \section{Proof of Canonical Form Lemma}
\label{appx:canonical}
This section proves Lemma~\ref{lem:canonical-grammar} in Section~\ref{sec:canonical-grammar}.
We describe the required transformation as the composite of two transformations.
In the first step, we transform a grammar to \emph{semi-canonical} one:
\begin{defi}[Semi-Canonical Grammar]
  A rewriting rule is in \emph{semi-canonical form} if it is in canonical form or of the form
  \[
  N \rew_\grm N'.
  \]
  A grammar \( \grm = (\ralph, \nont, \rules) \) is \emph{semi-canonical} if every rewriting rule is in semi-canonical form.
\end{defi}

Lemma~\ref{lem:canonical-grammar} will be shown by the next two lemmas.
\begin{lem}\label{lem:making-semi-canonical}
  Let \( \grm = (\ralph, \nont, \rules) \) be a regular tree grammar that is unambiguous and essentially strongly-connected.
  Then one can (effectively) construct a grammar \( \grm' = (\ralph, \nont', \rules') \) with \( \nont \subseteq \nont' \) that satisfies the following conditions:
  \begin{itemize}
  \item \( \grm' \) is semi-canonical, unambiguous and essentially strongly-connected.
  \item \( \inhav{\grm,N} = \inhav{\grm',N} \) for every \( N \in \nont \).
  \item \( \ctxinf{\grm} \subseteq \ctxinf{\grm'} \).
  \end{itemize}
\end{lem}
\begin{proof}
  Let \( \grm = (\ralph, \nont, \rules) \) be a grammar that satisfies the assumption.
  We can assume without loss of generality that \( \inhav{\grm, N_0} \neq \emptyset \) for every \( N_0 \in \nont \).
  
  Assume \( \grm  \) has a rule of the form
  \[
  N' \rew a(\tr_1, \dots, \tr_n) \qquad\mbox{(\( \tr_1, \dots, \tr_n \) are \( (\ralph \cup \nont) \)-trees)}
  \]
  such that \( T_i \notin \nont \) for some \(i\).
  (If \( \grm \) does not have a rule of this form, then \( \grm \) is already semi-canonical.)
Let \( N'' \) be a fresh nonterminal not in \( \nont \).
  Consider the grammar \( \grm'' = (\ralph, \nont \cup \{ N'' \}, \rules'') \) where
\begin{align*}
  \rules'' =
(\rules \;\setminus&\; \set{\; N' \rew a(\tr_1, \dots, \tr_n) \;})
\\
  \cup&\; \{\; N' \rew a(\tr_1, \dots, \tr_{i-1}, N'', \tr_{i+1}, \dots, \tr_n), \quad N'' \rew \tr_i \;\}.
\end{align*}
  Then we can show the following claim by induction on the length of rewriting sequences.
  \begin{quote}
    \textbf{Claim}
    For every \( N \in \nont \) and \( (\ralph \cup \nont) \)-tree \( \tr \),
    \begin{itemize}
    \item 
      \( N \rew_{\grm}^* \tr \) if and only if \( N \rew_{\grm''}^* \tr \), and
    \item
      \( N \rew_{\grm,\ell}^* \tr \) if and only if \( N \rew_{\grm'',\ell}^* \tr \).
    \end{itemize}
  \end{quote}
  This claim implies that \( \grm'' \) is unambiguous and \( \inhav{\grm,N} = \inhav{\grm'',N} \) for every \( N \in \nont \).

  We show that \( \grm'' \) is essentially strongly-connected.
  Let \( N_1, N_2 \in \nont \cup \{ N'' \} \) and assume \( \card{\inhav{\grm'',N_1}} = \card{\inhav{\grm'',N_2}} = \infty \).
  We should show that \( N_2 \) is reachable from \( N_1 \).
  We can assume without loss of generality that \( N_1 \neq N'' \) and \( N_2 \neq N'' \) by the following argument:
  \begin{itemize}
  \item
    If \( N_1 = N'' \), then \( T_i \) must contain a nonterminal \( N_1' \in \nont \) with \( \card{\inhav{\grm'', N_1'}} = \infty \).
    Since \( N'' \rew_{\grm''} T_i \) is the unique rule for \( N'' \) and \( \card{\inhav{\grm'', N''}} = \infty \), we have \( \inhav{\grm, N_0} \neq \emptyset \) for each nonterminal \( N_0 \) appearing in \( T_i \).
    Hence \( N'' \rew_{\grm''} T_i \rew_{\grm''}^* S[N_1'] \) for some \linear{context} \( S \) over \( \ralph \) by rewriting each occurrence of a nonterminal \( N_0 \) in \( T_i \) (except for one occurrence of \( N_1' \)) with an element of \( \inhav{\grm'', N_0} \).
    So \( N_1' \) is reachable from \( N'' \) in \( \grm'' \).
    Since the reachability relation is transitive, it suffices to show that \( N_2 \) is reachable from \( N_1' \in \nont \), which satisfies \( \card{\inhav{\grm'', N_1'}} = \infty \).
  \item
    If \( N_2 = N'' \), then \( \inhav{\grm'', N'} \) is infinite and \( N'' \) is reachable from \( N' \).
    (Here we use the assumption that \( \inhav{\grm, N_0} \neq \emptyset \) for every \( N_0 \in \nont \).)
    Hence it suffices to show that \( N' \) is reachable from \( N_1 \).
  \end{itemize}
  Since \( \inhav{\grm'', N_i} = \inhav{\grm, N_i} \) (\( i = 1, 2 \)), both \( \inhav{\grm, N_1} \) and \( \inhav{\grm, N_2} \) are infinite.
  Since \( \grm \) is essentially strongly-connected, there exists a \( 1 \)-context \( S \) such that \( N_1 \rew_{\grm}^* S[N_2] \).
  Since \( S[N_2] \) is a \( (\ralph \cup \nont) \)-tree, by the above claim, \( N_1 \rew_{\grm''}^* S[N_2] \) and thus \( N_2 \) is reachable from \( N_1 \) in \( \grm'' \).

  We show that \( \ctxinf{\grm} \subseteq \ctxinf{\grm''} \).
  Suppose that \( S \in \inhav{\grm, N \Ar N'} \) with \( \card{\inhav{\grm, N}} = \card{\inhav{\grm, N'}} = \infty \).
  By definition, \( N' \rew_{\grm}^* S[N] \).
  By the above claim, \( N' \rew_{\grm''}^* S[N] \).
  Hence \( S \in \inhav{\grm'', N \Ar N'} \).
  Since \( \inhav{\grm, N} = \inhav{\grm'', N} \) and \( \inhav{\grm, N'} = \inhav{\grm'', N'} \), we have \( \card{\inhav{\grm'', N}} = \card{\inhav{\grm'', N'}} = \infty \).

    By applying the above transformation as much as needed, we obtain a grammar that satisfies the requirements.
    What remains is to show the termination of the iterated application of the above transformation.
    A termination measure can be given as follows: the measure of a rule \( N \rew C[N_1, \dots, N_k] \) is defined as \( \max(|C| - 1, 0) \) and the measure of a grammar is the sum of the measures of rules.
    The measure is by definition a non-negative integer and decreases by \( 1 \) by the above transformation.
\end{proof}

\begin{lem}\label{lem:making-canonical}
  Let \( \grm = (\ralph, \nont, \rules) \) be a regular tree grammar that is
  semi-canonical, unambiguous and essentially strongly-connected.
  Then one can (effectively) construct a grammar \( \grm' = (\ralph, \nont', \rules') \) and
  a family \( (\nset_N)_{N \in \nont} \) of subsets \( \nset_N \subseteq \nont' \) that satisfy the following conditions:
  \begin{itemize}
  \item \( \grm' \) is canonical, unambiguous and essentially strongly-connected.
  \item 
	\( \inhav{\grm,N} = \biguplus_{N' \in \nset_N} \inhav{\grm',N'} \) for every \(N \in \nont\).
  \item \( \ctxinf{\grm} \subseteq \ctxinf{\grm'} \).
  \end{itemize}
\end{lem}
\begin{proof}
  Let \( \grm = (\ralph, \nont, \rules) \) be a grammar that satisfies the assumption. 
  Let \( \nont' \) be the subset of \( \nont \) defined by
  \[
  N_0 \in \nont' \;\Leftrightarrow\; (N_0 \rew a(N_1, \dots, N_{\ralph(a)})) \in \rules \;\;\mbox{for some \( a, N_1, \dots, N_{\ralph(a)} \)}.
  \]
  The rewriting rule \( \rules' \) is defined by
  \begin{align*}
    \rules' := \{ N_0 \rew a(N_1', \dots, N_{\ralph(a)}') \mid~
    & N_0, N'_1, \dots, N'_{\ralph(a)} \in \nont', \\
    & N_0 \rew_{\grm} a(N_1, \dots, N_{\ralph(a)}) \rew_{\grm}^* a(N'_1, \dots, N'_{\ralph(a)}) \}.
  \end{align*}
  In other words, \( N_0 \rew a(N_1', \dots, N_{\ralph(a)}') \) is in \( \rules' \) if and only if there exists a rule \( N_0 \rew a(N_1, \dots, N_{\ralph(a)}) \) in \( \rules \) with \( N_i \rew_{\grm}^* N'_i \) for each \( i = 1, \dots, \ralph(a) \).
  Given a nonterminal \( N \in \nont \), let \( \nset_{N} \subseteq \nont' \) be the subset given by
  \[
  \nset_{N} \defeq \{ N' \in \nont' \mid N \rew_\grm^* N' \}.
  \]
  We show that the grammar \( \grm' \defeq (\ralph, \nont', \rules') \) together with \( (\nset_N)_{N \in \nont} \) satisfies the requirement.

  We can prove the following claims by induction on the length of rewriting sequences:
  \begin{enumerate}
  \item For \( N \in \nont \) and \( T \in \TR{\ralph \cup \nont'} \), if \( N \rew_{\grm}^* T \), then \( N' \rew_{\grm'}^* T \) for some \( N' \in \nset_{N} \).
  \item For \( N' \in \nont' \) and \( T \in \TR{\ralph \cup \nont'} \), if \( N' \rew_{\grm'}^* T \), then \( N' \rew_{\grm}^* T \).
  \end{enumerate}

  We show that \( \inhav{\grm,N} = \bigcup_{N' \in \nset_N} \inhav{\grm',N'} \).
  The above \((1)\) shows that \( \inhav{\grm,N} \subseteq \bigcup_{N' \in \nset_N} \inhav{\grm',N'} \).
  The above \((2)\) shows that \( \inhav{\grm',N'} \subseteq \inhav{\grm,N'} \) for every \( N' \in \nset_N \).
  By definition of \( \nset_N \), we have \( N \rew_{\grm}^* N' \) and thus \( \inhav{\grm,N'} \subseteq \inhav{\grm,N} \) for every \( N' \in \nset_N \).
  So \( \bigcup_{N' \in \nset_N} \inhav{\grm',N'} \subseteq \inhav{\grm,N} \).

  We prove by contraposition that, for \( N_1, N_2 \in \nset_N \), if \( N_1 \neq N_2 \), then \( \inhav{\grm', N_1} \cap \inhav{\grm', N_2} = \emptyset \).
  Suppose that \( T \in \inhav{\grm', N_1} \cap \inhav{\grm', N_2} \).
  Then we have
  \[
  N_i \rew_{\grm'} a(N'_{i,1}, \dots, N'_{i,\ralph(a)}) \rew_{\grm'}^* T   \qquad (i = 1, 2)
  \]
for some \(N'_{i,j} \in \nont'\).
  By the definition of \( \rules' \), we have
  \[
  N_i \rew_{\grm} a(N_{i,1}, \dots, N_{i,\ralph(a)}) \rew_{\grm}^* a(N'_{i,1}, \dots, N'_{i,\ralph(a)}) \rew_{\grm'}^* T   \qquad (i = 1, 2).
  \]
  Because \( N_i \in \nset_N \), we have \( N \rew_{\grm}^* N_i \), for \( i = 1, 2 \).
  By the claim \((2)\), we have
\[
  N \rew_{\grm}^* N_i \rew_{\grm} a(N_{i,1}, \dots, N_{i,\ralph(a)}) \rew_{\grm}^* a(N'_{i,1}, \dots,
 N'_{i,\ralph(a)}) \rew_{\grm}^* T   \qquad (i = 1, 2).
\]
  The above two rewriting sequences of \(\grm\) induce the following two leftmost rewriting sequences:
\[
  N \rew_{\grm,\ell}^* N_i \rew_{\grm,\ell} a(N_{i,1}, \dots, N_{i,\ralph(a)}) \rew_{\grm,\ell}^* T   \qquad (i = 1, 2),
\]
which are the same by the unambiguity of \(\grm\).
  Then the length of the sequence \(N \rew_{\grm,\ell}^* N_1\) and that of \(N \rew_{\grm,\ell}^* N_2\) are the same
  since just after them the root terminal \(a\) occurs, and hence \(N_1=N_2\), as required. 

  Next we prove the unambiguity of \(\grm'\), by contradiction:
  we assume that \( \grm' \) is not unambiguous,
  and show that \(\grm\) is not unambiguous, which is a contradiction.
  Let \( N_0' \in \nont' \) and \( T \in \TR{\ralph} \) be a pair of a nonterminal and a tree such that there exist two leftmost rewriting sequences from \( N_0' \) to \( T \) in \( \grm' \).
  We can assume without loss of generality that the leftmost rewriting sequences differ on the first step.
  Assume that the two rewriting sequences are
  \[
  N_0' \rew_{\grm', \ell} a(N_{1,1}', \dots, N_{1,\ralph(a)}') \rew_{\grm', \ell}^* a(T_1, \dots,
 T_{\ralph(a)}) = T
  \]
  and
  \[
  N_0' \rew_{\grm', \ell} a(N_{2,1}', \dots, N_{2,\ralph(a)}') \rew_{\grm', \ell}^* a(T_1, \dots,
 T_{\ralph(a)}) = T.
  \]
  We have \( N_{1,j}' \neq N_{2,j}' \) for some \( j \).
  By the definition of \( \rules' \), there are rules
  \[
  N_0' \rew_{\grm} a(N_{1,1}, \dots, N_{1,\ralph(a)}) \qquad\mbox{and}\qquad N_0' \rew_{\grm} a(N_{2,1}, \dots, N_{2,\ralph(a)})
  \]
  with \( N_{i,j} \rew_{\grm}^* N_{i,j}' \) for every \( (i,j) \in \{ 1, 2 \} \times \{ 1, \dots, \ralph(a) \} \).
  There are two cases:
  \begin{itemize}
  \item
    Case \( N_{1,j} \neq N_{2,j} \) for some \( j \in \{ 1, \dots, \ralph(a) \} \):
    By the rewriting sequences above and the claim \((2) \), we have two sequences
    \[
    N_0' \rew_{\grm} a(N_{i,1}, \dots, N_{i,\ralph(a)}) \rew_{\grm}^* a(N_{i,1}', \dots, N_{i,\ralph(a)}') \rew_{\grm}^* a(T_1, \dots, T_{\ralph(a)})
    \]
    for \( i = 1, 2 \).
    The corresponding leftmost rewriting sequences differ at the first step.
    So \( \grm \) is unambiguous.
  \item
    Case \( N_{1,j} = N_{2,j} \) for every \( j \in \{ 1, \dots, \ralph(a) \} \):
    Let \( j \) be an index such that \( N_{1,j}' \neq N_{2,j}' \) and \( N_j \defeq N_{1,j} = N_{2,j} \).
    Then \( N_j \rew_{\grm}^* N_{1,j}' \rew_{\grm'}^* T_j \) and \( N_j \rew_{\grm}^* N_{2,j}' \rew_{\grm'}^* T_j \).
    By the definition of \( \rules' \) and the claim \((2)\), we have
    \[
    N_j \rew_{\grm}^* N_{i,j}' \rew_{\grm} b(N_{i,1}'', \dots, N_{i,\ralph(b)}'') \rew_{\grm}^* T_j \qquad(i = 1, 2)
    \]
    with \( N_{1,j}' \neq N_{2,j}' \).
    Hence \( \grm \) is not unambiguous.
  \end{itemize}
  
  We show that \( \grm' \) is essentially strongly-connected.
  Let \( N_1, N_2 \in \nont' \) such that \( \card{\inhav{\grm', N_1}} = \card{\inhav{\grm', N_2}} = \infty \).
  We show that \( N_2 \) is reachable from \( N_1 \) in \(\grm'\).
  By the claim \((2)\), \( \inhav{\grm, N_2} \supseteq \inhav{\grm', N_2} \) is infinite.
  Since \( \inhav{\grm', N_1} \) is infinite, we have a rule
  \[
  N_1 \rew_{\grm'} a(N_{1,1}, \dots, N_{1, \ralph(a)})
  \]
  in \( \rules' \) such that \( \inhav{\grm',N_{1,i}} \) is infinite for some \( i \) and \( \inhav{\grm', N_{1,j}} \neq \emptyset \) for every \( j = 1, \dots, \ralph(a) \).
  Again, by the claim \((2)\), \( \inhav{\grm, N_{1,i}} \supseteq \inhav{\grm', N_{1,i}} \) is infinite.
  Since \( \grm \) is essentially strongly-connected, there exists a \linear{context} \( S \) such that \( N_{1,i} \rew_{\grm}^* S[N_2] \).
  By the claim \((1)\), there exists \( N_{1,i}' \in \nset_{N_{1,i}} \) such that \( N_{1,i}' \rew_{\grm'}^*
 S[N_2] \).
  Since \(N'_{1,i} \in \nset_{N_{1,i}}\), we have \(N_{1,i} \rew^*_{\grm} N'_{1,i}\), and hence
  by the definition of \( \rules' \), we find that
  \[
  N_1 \rew_{\grm'} a(N_{1,1}, \dots, N_{1, i-1}, N_{1,i}', N_{1, i+1}, \dots, N_{1, \ralph(a)})
  \]
  is in \( \rules' \).
  Let us pick \( T_{1,j} \in \inhav{\grm',N_{1,j}} \) for each \( j \neq i\); then we have:
  \begin{align*}
    N_1
    &\rew_{\grm'} a(N_{1,1}, \dots, N_{1, i-1}, N_{1,i}', N_{1, i+1}, \dots, N_{1, \ralph(a)}) \\
    &\rew_{\grm'}^* a(N_{1,1}, \dots, N_{1, i-1}, S[N_2], N_{1, i+1}, \dots, N_{1, \ralph(a)}) \\
    &\rew_{\grm'}^* a(T_{1,1}, \dots, T_{1, i-1}, S[N_2], T_{1, i+1}, \dots, T_{1, \ralph(a)}).
  \end{align*}
  Thus \(N_2\) is reachable from \(N_1\) in \(\grm'\).

  Lastly, we show that \( \ctxinf{\grm} \subseteq \ctxinf{\grm'} \).
  Assume that \( S \in \inhav{\grm, N_1 \Ar N_2} \) and \( \card{\inhav{\grm, N_1}} = \card{\inhav{\grm, N_2}} = \infty \).
  Since \( \inhav{\grm, N_1} = \bigcup_{N_1' \in \nset_{N_1}} \inhav{\grm', N_1'} \), there exists \( N_1' \in \nset_{N_1} \) such that \( \card{\inhav{\grm', N_1'}} = \infty \).
  Since \( N_2 \rew_{\grm}^* S[N_1] \) and \( N_1 \rew_{\grm}^* N_1' \), we have \( N_2 \rew_{\grm}^* S[N_1'] \).
  By the claim \((1)\), there exists \( N_2' \in \nset_{N_2} \) such that \( N_2' \rew_{\grm'}^* S[N_1'] \).
  Hence \( S \in \inhav{\grm', N_1' \Ar N_2'} \).
  Since \( \card{\inhav{\grm', N_1'}} = \infty \), we have \( \card{\inhav{\grm', N_2'}} = \infty \) as well.
\end{proof}

\begin{proof}[Proof of Lemma~\ref{lem:canonical-grammar}]
  Let \( \grm \) be a regular tree grammar that is unambiguous and strongly connected.
  Then \( \grm \) is essentially strongly-connected.
  So by applying Lemmas~\ref{lem:making-semi-canonical} and \ref{lem:making-canonical}, we obtain a grammar \( \grm' \) and a family \( (\nset_N)_{N \in \nont} \) that satisfy the following conditions:
  \begin{itemize}
  \item \( \grm' \) is canonical, unambiguous and essentially strongly-connected.
  \item 
\( \inhav{\grm,N} = \biguplus_{N' \in \nset_N} \inhav{\grm',N'} \) for every \(N \in \nont\).
  \item \( \ctxinf{\grm} \subseteq \ctxinf{\grm'} \).
  \end{itemize}
  The grammar \( \grm' \) and a family \( (\nset_N)_{N \in \nont} \) satisfy the first and second requirements of Lemma~\ref{lem:canonical-grammar}.
  We show that \( \card{\inhav{\grm}} = \infty \) implies \( \ctx{\grm} \subseteq \ctxinf{\grm'} \).

  Suppose that \( \card{\inhav{\grm}} = \infty \).
  Since \( \inhav{\grm} = \bigcup_{N \in \nont} \inhav{\grm, N} \) and \( \nont \) is finite, there exists \( N_0 \in \nont \) such that \( \card{\inhav{\grm,N_0}} = \infty \).
  By strong connectivity of \( \grm \), we have \( \card{\inhav{\grm,N}} = \infty \) for every \( N \in \nont \).
  Hence \( \nontinf = \nont \), and thus \( \ctx{\grm} = \ctxinf{\grm} \).
  We have \( \ctx{\grm} = \ctxinf{\grm} \subseteq \ctxinf{\grm'} \) as required.
\end{proof}

 \section{Proof of Lemma~\ref{lem:period-grammar2}}
\label{sec:periodicity}
We start from the formal definition of periodicity and its basic properties.
\begin{defi}[Period, Basic Period]\label{def:period}
  Let \( X \subseteq \nat \) be a subset of natural numbers.
  We say \( X \) is \emph{periodic} with \emph{period \( c \)} (\( c \in \nat \backslash \{0\} \)) if
  \[
  \exists n_0 \in \nat. \;\forall n \ge n_0. \;\forall k \in \nat. \quad n \in X \Leftrightarrow (n + ck) \in X.
  \]
  If \( X \) is periodic, we call the minimum period the \emph{basic period}.
  
  Let \( \grm \) be a regular tree grammar and \( \f = (N_1 \dots N_k \Ar N)\) be a context type of \( \grm \).
  We say \( \f \) is \emph{periodic with period \( c \)} if so is the subset of natural numbers
  \[
  \{ n \in \nat \mid \inhavn{\grm, \f} \neq \emptyset \},
  \]
  and call the minimum period the \emph{basic period}.
\end{defi}

\begin{lem}\label{lem:period-lem1}
  Let \( X \subseteq \nat \) be a subset of natural numbers and \( c \in \nat \backslash \{ 0 \} \).
  \begin{enumerate}
  \item\label{item:period-lem1-1}
    Suppose that
    \[
    \exists n_0 \in \nat. \;\forall n \ge n_0. \quad n \in X \Rightarrow n+c \in X.
    \]
    Then \( X \) is periodic with period \( c \).
  \item\label{item:period-lem1-k}
    Suppose that
    \[
    \exists n_0 \in \nat. \;\exists k_0 \in \nat. \;\forall n \ge n_0. \;\forall k \ge k_0. \quad n \in X
       \Rightarrow n + ck \in X.
    \]
    Then \( X \) is periodic with period \( c \).
  \item\label{item:period-lem1-gcd}
    Suppose that \( c \) and \( c' \) are periods of \( X \).
    Then \( \gcd(c, c') \) is also a period of \( X \), where \( \gcd(c, c') \) is the greatest common divisor.
  \end{enumerate}
\end{lem}
\begin{proof}
  \eqref{item:period-lem1-1}
  Let \( n_0 \in \nat \) be a witness of the assumption.
  Then for every \( n \ge n_0 \) and \( k \ge 0 \), we have \( n \in X \Rightarrow (n + ck) \in X \), by induction on \( k \).
  Let \( Y \subseteq \nat \) be the set of counterexamples of \( n + c \in X \Rightarrow n \in X \), namely,
  \[
  Y \defeq \{ n \in \nat \mid n + c \in X \;\wedge\; n \notin X \}.
  \]
  Then \( Y \) has at most \( n_0 + c \) elements because, for every \( n \ge n_0 \) and \( k \ge 1 \),
  \[
  n \in Y \quad\Rightarrow\quad n + c \in X \quad\Rightarrow\quad n + ck \in X \quad\Rightarrow\quad n + ck \notin Y.
  \]
  Let \( n_0' \ge n_0 \) be a natural number that is (strictly) greater than all elements of \( Y \).
  Then for every \( n \ge n_0' \),
  \[
  n \in X \quad\Leftrightarrow\quad n + c \in X.
  \]
  By induction on \( k \ge 0 \), we conclude that \( n \in X \Leftrightarrow n + kd \in X \) for every \( n \ge n_0' \).
  
  \eqref{item:period-lem1-k}
  Let \( n_0 \) and \( k_0 \) be witnesses of the assumption and \( Y \subseteq \nat \) be the set of counterexamples of \( n \in X \Rightarrow n+c \in X \), i.e.,
  \[
  Y \defeq \{ n \in \nat \mid n \in X \;\wedge\; n + c \notin X \}.
  \]
  Then \( Y \) has at most \( n_0 + c k_0 \) elements because, for every \( n \ge n_0 \) and \( k \ge k_0 \),
  \[
  n \in Y \quad\Rightarrow\quad n \in X \quad\Rightarrow\quad n + c(k+1) \in X \quad\Rightarrow\quad n + ck \notin Y.
  \]
  Let \( n_0' \) be a natural number that is (strictly) greater than all elements of \( Y \).
  Then for every \( n \ge n_0' \),
  \[
  n \in X \quad\Rightarrow\quad n + c \in X.
  \]
  Hence by item \eqref{item:period-lem1-1}, \( X \) is periodic with period \( c \).

  \eqref{item:period-lem1-gcd}
  The claim trivially holds if \( \gcd(c, c') = c \) or \( c' \).
  We assume that \( \gcd(c, c') < c, c' \).
  Let \( d = \gcd(c, c') \).
  There exist integers \( a, b \) such that \( a c + b c' = d \) (by B\'ezout's lemma).
  Since \( 0 < d < c, c' \), exactly one of \( a \) and \( b \) is negative.
  Suppose that \( a > 0 \) and \( b < 0 \).

  Let \( n_0 \) be a constant such that
  \[
  \forall n \ge n_0. \;\forall k \in \nat. \quad n \in X \Leftrightarrow (n + ck) \in X
  \]
  and \( n_0' \) be a constant for a similar condition for \( c' \).
  Let \( n_0'' = \max(n_0, n_0') - b c' \).
  Then, for every \( n \ge n_0'' \),
  \begin{align*}
    n \in X
    &\quad\Leftrightarrow\quad n + bc' \in X & & \mbox{(since \( n + bc' \ge n_0' \) and \( b < 0 \))} \\
    &\quad\Leftrightarrow\quad n + bc' + ac \in X & & \mbox{(since \( n + bc' \ge n_0 \) and \( a > 0 \))} \\
    &\quad\Leftrightarrow\quad n + d \in X.
  \end{align*}
  By induction on \( k \ge 0 \), we conclude that \( n \in X \Leftrightarrow n + kd \in X \) for every \( n \ge n_0'' \).
\end{proof}

We show that \( N \Ar N' \) is periodic for every \( N, N' \in \nontinf \) and its basic period is independent of the choice of \( N \) and \( N' \).
We call the basic period of \( N \Ar N' \), which is independent of \( N \) and \( N' \), the \emph{basic period of the grammar \( \grm \)}.
Further we show that \( U, U' \in \inhav{\grm, N \Ar N'} \) implies \( |U| \equiv |U'| \mod c \), where \( c \) is the basic period of the grammar \( \grm \).
\begin{lem}\label{lem:period-grammar1}
  Let \( \grm = (\ralph, \nont, \rules) \) be a regular tree grammar.
  Assume that \( \grm \) is essentially strongly-connected and \( \card{\inhav{\grm}} = \infty \).
  \begin{enumerate}
  \item\label{item:period-grammar1-1}
  \( \inhav{\grm, N \Ar N'} \) is infinite for every \( N, N' \in \nontinf \).
  \item\label{item:period-grammar1-2}
  \( N \Ar N \) is periodic for every \( N \in \nontinf \).
  \item\label{item:period-grammar1-3}
  Let \( N \in \nontinf \) and \( c_N \) be the basic period of \( N \Ar N \).
  Then \( \inhavn{\grm, N \Ar N} \neq \emptyset \) implies \( n \equiv 0 \mod c_N \).
  \item\label{item:period-grammar1-4}
  The basic period \( c_N \) of \( N \Ar N \) is independent of \( N \in \nontinf \).  Let \( c = c_N \).
  \item\label{item:period-grammar1-5}
  For every \( N, N' \in \nontinf \), if \( U_1, U_2 \in \inhav{\grm, N \Ar N'} \), then \( |U_1| \equiv |U_2| \mod c \).
  \item\label{item:period-grammar1-6}
  For every \( N, N' \in \nontinf \), \( N \Ar N' \) has the basic period \( c \).
  \end{enumerate}
\end{lem}
\begin{proof}
  \eqref{item:period-grammar1-1}
  There exist \( N_1, N_2 \in \nontinf \) such that \( \inhav{\grm, N_1 \Ar N_2} \) contains a context of size
 greater than \( 0 \), because there exists \(N_2 \in \nont\) with \( \inhav{\grm, N_2} = \infty \) due to \( \card{\inhav{\grm}} = \infty \).
  Let \( U \in \inhav{\grm, N_{1} \Ar N_{2}} \) be a context with \( |U| > 0 \).
  Let \( N \in \nontinf \).
  By essential strong-connectivity there exist \( S_{N} \in \inhav{\grm, N \Ar N_{1}} \) and \( S'_{N} \in
 \inhav{\grm, N_{2} \Ar N} \), and then \( U_{N} \defeq S'_{N}[U[S_{N}]] \in \inhav{\grm, N \Ar N} \) with \( |U_{N}| > 0 \).
  Let \( U^n_N \) be the context defined by \( U^0_N = \hole \) and \( U^{n+1}_N = U_N[U_N^n] \).
  Then \( U^n_N \in \inhav{\grm, N \Ar N} \) for every \( n \) and \( |U^n| \ge n \); hence \( \inhav{\grm, N \Ar N} \) is an infinite set.
  Let \( N' \in \nontinf \).
  By essential strong-connectivity there exists \( S_{N,N'} \in \inhav{\grm, N \Ar N'} \), which implies \( \inhav{\grm, N \Ar N'} \) is also infinite.

  \eqref{item:period-grammar1-2}
  Let \( N \in \nontinf \).
  By \eqref{item:period-grammar1-1}, \( \inhav{\grm, N \Ar N} \) is an infinite set.
  In particular, there exists \( U_N \in \inhav{\grm, N \Ar N} \) with \( |U_N| > 0 \).
  Let \( n \in \nat \) and assume that \( \inhavn{\grm, N \Ar N} \neq \emptyset \).
  Then \( U \in \inhavn{\grm, N \Ar N} \) for some \( U \) and thus \( U[U_N] \in \inhavnn{n + |U_N|}{\grm, N \Ar N} \).
  Therefore \( \inhavn{\grm, N \Ar N} \neq \emptyset \) implies \( \inhavnn{n + |U_N|}{\grm, N \Ar N} \neq \emptyset \).
  By applying Lemma~\ref{lem:period-lem1}\eqref{item:period-lem1-1}, \( N \Ar N \) is periodic with period \( |U_N| \).

  \eqref{item:period-grammar1-3}
  Let \( N \in \nontinf \) and \( c_N \) be the basic period of \( N \Ar N \),
  and suppose that \( \inhavn{\grm, N \Ar N} \neq \emptyset \).
  If \( n = 0 \), then the claim trivially holds.
  Suppose that \( n > 0 \) and let \( U \in \inhavn{\grm, N \Ar N} \).
  Then by the argument of the proof of \eqref{item:period-grammar1-1}, \( N \Ar N \) has period \( |U| = n \).
  By Lemma~\ref{lem:period-lem1}\eqref{item:period-lem1-gcd}, \( N \Ar N \) has period \( \gcd(c_N, n) \), which is smaller than or equal to \( c_N \).
  Since \( c_N \) is basic, this implies \( c_N = \gcd(c_N, n) \), i.e.,~\( n \equiv 0 \mod c_N \).

  \eqref{item:period-grammar1-4}
  Let \( N_1, N_2 \in \nontinf \) and \( c_1 \) and \( c_2 \) be the basic periods of \( N_1 \Ar N_1 \) and of \( N_2 \Ar N_2 \), respectively.
  We show that \( c_1 \equiv 0 \mod c_2 \).
  By essential strong-connectivity, there exist \( S_{1,2} \in \inhav{\grm, N_1 \Ar N_2} \) and \( S_{2,1} \in \inhav{\grm, N_2 \Ar N_1} \).
  Since \( \inhav{\grm, N_1 \Ar N_1} \) is infinite (by \eqref{item:period-grammar1-1}) and \( N_1 \Ar N_1 \) has period \( c_1 \), there exist contexts \( U, U' \in \inhav{\grm, N_1 \Ar N_1} \) such that \( |U'| = |U| + c_1 \).
  Let \( n = |S_{1,2}[U[S_{2,1}]]| \) and \( n' = |S_{1,2}[U'[S_{2,1}]]| \); then \( n' = n + c_1 \).
  Then both \( \inhavnn{n}{\grm, N_2 \Ar N_2} \) and \( \inhavnn{n'}{\grm, N_2 \Ar N_2} \) are non-empty.
  Hence \( n \equiv 0 \mod c_2 \) and \( n' \equiv 0 \mod c_2 \) by \eqref{item:period-grammar1-3}.
  Therefore \( c_1 \equiv n' - n \equiv 0 \mod c_2 \).

  \eqref{item:period-grammar1-5}
  Let \( N, N' \in \nontinf \) and assume that \( U_1, U_2 \in \inhav{\grm, N \Ar N'} \).
  By essential strong-connectivity, there exists \( S \in \inhav{\grm, N' \Ar N} \).
  Then \( S[U_1], S[U_2] \in \inhav{\grm, N \Ar N} \).
  By \eqref{item:period-grammar1-3},
  \[
  |S| + |U_1| \equiv |S[U_1]| \equiv |S[U_2]| \equiv |S| + |U_2| \mod c.
  \]
  Hence \( |U_1| \equiv |U_2| \mod c \).
  
  \eqref{item:period-grammar1-6}
  Let \( N, N' \in \nontinf \).
  Recall that \( c \) is the basic period of \( N \Ar N \).
  Since \( c \) is a period of \( N \Ar N \) and by \eqref{item:period-grammar1-3}, there exists \( k_0 \ge 0 \) such that, for every \( k \ge k_0 \), we have \( \inhavnn{ck}{\grm, N \Ar N} \neq \emptyset \).
  Let \( (U_k)_{k \ge k_0} \) be a family such that \( U_k \in \inhavnn{ck}{\grm, N \Ar N} \).
  Now \( \inhavnn{n}{\grm, N \Ar N'} \neq \emptyset \) implies \( \inhavnn{n + ck}{\grm, N \Ar N'} \neq \emptyset \) for every \( k \ge k_0 \), since \( U' \in \inhavnn{n}{\grm, N \Ar N'} \) implies \( U'[U_{k}] \in \inhavnn{n+ck}{\grm, N \Ar N'} \).
  By Lemma~\ref{lem:period-lem1}\eqref{item:period-lem1-k}, \( N \Ar N' \) has period \( c \).

  Conversely assume that \( N \Ar N' \) has period \( c' \).
  Since \( \inhav{\grm, N \Ar N'} \) is an infinite set by \eqref{item:period-grammar1-1}, we have \( U_1, U_2 \in \inhav{\grm, N \Ar N'} \) such that \( |U_2| = |U_1| + c' \).
  By \eqref{item:period-grammar1-5}, we have \( |U_1| \equiv |U_2| \mod c \), which implies \( c' \equiv 0 \mod c \).
\end{proof}

 We are now ready to prove Lemma~\ref{lem:period-grammar2}.
\begin{proof}[Proof of Lemma~\ref{lem:period-grammar2}]
  Let \( c \) be the basic period of \( \grm \). 
  We first define \( d_{N,N'} \) for each \( N, N' \in \nontinf \).
  Since \( \inhav{\grm, N \Ar N'} \) is infinite (Lemma~\ref{lem:period-grammar1}\eqref{item:period-grammar1-1}), there exists \( U_{N,N'} \in \inhav{\grm, N \Ar N'} \).
  Let \( d_{N,N'} \) be the unique natural number \( 0 \le d_{N,N'} < c \) such that \( |U_{N,N'}| \equiv d_{N,N'} \mod c \).
  This definition of \(d_{N,N'}\) does not depend on the choice of \(U_{N,N'}\), by Lemma~\ref{lem:period-grammar1}\eqref{item:period-grammar1-5}.

  We define \( n_0 \).
  For every \( N, N' \in \nontinf \), since \( c \) is a period of \( \inhav{\grm, N \Ar N'} \), 
  there exists \(n_0^{N,N'}\) such that:
  \[
  \forall n \ge n_0^{N,N'}. \;\forall k \in \nat. \quad 
  \inhavn{\grm, N \Ar N'} \neq \emptyset \ \Leftrightarrow\  \inhavnn{n + ck}{\grm, N \Ar N'} \neq \emptyset.
  \]
  Let \( n_0 \defeq \max \{ n_0^{N,N'} \mid N, N' \in \nontinf \} \).

  Then:
  \begin{itemize}
  \item \eqref{item:period-grammar2-every} follows from Lemma~\ref{lem:period-grammar1}\eqref{item:period-grammar1-5}.
  \item \eqref{item:period-grammar2-suff} follows from the definition of \( n_0 \), infinity of \( \inhav{\grm, N \Ar N'} \)
	(Lemma~\ref{lem:period-grammar1}\eqref{item:period-grammar1-1}), and Lemma~\ref{lem:period-grammar1}\eqref{item:period-grammar1-5}
	(given \(n \ge n_0\), pick \(U \in \inhav{\grm,N \Ar N'}\) such that \(|U| \ge n\); then \(|U| = n+ck\) for
	some \(k \ge 0\)).
  \item \eqref{item:period-grammar2-id} follows from Lemma~\ref{lem:period-grammar1}\eqref{item:period-grammar1-3}.
  \item \eqref{item:period-grammar2-comp} follows from \( U_{N',N''}[U_{N,N'}] \in \inhav{\grm, N \Ar N''} \) and
	Lemma~\ref{lem:period-grammar1}\eqref{item:period-grammar1-5}.
\qedhere
  \end{itemize}
\end{proof}
 
\end{document}